\documentclass[12pt, reqno]{amsart}

\makeatletter
\g@addto@macro{\endabstract}{\@setabstract}
\makeatother

\makeatletter
\def\@seccntformat#1{%
	\ifstrequal{#1}{subsection}%
	{\bfseries\csname the#1\endcsname.}%
	{\csname the#1\endcsname.}%
}
\makeatother

\usepackage{etoolbox}
\usepackage{amsmath, amssymb, amsthm}
\usepackage{amsfonts}
\usepackage{graphicx}
\usepackage[round,authoryear]{natbib}
\usepackage{xcolor}
\usepackage[citecolor=purple, colorlinks=true, linkcolor=blue, urlcolor=blue]{hyperref}
\usepackage{mathrsfs}
\usepackage[left=1.25in, right=1.25in, top=1.0in, bottom=1.15in, includehead, includefoot]{geometry}
\usepackage{enumitem}
\usepackage{booktabs}
\usepackage{multirow}
\usepackage{makecell}
\usepackage{adjustbox}
\usepackage{lscape}
\usepackage{float}
\usepackage{caption}
\usepackage{subcaption}
\usepackage[flushleft]{threeparttable}
\usepackage{bm}
\usepackage{tikz}
\usepackage{pgfplots}
\pgfplotsset{compat=1.15}

\renewcommand{\epsilon}{\varepsilon}

\theoremstyle{plain}
\newtheorem{theorem}{Theorem}[section]
\newtheorem{corollary}{Corollary}[section]
\newtheorem{lemma}{Lemma}[section]
\newtheorem{proposition}{Proposition}[section]

\theoremstyle{definition}

\newtheorem{example}{Example}[section]

\newtheorem{assumption}{Assumption}[section]

\usepackage{manyfoot}
\DeclareNewFootnote{default}
\DeclareNewFootnote{symbol}[fnsymbol]

\makeatletter
\renewcommand{\@makefnmark}{%
	\hbox{\@textsuperscript{\normalfont\footnotesize\@thefnmark}}%
}
\makeatother

\begin{document}

\title{}

\date{\today}

\begin{center}
	\LARGE
	From Centrality Discounts to Centrality Premia:\\
	Interoperability and Platform Competition in Social Networks


	\vspace{1em}
	\normalsize
	Weiming Li \quad
	Jing Sun \quad
	Xinxi Song \quad
	Bin Wu\footnote{Corresponding author.
	Email address: \texttt{bwu.econ@gmail.com}}

	\vspace{0.6em}
	International School of Economics and Management,\\
	Capital University of Economics and Business

	\par

	\vspace{1em}
	\normalsize{\today}
\end{center}

\begin{abstract}
We study how interoperability reshapes competitive price discrimination when consumers are embedded in a social network. Two differentiated platforms set personalized prices; consumers benefit from neighbors' consumption of the same platform and, under interoperability, of the rival. Equilibrium prices obtain in closed form for arbitrary networks and contain a network-position term, proportional to Katz–Bonacich centrality, whose sign is determined by whether interoperability exceeds product substitutability. Below this threshold, platforms contest central consumers and grant centrality discounts; above it, central consumers become gateways to a shared cross-platform network and pay premia; at the threshold, prices are independent of network position. Interoperability softens price competition, can make platforms favor denser consumer networks, and reverses which side of the market gains from price discrimination.\\
	\vspace{1em}
	\noindent
	\textit{JEL Classifications:} D43; D85; L13; L14; L15; L86 \\
	\textit{Keywords:} Social networks; platform interoperability; competitive pricing; price discrimination; network externalities; Katz-Bonacich centrality.
\end{abstract}

\section{Introduction}

On March 7, 2024, the European Union's Digital Markets Act became binding on the platforms it designates as gatekeepers---and with it a requirement that reshapes how digital platforms compete for users. Under Article 7, Meta must make the core functionalities of WhatsApp and Facebook Messenger interoperable with rival messaging services; in a rollout that began in late 2025, European WhatsApp users who opt in can exchange text messages, images, voice messages, videos, and files with users of the first third-party apps that have chosen to interoperate. The mandate targets a familiar source of market power. A messaging app is valuable to a consumer not for its global user count but because her own contacts are reachable on it. That is what makes a dominant network hard to dislodge: a contact adds value only if she is on the same app. Interoperability severs that tie. Once messages cross platform boundaries, a contact generates value for a consumer regardless of which app the contact uses.

This change does more than lower switching costs. It alters the strategic value of individual consumers---and therefore the prices platforms set for them. Because consumers are embedded in a social network, they are not interchangeable to a platform. A well-connected consumer is worth more than a peripheral one, not only for her own demand but because attracting her helps a platform reach the consumers around her. Platforms compete fiercely for such consumers, and much of what the industry calls growth marketing is in effect position-contingent pricing: referral bonuses pay users to bring in their contacts, influencer and ambassador programs subsidize high-degree nodes, and seeding campaigns target local hubs. The underlying logic is business stealing: winning an influential consumer tilts her neighbors toward one platform and away from its rival. Interoperability acts directly on this logic, because it determines whether the neighbors a hub delivers remain exclusive to the platform that wins her.

The DMA is one instance of a broader movement. Instant-payment systems increasingly let an account with one provider pay an account with another; cross-play lets gamers on rival consoles share the same match; federated social networks let a post on one service reach followers on another. Across these settings the same parameter varies---the degree to which network value spills across platform boundaries---from fully closed ecosystems, $\theta=0$, to full interoperability, $\theta=1$. This paper asks how that parameter reshapes competitive pricing in social networks: as platform boundaries open, what happens to the relationship between a consumer's network position and the price she pays?

We study two differentiated platforms selling to consumers connected by a fixed social network. The platforms observe the network and simultaneously set a personalized price for every consumer; consumers then choose continuous quantities of both services, and may patronize the two at once. The degree to which one service crowds out the marginal value of the other is measured by a substitutability parameter $\beta$. A consumer's value from a platform rises with her neighbors' consumption of the same platform---the familiar within-platform network externality---and, at a fraction $\theta$ of that rate, with her neighbors' consumption of the rival---the cross-platform externality that interoperability creates. We characterize the unique symmetric equilibrium of the pricing game in closed form for arbitrary networks. When $\theta=0$, the model collapses to the competitive network-pricing framework of  \citet{ChenZenouZhou2018RAND}, so any departure from its predictions is attributable to interoperability alone.

The model delivers a sharp answer. Equilibrium prices decompose into three parts: a monopoly benchmark, a competitive markdown driven by product substitutability, and a network-position term proportional to the network-generated component of each consumer's weighted Katz--Bonacich centrality. The sign of this last term---the paper's central result---turns on a single comparison between interoperability $\theta$ and substitutability $\beta$. When $\theta<\beta$, the standard logic prevails: a central consumer is a contested, exclusive asset---winning her delivers her neighbors to one platform and denies them to the other---so competition drives firms to subsidize centrality, and the well-connected receive discounts. Interoperability dissolves this exclusivity. As $\theta$ rises, a hub's neighbors generate value for her no matter which platform they patronize; capturing the hub therefore no longer locks in her surroundings, and the business-stealing motive behind centrality discounts fades. At $\theta=\beta$, the two forces cancel exactly---for every network architecture and every admissible strength of network effects---and equilibrium prices cease to depend on network position altogether. Once $\theta>\beta$, the sign flips: central consumers, now high-value gateways into a large cross-platform neighborhood rather than prizes to be fought over, command a premium instead of a discount.

The intuition is that interoperability changes the ownership of network value. In a closed ecosystem, the network benefit created by attracting a central consumer is firm-specific: her presence raises her neighbors' willingness to pay for the same platform, and for that platform only---which is precisely why platforms fight for her. Interoperability lets part of this value spill across the platform boundary, so winning her no longer confers an exclusive advantage. At the same time, the central consumer herself becomes more valuable to serve: she now draws access value from her entire neighborhood rather than from the subset that shares her platform, and this network-inflated willingness to pay accrues to her whichever platform she patronizes---so it can be priced rather than fought over. The business-stealing force scales with substitutability; the access-value force scales with interoperability. When the latter dominates---exactly when $\theta$ exceeds $\beta$---firms no longer subsidize centrality; they monetize it.

We then turn to welfare. Interoperability creates network value directly, by letting valuable interactions cross platform boundaries, but it also softens price competition, because network value is no longer exclusive to either platform; its net effect on consumer surplus and total welfare is therefore ambiguous in general. Regular networks, which shut down centrality heterogeneity, isolate the competitive channel: there, interoperability raises equilibrium prices and platform profits, while its effect on consumer surplus depends on product substitutability and the strength of network effects. Interoperability thus embodies a genuine trade-off---it enlarges the surplus the network generates while letting platforms capture more of it---rather than an unambiguous competitive gain.

Next, we examine how interoperability affects platforms' incentives toward the consumer network itself. In  \citet{ChenZenouZhou2018RAND}, denser networks can reduce profits: when products are sufficiently close substitutes, stronger network effects intensify the contest for influential consumers and the rents are competed away. Interoperability mitigates this competitive cost. When cross-platform externalities are sufficiently strong, platforms prefer denser networks even under parameter values for which, absent interoperability, they would prefer sparser ones. Interoperability can therefore convert the consumer network from a source of competitive pressure into a source of shared, monetizable value---with corresponding implications for whether platforms encourage or suppress social connectivity among their users.

Finally, we compare uniform and discriminatory pricing. In non-regular networks, we show that the welfare and distributional consequences of price discrimination are governed by the same comparison between $\theta$ and $\beta$. When $\theta<\beta$, discrimination plays out as competing centrality discounts: consumer surplus rises and platforms would jointly prefer to commit to uniform prices---the familiar prisoner's-dilemma logic of competitive price discrimination. When $\theta>\beta$, the roles reverse: discrimination becomes an instrument for extracting centrality premia, so platforms gain from it while consumers would prefer it banned. The threshold that fixes the sign of centrality pricing thus also determines which side of the market gains from the ability to price on network position at all.

Taken together, the results carry a distributional message for interoperability policy. Mandates such as the DMA's Article 7 are usually debated in terms of entry, tipping, and the intensity of platform competition. Our analysis shows that they also determine which consumers pay for access: interoperability compresses and eventually inverts the centrality discounts of closed ecosystems, shifting surplus between the well-connected and the peripheral---and it can soften the very price competition it is meant to invigorate.

The paper makes three contributions. First, it extends models of competitive pricing in social networks by introducing cross-platform local network externalities, while preserving full tractability: equilibrium prices obtain in closed form on any network. Second, it identifies a centrality-reversal mechanism---interoperability can overturn the classic centrality discount and generate centrality premia---and locates the exact threshold, $\theta=\beta$, at which the sign flips. Third, it traces the consequences of this consumer-level mechanism for welfare, for platforms' incentives toward the formation of the network itself, and for who gains from price discrimination.

The rest of the paper is organized as follows. Section 2 reviews the related literature. Section 3 presents the model. Section 4 characterizes equilibrium prices and establishes the centrality-reversal result. Section 5 studies welfare and distributional effects. Section 6 designs interoperability. Section 7 analyzes network formation and platform incentives. Section 8 compares uniform and discriminatory pricing.

\section{Related Literature}

This paper is related to three strands of literature.

First, it contributes to the literature on pricing in social networks and network games. A large body of work builds on the linear-quadratic network game framework of \citet{BallesterCalvoArmengolZenou2006Econometrica}, in which equilibrium actions are proportional to Bonacich centrality. \citet{BlochQuerou2013GEB} studies monopoly pricing in social networks and identifies cases in which the monopolist does not discriminate across network positions. \citet{CandoganBimpikisOzdaglar2012OR} studies optimal monopoly pricing for divisible goods with local network effects and characterizes prices in terms of a network-independent component, a discount proportional to a consumer's influence, and a markup proportional to the influence exerted on that consumer. \citet{FainmesserGaleotti2016REStud} studies how information about consumers' network positions affects pricing when consumers generate network effects. \citet{FainmesserGaleotti2020AEJMicro} studies influencer marketing in oligopoly markets, where firms gather information on consumers' influence and price discriminate using this information. The closest paper to ours is \citet{ChenZenouZhou2018RAND}, who show that competing firms charge lower prices to more central consumers. We follow their linear-quadratic competitive pricing framework closely, but introduce cross-platform local network externalities. This allows us to ask whether the centrality-discount result survives when rival platforms become interoperable.

Second, our paper is related to work on price competition under local network externalities. \citet{Aoyagi2018JET} studies Bertrand competition when buyers are located on a network and the value of each firm's good depends on neighbors' adoption decisions. His analysis shows that, under linear local externalities, marginal-cost pricing and bipartition pricing can be consistent with equilibria, and applies the framework to platform competition in two-sided markets. Other contributions study price competition, market segmentation, and strategic interaction under local network effects \citep[e.g.,][]{BanerjiDutta2009IJIO,Jullien2011AEJMicro}. Our model differs from this literature in two respects. Consumers choose continuous consumption levels from differentiated platforms, and interoperability allows the consumption of one platform to generate local network benefits for consumers of the rival platform. Our focus is therefore not only on the existence or structure of equilibrium prices, but on how interoperability changes the sign of centrality-based price discrimination.

Third, our paper speaks to the literature on platforms, compatibility, interoperability, and multihoming. Classic contributions study how network externalities and compatibility affect adoption and competition \citep{Rohlfs1974Bell,KatzShapiro1985AER,FarrellSaloner1985RAND,FarrellSaloner1986AER,Economides1996IJIO}. The two-sided platform literature studies pricing when platforms internalize cross-group network effects \citep{Armstrong2006RAND,RochetTirole2003JEEA,RochetTirole2006RAND,CaillaudJullien2003RAND}. More recent work examines platform pricing, multihoming, and competition in richer platform environments \citep{HagiuWright2015IJIO,TanZhou2021REStud,JullienPavanRysman2021Handbook,TehLiuWrightZhou2023AEJMicro}. \citet{DoganogluWright2006IJIO} studies multihoming and compatibility in markets with network effects, and \citet{HuangTanTehZhou2026SSRN} develops a network approach to interoperability by modeling interoperability as a weighted network of connections among competing platforms and studying how different interoperability configurations affect equilibrium prices and welfare. These papers primarily study compatibility, multihoming, or interoperability at the platform level. By contrast, we take interoperability as a parameter governing cross-platform local network externalities among consumers and study its implications for individualized prices and the distribution of surplus within a social network.

Our contribution is to connect these literatures at the level of consumer network positions. In the social-network pricing literature, central consumers are typically valuable competitive targets because they help a firm expand its own network. In the platform-interoperability literature, cross-platform connections change how network value is shared across platforms. We show that when these forces interact, the sign of centrality-based price discrimination is no longer fixed. Interoperability can overturn the classic centrality discount and generate centrality premia. It can also change firms' incentives toward network formation, turning dense consumer networks from a source of competitive pressure into a source of shared monetizable value.

\section{Model}

Consider a market with two competing platforms, denoted by \(A\) and \(B\). The two platforms provide differentiated but substitutable services to a set of consumers
\[
	N=\{1,\ldots,n\}.
\]
Consumers are embedded in an undirected social network represented by an \(n\times n\) weighted adjacency matrix
\[
	\bm G=(g_{ij})_{i,j\in N}.
\]
We assume that \(g_{ij}\geq 0\), \(g_{ii}=0\) for all \(i\in N\), and \(g_{ij}=g_{ji}\). Thus, \(g_{ij}\) measures the strength of the social link between consumers \(i\) and \(j\), and \(g_{ij}=0\) means that there is no direct link between them. No regularity of the network is imposed unless explicitly stated.

Each consumer can consume services from both platforms. Let \(x_i^k\geq 0\) denote consumer \(i\)'s consumption of platform \(k\in\{A,B\}\). Let \(p_i^k\) be the price charged by platform \(k\) to consumer \(i\), and let \(a_i^k\) denote consumer \(i\)'s intrinsic marginal utility from platform \(k\). Consumer \(i\)'s utility is given by
\begin{align}
\begin{aligned}
	u_i
	=&\ a_i^A x_i^A+a_i^B x_i^B
	-\left\{\frac{1}{2}(x_i^A)^2+\frac{1}{2}(x_i^B)^2+\beta x_i^A x_i^B\right\} \\
	&+\delta\sum_{j\in N}g_{ij}x_i^A x_j^A
	+\delta\sum_{j\in N}g_{ij}x_i^B x_j^B \\
	&+\theta\delta\sum_{j\in N}g_{ij}x_i^A x_j^B
	+\theta\delta\sum_{j\in N}g_{ij}x_i^B x_j^A \\
	&-p_i^A x_i^A-p_i^B x_i^B .
\end{aligned}
\label{eq:utility}
\end{align}

The utility function has four components. The first line captures the consumer's intrinsic benefits and the standard quadratic costs of consuming the two services. The parameter
\[
	\beta\in[0,1)
\]
measures product substitutability: a higher \(\beta\) means that consuming one platform reduces the marginal utility of consuming the other platform more strongly.

The second line captures within-platform local network externalities. If consumer \(j\) consumes platform \(A\), then consumer \(i\)'s utility from consuming platform \(A\) increases in proportion to \(g_{ij}\); the same applies to platform \(B\). The parameter \(\delta\geq 0\) measures the strength of these same-platform network externalities.

The third line is the key departure from the standard competitive pricing model in social networks. It captures cross-platform local network externalities induced by interoperability. The parameter
\[
	\theta\in[0,1]
\]
measures the degree of interoperability between the two platforms. When \(\theta=0\), there is no cross-platform network externality: consumer \(i\)'s value from platform \(A\) is not directly affected by neighbors' consumption of platform \(B\), and vice versa. This does not mean that the two platforms are independent, since they may still be substitutes through \(\beta\). When \(\theta=1\), cross-platform network externalities are as strong as same-platform network externalities along the existing consumer network. Intermediate values of \(\theta\) represent partial interoperability.

For example, in online gaming, interoperability allows a user on one platform to interact with friends using another platform. In messaging or social-media services, interoperability may allow messages, content, or social interactions to travel across platform boundaries. In such environments, a consumer's value from one platform may depend not only on neighbors who consume the same platform, but also on neighbors who consume a rival interoperable platform.

The last line of \eqref{eq:utility} is consumer \(i\)'s total expenditure on the two platforms. We focus on the interior region of the linear-quadratic model, where the equilibrium quantities characterized below are nonnegative.

Platform \(k\)'s marginal cost of serving consumer \(i\) is denoted by \(c_i^k\). We assume
\[
	a_i^k>c_i^k,\qquad i\in N,\quad k\in\{A,B\}.
\]
Let
\[
	\bm p^k=(p_1^k,\ldots,p_n^k)^\top,\quad
	\bm x^k=(x_1^k,\ldots,x_n^k)^\top,\quad
	\bm c^k=(c_1^k,\ldots,c_n^k)^\top.
\]
Platform \(k\)'s profit is
\begin{align}
	\Pi^k
	=
	\left\langle \bm p^k-\bm c^k,\bm x^k\right\rangle
	=
	\sum_{i\in N}(p_i^k-c_i^k)x_i^k .
\label{eq:profit}
\end{align}
Our baseline analysis allows platforms to charge personalized prices. We later discuss uniform pricing and compare it with discriminatory pricing.

We impose the following stability condition.

\begin{assumption}[Network stability]
\label{ass:stability}
The parameters satisfy
\[
	\delta(1+\theta)\rho(\bm G)<1+\beta
	\quad\text{and}\quad
	\delta(1-\theta)\rho(\bm G)<1-\beta,
\]
where \(\rho(\bm G)\) is the spectral radius of \(\bm G\).
\end{assumption}

Assumption \ref{ass:stability} guarantees that the consumption game has a unique interior equilibrium for any price vectors in the relevant range. Equivalently, the matrices
\[
	(1+\beta)\bm I-\delta(1+\theta)\bm G
	\quad\text{and}\quad
	(1-\beta)\bm I-\delta(1-\theta)\bm G
\]
are nonsingular and well behaved. Notice that when \(\theta=1\), the second condition becomes \(0<1-\beta\), which is automatically satisfied because \(\beta<1\).

\begin{lemma}[Stability in the centrality-premium region]
\label{lem:stability_premium_region}
Suppose Assumption \ref{ass:stability} holds. Then
\[
	\delta(2-\theta)\rho(\bm G)<2-\beta.
\]
Consequently, the matrix
\[
	(2-\beta)\bm I-\delta(2-\theta)\bm G
\]
is nonsingular, and the Katz-Bonacich term appearing in the equilibrium price formula is well defined.

Moreover, when \(\theta\geq\beta\), the first stability condition,
\[
	\delta(1+\theta)\rho(\bm G)<1+\beta,
\]
already implies
\[
	\delta(2-\theta)\rho(\bm G)<2-\beta.
\]
Thus, the centrality-premium region \(\theta>\beta\) does not require an additional stability restriction beyond Assumption \ref{ass:stability}.
\end{lemma}

Lemma \ref{lem:stability_premium_region} clarifies that the case \(\theta>\beta\), in which central consumers may face premia rather than discounts, is not a mathematically unstable region of the model. The same stability assumption that guarantees a well-defined consumption game also guarantees that the pricing expressions used below are well defined.

We will use two symmetry assumptions at different points in the analysis.

\begin{assumption}[Platform symmetry]
\label{ass:platform_symmetry}
For every consumer \(i\in N\),
\[
	a_i^A=a_i^B=a_i,
	\qquad
	c_i^A=c_i^B=c_i .
\]
\end{assumption}

Assumption \ref{ass:platform_symmetry} imposes symmetry across platforms for each consumer, while allowing consumers to differ in their intrinsic valuations and marginal costs. This is the baseline assumption for our main pricing characterization.

\begin{assumption}[Full symmetry]
\label{ass:full_symmetry}
For every consumer \(i\in N\),
\[
	a_i^A=a_i^B=a,
	\qquad
	c_i^A=c_i^B=c .
\]
\end{assumption}

Assumption \ref{ass:full_symmetry} further removes ex ante heterogeneity across consumers. Under this assumption, consumers differ only in their positions in the social network. We use this stronger assumption when deriving transparent comparative statics and examples.

The timing is as follows. In the first stage, platforms \(A\) and \(B\) simultaneously choose price vectors \(\bm p^A\) and \(\bm p^B\). In the second stage, after observing both price vectors, consumers simultaneously choose their consumption bundles
\[
	\bm x_i=(x_i^A,x_i^B).
\]
The solution concept is subgame-perfect Nash equilibrium. We solve the game by backward induction: first characterizing the consumption equilibrium for given prices, and then deriving the platforms' equilibrium pricing strategies.


\section{Equilibrium Pricing}
\label{sec:equilibrium_pricing}

We solve the game by backward induction. This section first characterizes the interior consumption equilibrium for price vectors that induce nonnegative quantities. We then derive the equilibrium pricing strategies under platform symmetry and show how interoperability changes the relationship between network centrality and prices. All proofs are collected in the Appendix.

For notational convenience, define
\[
	\bm M^+
	=
	\left[(1+\beta)\bm I-\delta(1+\theta)\bm G\right]^{-1},
	\qquad
	\bm M^-
	=
	\left[(1-\beta)\bm I-\delta(1-\theta)\bm G\right]^{-1}.
\]
Under Assumption \ref{ass:stability}, both matrices are well defined.

\begin{proposition}[Consumption equilibrium]
\label{prop:consumption}
Suppose Assumption \ref{ass:stability} holds and the price vectors are such that the interior solution in \eqref{eq:consumption_equilibrium} is nonnegative. Then the consumption stage has a unique interior equilibrium characterized by 
\begin{align}
\begin{aligned}
	\bm x^A
	&=
	\frac{\bm M^++\bm M^-}{2}(\bm a^A-\bm p^A)
	+
	\frac{\bm M^+-\bm M^-}{2}(\bm a^B-\bm p^B),
	\\
	\bm x^B
	&=
	\frac{\bm M^++\bm M^-}{2}(\bm a^B-\bm p^B)
	+
	\frac{\bm M^+-\bm M^-}{2}(\bm a^A-\bm p^A).
\end{aligned}
\label{eq:consumption_equilibrium}
\end{align}
\end{proposition}

The two matrices \(\bm M^+\) and \(\bm M^-\) correspond to the sum and difference of the two consumption systems. The matrix \(\bm M^+\) captures the propagation of aggregate consumption across platforms, while \(\bm M^-\) captures the propagation of relative consumption between the two platforms. Interoperability strengthens the former through \(1+\theta\) and weakens the latter through \(1-\theta\). Thus, interoperability makes the two platform goods more connected in the network externality channel, even though they may remain substitutes through \(\beta\).

Under Assumption \ref{ass:platform_symmetry}, we have \(\bm a^A=\bm a^B=\bm a\) and \(\bm c^A=\bm c^B=\bm c\). In any symmetric pricing equilibrium with \(\bm p^A=\bm p^B=\bm p\), Proposition \ref{prop:consumption} implies
\[
	\bm x^A=\bm x^B=\bm x=\bm M^+(\bm a-\bm p).
\]
We now characterize the equilibrium price vector.

For any scalar \(\alpha\) satisfying \(\alpha\rho(\bm G)<1\) and any vector \(\bm y\in\mathbb R^n\), define the weighted Katz-Bonacich centrality vector by
\[
	\bm b(\bm G,\alpha,\bm y)
	=
	(\bm I-\alpha\bm G)^{-1}\bm y .
\]
In our pricing formula, the relevant discount factor is
\[
	\alpha
	=
	\frac{\delta(2-\theta)}{2-\beta}.
\]
Assumption \ref{ass:stability} guarantees that \(\alpha\rho(\bm G)<1\).

\begin{proposition}[Equilibrium price]
\label{prop:pricing}
Suppose Assumptions \ref{ass:stability} and \ref{ass:platform_symmetry} hold. Then the pricing game admits a unique symmetric equilibrium satisfying
\[
	\bm p^{A*}=\bm p^{B*}=\bm p^*.
\]
The equilibrium price vector is
\begin{align}
\boxed{
\bm p^*
=
\frac{\bm a+\bm c}{2}
-
\frac{\beta}{2(2-\beta)}(\bm a-\bm c)
+
\frac{\theta-\beta}{(2-\beta)(2-\theta)}
\left[
	\bm b\left(
		\bm G,
		\frac{\delta(2-\theta)}{2-\beta},
		\bm a-\bm c
	\right)
	-(\bm a-\bm c)
\right].
}
\label{eq:price_decomposition}
\end{align}
Equivalently,
\begin{align}
\bm p^*
=
\bm c
+
\frac{1-\theta}{2-\theta}(\bm a-\bm c)
+
\frac{\theta-\beta}{(2-\beta)(2-\theta)}
\bm b\left(
	\bm G,
	\frac{\delta(2-\theta)}{2-\beta},
	\bm a-\bm c
\right).
\label{eq:price_kb}
\end{align}
\end{proposition}

Equation \eqref{eq:price_decomposition} decomposes equilibrium prices into three components. The first term,
\[
	\frac{\bm a+\bm c}{2},
\]
is the standard monopoly benchmark under linear demand. The second term,
\[
	-\frac{\beta}{2(2-\beta)}(\bm a-\bm c),
\]
captures the direct effect of product substitutability. It is zero when \(\beta=0\), and becomes more negative as the two platforms become closer substitutes.

The third term is the network-position component. Its magnitude depends on weighted Katz-Bonacich centrality, while its sign is governed by
\[
	\theta-\beta.
\]
The subtraction of \(\bm a-\bm c\) isolates the part of Katz-Bonacich centrality generated by network interactions. In particular, when \(\delta=0\), we have
\[
	\bm b\left(
		\bm G,
		\frac{\delta(2-\theta)}{2-\beta},
		\bm a-\bm c
	\right)
	=
	\bm a-\bm c,
\]
so the network-position component vanishes.
This term captures the central mechanism of the paper. Product substitution creates a business-stealing incentive: each platform wants to attract influential consumers because doing so helps it win their neighbors. This force pushes firms to offer discounts to central consumers. Interoperability creates an opposing force: when cross-platform interactions are strong, the value generated by central consumers is no longer exclusive to one platform. Central consumers become high-value access points to a broader cross-platform network, which weakens the incentive to subsidize them and may generate centrality premia.

When \(\theta=0\), there is no cross-platform network externality, and the model reduces to the competitive pricing environment of \citet{ChenZenouZhou2018RAND} with product-specific local network effects. When \(\beta=\theta\), the product-substitution effect and the interoperability effect exactly offset each other in the network-position component. In that case, equilibrium prices depend on consumer fundamentals \(\bm a-\bm c\), but not on network position through Katz-Bonacich centrality. When \(\theta>\beta\), the sign of the centrality component is reversed relative to the standard centrality-discount result.

The following corollary states this implication in the cleanest case, where consumers differ only in their network positions.

\begin{corollary}[Centrality discount and centrality premium]
\label{cor:centrality_reversal}
Suppose Assumptions \ref{ass:stability} and \ref{ass:full_symmetry} hold. Let
\[
	\alpha=\frac{\delta(2-\theta)}{2-\beta},
	\qquad
	\bm b(\bm G,\alpha,\bm 1)=(b_1,\ldots,b_n)^\top.
\]
Then for each consumer \(i\),
\begin{align}
p_i^*
=
\frac{a+c}{2}
-
\frac{\beta}{2(2-\beta)}(a-c)
+
\frac{\theta-\beta}{(2-\beta)(2-\theta)}(a-c)(b_i-1).
\label{eq:individual_price}
\end{align}
Consequently, for any two consumers \(i\) and \(j\),
\[
	p_i^*-p_j^*
	=
	\frac{\theta-\beta}{(2-\beta)(2-\theta)}(a-c)(b_i-b_j).
\]
Hence, if \(\theta<\beta\), more central consumers pay lower prices; if \(\theta>\beta\), more central consumers pay higher prices; and if \(\theta=\beta\), network centrality has no effect on equilibrium prices.
\end{corollary}

Corollary \ref{cor:centrality_reversal} is the main pricing result. In closed or weakly interoperable platforms, the business-stealing force dominates, and central consumers receive influence-based discounts. Under sufficiently strong interoperability, the cross-platform access-value force dominates, and the same central consumers may be charged premia. Thus, interoperability can transform central consumers from subsidized influencers into premium consumers.

We finally record a useful benchmark for regular networks. This case shuts down centrality heterogeneity and isolates the effect of interoperability on the common equilibrium price.

\begin{corollary}[Regular networks]
\label{cor:regular_price}
Suppose Assumptions \ref{ass:stability} and \ref{ass:full_symmetry} hold, and suppose that \(\bm G\) is \(d\)-regular. Then all consumers face the same equilibrium price:
\begin{align}
	p^*
	=
	c+
	\frac{(1-\beta)-\delta d(1-\theta)}
	{(2-\beta)-\delta d(2-\theta)}
	(a-c).
\label{eq:regular_price}
\end{align}
Moreover,
\begin{align}
	\frac{\partial p^*}{\partial\theta}
	=
	\frac{\delta d(1-\delta d)}
	{\left[(2-\beta)-\delta d(2-\theta)\right]^2}
	(a-c)
	>
	0.
\label{eq:regular_price_derivative}
\end{align}
\end{corollary}

Corollary \ref{cor:regular_price} shows that, when all consumers have the same network position, interoperability raises the equilibrium price. The reason is that higher interoperability softens the competition for exclusive network participation. Since the regular network eliminates centrality-based price discrimination, this result should be interpreted as an aggregate pricing benchmark rather than as a statement about which consumers gain or lose from interoperability.

\section{Welfare and Distributional Effects}

This section studies how interoperability affects consumption, consumer surplus, platform profits, and total welfare. Throughout this section, we maintain Assumptions \ref{ass:stability} and \ref{ass:platform_symmetry}, so that the equilibrium is symmetric across platforms:
\[
	\bm p^{A*}=\bm p^{B*}=\bm p^*,
	\qquad
	\bm x^{A*}=\bm x^{B*}=\bm x^*.
\]
Let
\[
	\bm r=\bm a-\bm c,
\]
and define
\[
	\bm D=(1+\beta)\bm I-\delta(1+\theta)\bm G,
	\qquad
	\bm K=(1-\beta)\bm I-\delta(1-\theta)\bm G,
\]
\[
	\bm V=(2-\beta)\bm I-\delta(2-\theta)\bm G.
\]
Under Assumption \ref{ass:stability}, these matrices are nonsingular. Since they are all affine functions of \(\bm G\), they commute with each other.

The equilibrium consumption vector can be written as
\[
	\bm x^*
	=
	(\bm I-\delta\bm G)\bm D^{-1}\bm V^{-1}\bm r.
\]
It is useful to define the following operators:
\[
	\bm\Phi^{X}
	=
	(\bm I-\delta\bm G)\bm D^{-1}\bm V^{-1},
\]
\[
	\bm\Phi^{CS}
	=
	(1+\beta)(\bm\Phi^{X})^2,
\]
and
\[
	\bm\Phi^{\Pi}
	=
	(\bm I-\delta\bm G)\bm D^{-1}\bm K\bm V^{-2}.
\]

\begin{proposition}[Equilibrium welfare objects]
\label{prop:welfare_objects}
Suppose Assumptions \ref{ass:stability} and \ref{ass:platform_symmetry} hold. In the symmetric equilibrium,
\[
	\bm x^*=\bm\Phi^X\bm r.
\]
Moreover, consumer surplus, each platform's profit, and total welfare are given by
\begin{align}
	CS^*
	&=
	\left\langle
		\bm r,\bm\Phi^{CS}\bm r
	\right\rangle,
\label{eq:cs_operator}
\\
	\Pi^*
	&=
	\left\langle
		\bm r,\bm\Phi^{\Pi}\bm r
	\right\rangle,
\label{eq:profit_operator}
\\
	TW^*
	&=
	CS^*+2\Pi^*
	=
	\left\langle
		\bm r,
		(\bm\Phi^{CS}+2\bm\Phi^{\Pi})\bm r
	\right\rangle.
\label{eq:tw_operator}
\end{align}

At the individual level, consumer \(i\)'s equilibrium surplus is
\begin{align}
	CS_i^*
	=
	(1+\beta)(x_i^*)^2 .
\label{eq:individual_cs}
\end{align}
\end{proposition}

Proposition \ref{prop:welfare_objects} shows that the welfare effects of interoperability operate through equilibrium consumption and markups. Equation \eqref{eq:individual_cs} makes clear that the distribution of consumer surplus is tied to the distribution of equilibrium consumption across network positions. In general networks, the effect of interoperability on \(x_i^*\), \(CS_i^*\), and \(\Pi^*\) is ambiguous because interoperability both creates additional cross-platform network value and changes the intensity of price competition.

To obtain sharper results, we first consider regular networks. This benchmark shuts down centrality heterogeneity and isolates the aggregate effect of interoperability.

\begin{proposition}[Welfare objects in regular networks]
\label{prop:regular_welfare}
Suppose Assumptions \ref{ass:stability} and \ref{ass:full_symmetry} hold, and suppose that \(\bm G\) is \(d\)-regular. Then
\begin{align}
	\bm x^*
	=
	\frac{(a-c)(1-\delta d)}
	{\left[(1+\beta)-\delta d(1+\theta)\right]
	 \left[(2-\beta)-\delta d(2-\theta)\right]}
	\bm 1.
\label{eq:regular_consumption}
\end{align}
Consumer surplus and each platform's profit are
\begin{align}
	CS^*
	&=
	n(1+\beta)
	\left[
	\frac{(a-c)(1-\delta d)}
	{\left[(1+\beta)-\delta d(1+\theta)\right]
	 \left[(2-\beta)-\delta d(2-\theta)\right]}
	\right]^2,
\label{eq:regular_cs}
\\
	\Pi^*
	&=
	\frac{
	n(a-c)^2(1-\delta d)
	\left[(1-\beta)-\delta d(1-\theta)\right]
	}
	{
	\left[(1+\beta)-\delta d(1+\theta)\right]
	\left[(2-\beta)-\delta d(2-\theta)\right]^2
	}.
\label{eq:regular_profit}
\end{align}
\end{proposition}

The regular-network benchmark highlights a basic trade-off. Interoperability increases the value of cross-platform interactions, but it also softens competition by reducing the exclusivity of network benefits. These two forces need not affect consumers and platforms in the same direction.

Let
\[
	s=1-\delta d,
	\qquad
	m^+=(1+\beta)-\delta d(1+\theta),
\]
\[
	m^-=(1-\beta)-\delta d(1-\theta),
	\qquad
	v=(2-\beta)-\delta d(2-\theta).
\]
Under Assumption \ref{ass:stability}, these terms are positive.

\begin{proposition}[Interoperability, consumers, and platforms]
\label{prop:cs_profit_comparative}
Suppose Assumptions \ref{ass:stability} and \ref{ass:full_symmetry} hold, and suppose that \(\bm G\) is \(d\)-regular with \(\delta d>0\). Then
\[
	\operatorname{sign}
	\left\{
		\frac{\partial x^*}{\partial\theta}
	\right\}
	=
	\operatorname{sign}
	\left\{
		\frac{\partial CS^*}{\partial\theta}
	\right\}
	=
	\operatorname{sign}
	\left\{
		1-2\beta-\delta d+2\delta d\theta
	\right\}.
\]
Moreover,
\[
	\frac{\partial \Pi^*}{\partial\theta}>0.
\]
\end{proposition}

Proposition \ref{prop:cs_profit_comparative} shows that platforms benefit from greater interoperability in regular networks, while consumers may gain or lose. The consumer effect depends on
\[
	1-2\beta-\delta d+2\delta d\theta.
\]

The condition \(\delta d>0\) rules out the degenerate case in which interoperability has no bite. If \(\delta d=0\), equilibrium consumption and consumer surplus are independent of \(\theta\). If \(\delta d>0\), consumer surplus is strictly increasing in \(\theta\) on \([0,1]\) if
\[
	\beta\leq\frac{1-\delta d}{2},
\]
and strictly decreasing in \(\theta\) on \([0,1]\) if
\[
	\beta\geq\frac{1+\delta d}{2}.
\]
For the intermediate region,
\[
	\frac{1-\delta d}{2}
	<
	\beta
	<
	\frac{1+\delta d}{2},
\]
the sign changes at
\[
	\theta^{CS}
	=
	\frac{2\beta+\delta d-1}{2\delta d}.
\]
In this case, the consumer-surplus effect is negative when \(\theta<\theta^{CS}\) and positive when \(\theta>\theta^{CS}\). Hence, it is not generally correct to say that consumers always prefer either complete interoperability or no interoperability based only on whether \(\beta\) is below or above \(1/2\). The strength of network effects and the current level of interoperability also matter.

The monotonic increase in platform profit reflects a competition-softening mechanism. As interoperability increases, attracting a consumer becomes less valuable as an exclusive competitive asset, because part of the network value generated by that consumer is shared across platforms. This reduces the intensity of competition for exclusive network participation and allows platforms to sustain higher markups.

We next turn to total welfare. From \eqref{eq:tw_operator},
\[
	\frac{\partial TW^*}{\partial\theta}
	=
	\bm r^\top
	\left(
		\frac{\partial \bm\Phi^{CS}}{\partial\theta}
		+
		2\frac{\partial \bm\Phi^{\Pi}}{\partial\theta}
	\right)
	\bm r,
\]
whose sign is generally ambiguous. In regular networks, however, the derivative can be decomposed into a consumer-surplus effect and a profit effect.

\begin{proposition}[Interoperability and total welfare]
\label{prop:total_welfare}
Suppose Assumptions \ref{ass:stability} and \ref{ass:full_symmetry} hold, and suppose that \(\bm G\) is \(d\)-regular with \(\delta d>0\). Then
\begin{align}
\operatorname{sign}
\left\{
	\frac{\partial TW^*}{\partial\theta}
\right\}
=
\operatorname{sign}
\left\{
	\underbrace{s(1+\beta)(v-m^+)}_{\text{consumer-surplus effect}}
	+
	\underbrace{2m^+\left[s^2-m^-s+(m^-)^2\right]}_{\text{profit effect}}
\right\}.
\label{eq:tw_derivative_sign}
\end{align}
\end{proposition}

The first term in \eqref{eq:tw_derivative_sign} is proportional to the consumer-surplus effect. Since
\[
	v-m^+
	=
	1-2\beta-\delta d+2\delta d\theta,
\]
this term may be positive or negative. The second term is strictly positive, reflecting the positive effect of interoperability on platform profits. Therefore, total welfare increases with interoperability whenever the consumer-surplus effect is nonnegative. In particular,
\[
	\theta
	\geq
	\frac{2\beta-1+\delta d}{2\delta d}
	\quad
	\Longrightarrow
	\quad
	\frac{\partial TW^*}{\partial\theta}>0.
\]
When the consumer-surplus effect is negative, the welfare effect is ambiguous and depends on whether the profit gain is large enough to offset the reduction in consumer surplus.

This result highlights the central welfare trade-off created by interoperability. On the one hand, interoperability creates additional network value by allowing interactions across platform boundaries. On the other hand, it softens competition by reducing the value of exclusive network participation, thereby increasing equilibrium markups. Whether interoperability improves total welfare depends on the relative strength of these two forces.

Finally, although the regular-network case is useful for transparent aggregate comparative statics, it eliminates differences in network centrality. In nonregular networks, interoperability also has distributional effects across consumers. The pricing results in Section \ref{sec:equilibrium_pricing} show that the same central consumer may receive a discount when \(\theta<\beta\) but face a premium when \(\theta>\beta\). Since individual surplus is \(CS_i^*=(1+\beta)(x_i^*)^2\), the welfare consequences of interoperability depend not only on aggregate competition and network-value effects, but also on how prices and consumption change across network positions.

\medskip
\noindent\textbf{Distributional effects in nonregular networks.}
The preceding regular-network analysis shuts down heterogeneity in network position. We now use a local approximation to show how interoperability affects consumers with different degrees in a general network. This result directly addresses which consumers gain or lose from stronger interoperability.

Throughout this part, suppose Assumptions \ref{ass:stability} and \ref{ass:full_symmetry} hold. Recall that
\[
	\bm d=\bm G\bm 1,
	\qquad
	d_i=\sum_{j=1}^n g_{ij}.
\]
For sufficiently small \(\delta\), the discriminatory price satisfies
\begin{align}
p_i^*
=
c
+
(a-c)
\left[
	\frac{1-\beta}{2-\beta}
	+
	\frac{\delta(\theta-\beta)}{(2-\beta)^2}d_i
\right]
+
O(\delta^2),
\label{eq:local_price_degree}
\end{align}
and equilibrium consumption satisfies
\begin{align}
x_i^*
=
\frac{a-c}{(2-\beta)(1+\beta)}
\left[
	1
	+
	\delta
	\left(
		\frac{1+\theta}{1+\beta}
		-
		\frac{\theta-\beta}{2-\beta}
	\right)
	d_i
\right]
+
O(\delta^2).
\label{eq:local_consumption_degree}
\end{align}

\begin{proposition}[Local distributional effect of interoperability]
\label{prop:local_distributional_theta}
Suppose Assumptions \ref{ass:stability} and \ref{ass:full_symmetry} hold. For sufficiently small \(\delta\), the effect of interoperability on the discriminatory price and individual consumer surplus is given by
\begin{align}
\frac{\partial p_i^*}{\partial\theta}
&=
\frac{\delta(a-c)}{(2-\beta)^2}d_i
+
O(\delta^2),
\label{eq:dpidtheta_local}
\\
\frac{\partial CS_i^*}{\partial\theta}
&=
\frac{
	2\delta(a-c)^2(1-2\beta)
}{
	(2-\beta)^3(1+\beta)^2
}
d_i
+
O(\delta^2).
\label{eq:dcsidtheta_local}
\end{align}
Consequently, for any consumer with \(d_i>0\), stronger interoperability locally raises her price. Its effect on her consumer surplus is positive if \(\beta<1/2\), negative if \(\beta>1/2\), and zero to the first order if \(\beta=1/2\). The magnitude of both effects is increasing in the consumer's degree \(d_i\).
\end{proposition}

Proposition \ref{prop:local_distributional_theta} separates the price effect from the welfare effect. Stronger interoperability raises the price charged to more connected consumers because it increases the access value of their network position. However, a higher price does not necessarily imply that these consumers are worse off. When products are sufficiently differentiated, \(\beta<1/2\), interoperability raises the network value of consumption enough that high-degree consumers may pay more and still enjoy higher surplus. When products are close substitutes, \(\beta>1/2\), the price increase dominates locally, and high-degree consumers are hurt more strongly.

Aggregating \eqref{eq:dcsidtheta_local} over consumers gives
\begin{align}
\frac{\partial CS^*}{\partial\theta}
=
\frac{
	2\delta(a-c)^2(1-2\beta)
}{
	(2-\beta)^3(1+\beta)^2
}
\sum_{i=1}^n d_i
+
O(\delta^2).
\label{eq:dcsdtheta_local_aggregate}
\end{align}
Thus, in the local approximation, the aggregate consumer-surplus effect and the individual consumer-surplus effects have the same sign, but the incidence is uneven: more connected consumers experience larger gains when \(\beta<1/2\) and larger losses when \(\beta>1/2\).

\begin{example}[Exact star-network example]
\label{ex:exact_star}
Consider an unweighted star network with one center \(C\) and \(m=n-1\) peripheral consumers. Let
\[
	\alpha=\frac{\delta(2-\theta)}{2-\beta},
	\qquad
	\Delta_\alpha(m)=1-\alpha^2m.
\]
By symmetry, all peripheral consumers have the same Katz-Bonacich centrality. Solving the two-type system gives
\[
	b_C=\frac{1+\alpha m}{\Delta_\alpha(m)},
	\qquad
	b_P=\frac{1+\alpha}{\Delta_\alpha(m)}.
\]
Hence the exact discriminatory prices are
\[
	p_C^*
	=
	c+\frac{1-\theta}{2-\theta}(a-c)
	+
	\frac{\theta-\beta}{(2-\beta)(2-\theta)}(a-c)b_C,
\]
and
\[
	p_P^*
	=
	c+\frac{1-\theta}{2-\theta}(a-c)
	+
	\frac{\theta-\beta}{(2-\beta)(2-\theta)}(a-c)b_P.
\]
Therefore,
\[
	p_C^*-p_P^*
	=
	\frac{\theta-\beta}{(2-\beta)(2-\theta)}(a-c)
	\frac{\alpha(m-1)}{1-\alpha^2m}.
\]
Thus, in the star network, the center receives a discount relative to peripheral consumers when \(\theta<\beta\), but pays a premium when \(\theta>\beta\).

The exact quantities can also be reduced to a two-type system. For any \(\lambda\) and \(\mu\), define
\[
	R_m(\lambda,\mu)
	=
	\begin{pmatrix}
		\lambda & -\mu m \\
		-\mu & \lambda
	\end{pmatrix}.
\]
Then
\[
	\begin{pmatrix}
		x_C^* \\
		x_P^*
	\end{pmatrix}
	=
	(a-c)
	R_m(1,\delta)
	R_m(1+\beta,\delta(1+\theta))^{-1}
	R_m(2-\beta,\delta(2-\theta))^{-1}
	\begin{pmatrix}
		1\\
		1
	\end{pmatrix}.
\]
Consequently,
\[
	CS_C^*=(1+\beta)(x_C^*)^2,
	\qquad
	CS_P^*=(1+\beta)(x_P^*)^2.
\]

Table \ref{tab:exact_star} reports exact equilibrium values for a star network with \(m=8\), \(a=3\), \(c=1\), \(\beta=0.5\), and \(\delta=0.2\). The case \(\theta=0.5\) corresponds to \(\theta=\beta\), where the center and peripheral consumers face the same price. When \(\theta=0.2<\beta\), the center receives a large centrality discount. When \(\theta=0.8>\beta\), the center pays a centrality premium. Comparing \(\theta=0.5\) and \(\theta=0.8\), the center pays a higher price but also consumes more and obtains higher surplus. Thus, a centrality premium need not imply that central consumers are worse off.

\begin{table}[H]
\centering
\caption{Exact Equilibrium Values in a Star Network}
\label{tab:exact_star}
\resizebox{\textwidth}{!}{
\begin{tabular}{lccccccccc}
\toprule
& \multicolumn{3}{c}{\(\theta=0.2\)} 
& \multicolumn{3}{c}{\(\theta=0.5\)} 
& \multicolumn{3}{c}{\(\theta=0.8\)} \\
\cmidrule(lr){2-4}\cmidrule(lr){5-7}\cmidrule(lr){8-10}
Consumer type 
& Price & Quantity & Surplus 
& Price & Quantity & Surplus 
& Price & Quantity & Surplus \\
\midrule
Center \(C\)      
& 0.685 & 3.681 & 20.327 
& 1.667 & 3.399 & 17.327 
& 2.289 & 3.681 & 20.327 \\
Peripheral \(P\) 
& 1.378 & 1.670 & 4.185 
& 1.667 & 1.569 & 3.691 
& 1.820 & 1.670 & 4.185 \\
\bottomrule
\end{tabular}
}
\end{table}

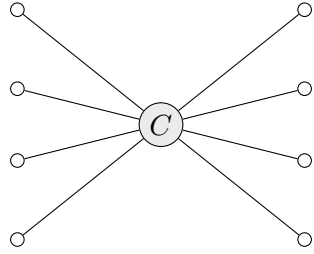
\begin{figure}[H]
\centering
\begin{tikzpicture}[scale=0.95, every node/.style={circle, draw, inner sep=1.8pt, font=\small}]
	\node[fill=black!8] (C) at (0,0) {\(C\)};
	\node (L1) at (-2,1.6) {};
	\node (L2) at (-2,0.5) {};
	\node (L3) at (-2,-0.5) {};
	\node (L4) at (-2,-1.6) {};
	\node (R1) at (2,1.6) {};
	\node (R2) at (2,0.5) {};
	\node (R3) at (2,-0.5) {};
	\node (R4) at (2,-1.6) {};
	\draw (C)--(L1);
	\draw (C)--(L2);
	\draw (C)--(L3);
	\draw (C)--(L4);
	\draw (C)--(R1);
	\draw (C)--(R2);
	\draw (C)--(R3);
	\draw (C)--(R4);
\end{tikzpicture}
\caption{A Star Network with \(m=8\)}
\label{fig:star_network}
\end{figure}

Figure \ref{fig:star_price_quantity} plots the exact equilibrium prices and quantities as \(\theta\) varies from \(0\) to \(1\), holding \(m=8\), \(a=3\), \(c=1\), \(\beta=0.5\), and \(\delta=0.2\) fixed. The center's price rises much faster than the peripheral price as interoperability increases. Quantities are minimized around the neutral point \(\theta=\beta\), and increase as interoperability becomes sufficiently strong. Hence, moving from \(\theta=0.5\) to \(\theta=0.8\), the center is charged a higher price but also obtains a larger quantity and higher surplus.

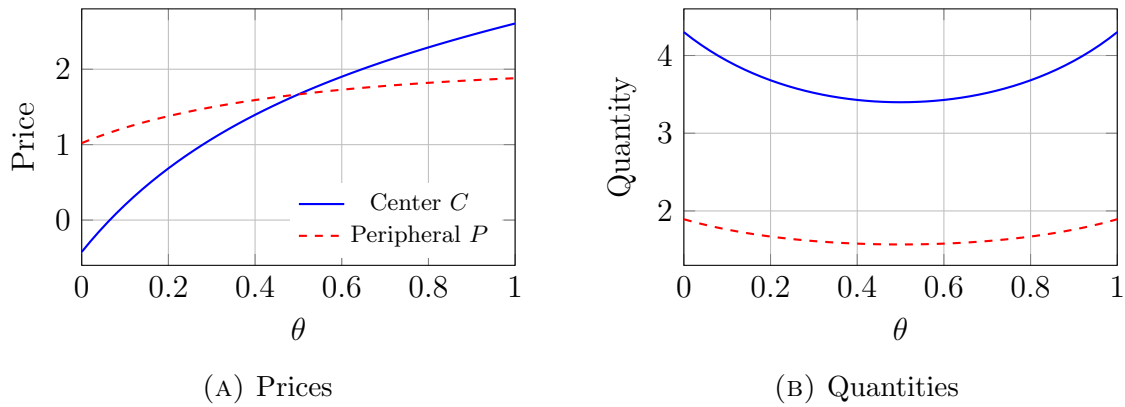
\begin{figure}[H]
\centering
\begin{subfigure}[b]{0.48\textwidth}
\centering
\begin{tikzpicture}
\begin{axis}[
	width=\textwidth,
	height=0.68\textwidth,
	xlabel={\(\theta\)},
	ylabel={Price},
	xmin=0, xmax=1,
	ymin=-0.6, ymax=2.8,
	xtick={0,0.2,0.4,0.6,0.8,1.0},
	grid=both,
	legend style={draw=none, at={(0.97,0.03)}, anchor=south east, font=\scriptsize}
]
\addplot+[mark=none, thick] coordinates {
	(0.00,-0.4227)
	(0.02,-0.2837)
	(0.04,-0.1524)
	(0.06,-0.0282)
	(0.08,0.0896)
	(0.10,0.2013)
	(0.12,0.3076)
	(0.14,0.4088)
	(0.16,0.5052)
	(0.18,0.5974)
	(0.20,0.6855)
	(0.22,0.7698)
	(0.24,0.8506)
	(0.26,0.9281)
	(0.28,1.0026)
	(0.30,1.0743)
	(0.32,1.1432)
	(0.34,1.2096)
	(0.36,1.2737)
	(0.38,1.3355)
	(0.40,1.3953)
	(0.42,1.4531)
	(0.44,1.5090)
	(0.46,1.5632)
	(0.48,1.6157)
	(0.50,1.6667)
	(0.52,1.7161)
	(0.54,1.7642)
	(0.56,1.8109)
	(0.58,1.8564)
	(0.60,1.9007)
	(0.62,1.9438)
	(0.64,1.9858)
	(0.66,2.0269)
	(0.68,2.0669)
	(0.70,2.1060)
	(0.72,2.1442)
	(0.74,2.1816)
	(0.76,2.2182)
	(0.78,2.2540)
	(0.80,2.2891)
	(0.82,2.3234)
	(0.84,2.3572)
	(0.86,2.3902)
	(0.88,2.4227)
	(0.90,2.4546)
	(0.92,2.4859)
	(0.94,2.5167)
	(0.96,2.5470)
	(0.98,2.5769)
	(1.00,2.6062)
};
\addlegendentry{Center \(C\)}
\addplot+[mark=none, dashed, thick] coordinates {
	(0.00,1.0206)
	(0.02,1.0664)
	(0.04,1.1095)
	(0.06,1.1500)
	(0.08,1.1883)
	(0.10,1.2243)
	(0.12,1.2584)
	(0.14,1.2907)
	(0.16,1.3213)
	(0.18,1.3503)
	(0.20,1.3778)
	(0.22,1.4040)
	(0.24,1.4289)
	(0.26,1.4527)
	(0.28,1.4753)
	(0.30,1.4968)
	(0.32,1.5174)
	(0.34,1.5371)
	(0.36,1.5558)
	(0.38,1.5738)
	(0.40,1.5910)
	(0.42,1.6075)
	(0.44,1.6232)
	(0.46,1.6383)
	(0.48,1.6528)
	(0.50,1.6667)
	(0.52,1.6800)
	(0.54,1.6928)
	(0.56,1.7050)
	(0.58,1.7168)
	(0.60,1.7281)
	(0.62,1.7390)
	(0.64,1.7494)
	(0.66,1.7595)
	(0.68,1.7691)
	(0.70,1.7784)
	(0.72,1.7873)
	(0.74,1.7958)
	(0.76,1.8041)
	(0.78,1.8120)
	(0.80,1.8196)
	(0.82,1.8269)
	(0.84,1.8339)
	(0.86,1.8406)
	(0.88,1.8471)
	(0.90,1.8533)
	(0.92,1.8593)
	(0.94,1.8650)
	(0.96,1.8705)
	(0.98,1.8758)
	(1.00,1.8808)
};
\addlegendentry{Peripheral \(P\)}
\end{axis}
\end{tikzpicture}
\caption{Prices}
\end{subfigure}
\hfill
\begin{subfigure}[b]{0.48\textwidth}
\centering
\begin{tikzpicture}
\begin{axis}[
	width=\textwidth,
	height=0.68\textwidth,
	xlabel={\(\theta\)},
	ylabel={Quantity},
	xmin=0, xmax=1,
	ymin=1.3, ymax=4.6,
	xtick={0,0.2,0.4,0.6,0.8,1.0},
	grid=both
]
\addplot+[mark=none, thick] coordinates {
	(0.00,4.3011)
	(0.02,4.2153)
	(0.04,4.1359)
	(0.06,4.0624)
	(0.08,3.9943)
	(0.10,3.9313)
	(0.12,3.8730)
	(0.14,3.8190)
	(0.16,3.7693)
	(0.18,3.7234)
	(0.20,3.6812)
	(0.22,3.6425)
	(0.24,3.6071)
	(0.26,3.5748)
	(0.28,3.5456)
	(0.30,3.5193)
	(0.32,3.4958)
	(0.34,3.4750)
	(0.36,3.4569)
	(0.38,3.4413)
	(0.40,3.4282)
	(0.42,3.4175)
	(0.44,3.4092)
	(0.46,3.4034)
	(0.48,3.3999)
	(0.50,3.3987)
	(0.52,3.3999)
	(0.54,3.4034)
	(0.56,3.4092)
	(0.58,3.4175)
	(0.60,3.4282)
	(0.62,3.4413)
	(0.64,3.4569)
	(0.66,3.4750)
	(0.68,3.4958)
	(0.70,3.5193)
	(0.72,3.5456)
	(0.74,3.5748)
	(0.76,3.6071)
	(0.78,3.6425)
	(0.80,3.6812)
	(0.82,3.7234)
	(0.84,3.7693)
	(0.86,3.8190)
	(0.88,3.8730)
	(0.90,3.9313)
	(0.92,3.9943)
	(0.94,4.0624)
	(0.96,4.1359)
	(0.98,4.2153)
	(1.00,4.3011)
};
\addplot+[mark=none, dashed, thick] coordinates {
	(0.00,1.8931)
	(0.02,1.8623)
	(0.04,1.8338)
	(0.06,1.8075)
	(0.08,1.7830)
	(0.10,1.7604)
	(0.12,1.7394)
	(0.14,1.7200)
	(0.16,1.7021)
	(0.18,1.6856)
	(0.20,1.6704)
	(0.22,1.6565)
	(0.24,1.6437)
	(0.26,1.6321)
	(0.28,1.6216)
	(0.30,1.6121)
	(0.32,1.6037)
	(0.34,1.5962)
	(0.36,1.5896)
	(0.38,1.5840)
	(0.40,1.5793)
	(0.42,1.5754)
	(0.44,1.5724)
	(0.46,1.5703)
	(0.48,1.5690)
	(0.50,1.5686)
	(0.52,1.5690)
	(0.54,1.5703)
	(0.56,1.5724)
	(0.58,1.5754)
	(0.60,1.5793)
	(0.62,1.5840)
	(0.64,1.5896)
	(0.66,1.5962)
	(0.68,1.6037)
	(0.70,1.6121)
	(0.72,1.6216)
	(0.74,1.6321)
	(0.76,1.6437)
	(0.78,1.6565)
	(0.80,1.6704)
	(0.82,1.6856)
	(0.84,1.7021)
	(0.86,1.7200)
	(0.88,1.7394)
	(0.90,1.7604)
	(0.92,1.7830)
	(0.94,1.8075)
	(0.96,1.8338)
	(0.98,1.8623)
	(1.00,1.8931)
};
\end{axis}
\end{tikzpicture}
\caption{Quantities}
\end{subfigure}
\caption{Prices and Quantities in the Star Network, $\beta=0.5$}
\label{fig:star_price_quantity}
\end{figure}

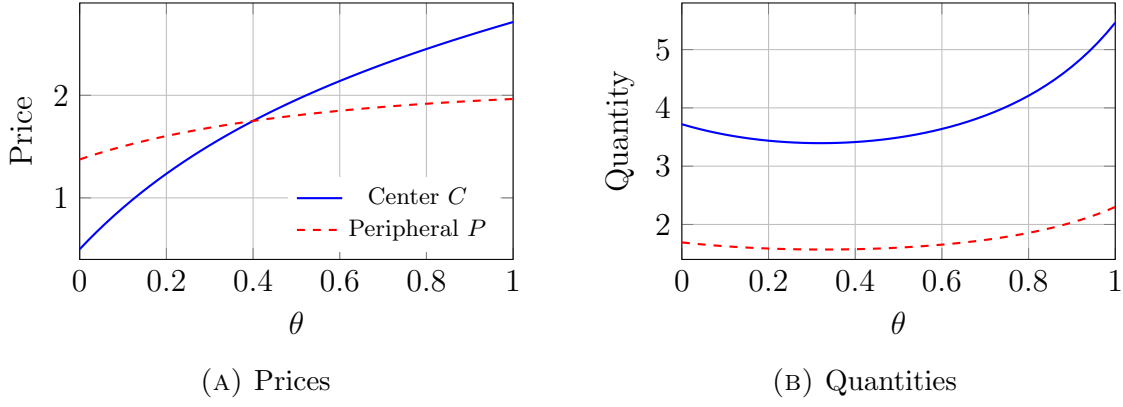
\begin{figure}[H]
\centering
\begin{subfigure}[b]{0.48\textwidth}
\centering
\begin{tikzpicture}
\begin{axis}[
	width=\textwidth,
	height=0.68\textwidth,
	xlabel={\(\theta\)},
	ylabel={Price},
	xmin=0, xmax=1,
	ymin=0.4, ymax=2.9,
	xtick={0,0.2,0.4,0.6,0.8,1.0},
	grid=both,
	legend style={draw=none, at={(0.97,0.03)}, anchor=south east, font=\scriptsize}
]
\addplot+[mark=none, thick] coordinates {
	(0.00,0.5000)
	(0.02,0.5880)
	(0.04,0.6722)
	(0.06,0.7528)
	(0.08,0.8301)
	(0.10,0.9043)
	(0.12,0.9756)
	(0.14,1.0442)
	(0.16,1.1103)
	(0.18,1.1739)
	(0.20,1.2353)
	(0.22,1.2946)
	(0.24,1.3518)
	(0.26,1.4072)
	(0.28,1.4608)
	(0.30,1.5127)
	(0.32,1.5630)
	(0.34,1.6119)
	(0.36,1.6592)
	(0.38,1.7053)
	(0.40,1.7500)
	(0.42,1.7935)
	(0.44,1.8359)
	(0.46,1.8771)
	(0.48,1.9173)
	(0.50,1.9565)
	(0.52,1.9948)
	(0.54,2.0321)
	(0.56,2.0686)
	(0.58,2.1042)
	(0.60,2.1391)
	(0.62,2.1732)
	(0.64,2.2066)
	(0.66,2.2392)
	(0.68,2.2713)
	(0.70,2.3027)
	(0.72,2.3335)
	(0.74,2.3637)
	(0.76,2.3934)
	(0.78,2.4226)
	(0.80,2.4512)
	(0.82,2.4794)
	(0.84,2.5071)
	(0.86,2.5344)
	(0.88,2.5612)
	(0.90,2.5876)
	(0.92,2.6137)
	(0.94,2.6393)
	(0.96,2.6647)
	(0.98,2.6896)
	(1.00,2.7143)
};
\addlegendentry{Center \(C\)}
\addplot+[mark=none, dashed, thick] coordinates {
	(0.00,1.3750)
	(0.02,1.4030)
	(0.04,1.4297)
	(0.06,1.4551)
	(0.08,1.4792)
	(0.10,1.5023)
	(0.12,1.5243)
	(0.14,1.5453)
	(0.16,1.5654)
	(0.18,1.5846)
	(0.20,1.6029)
	(0.22,1.6205)
	(0.24,1.6374)
	(0.26,1.6536)
	(0.28,1.6691)
	(0.30,1.6840)
	(0.32,1.6982)
	(0.34,1.7120)
	(0.36,1.7251)
	(0.38,1.7378)
	(0.40,1.7500)
	(0.42,1.7617)
	(0.44,1.7730)
	(0.46,1.7838)
	(0.48,1.7943)
	(0.50,1.8043)
	(0.52,1.8140)
	(0.54,1.8234)
	(0.56,1.8323)
	(0.58,1.8410)
	(0.60,1.8493)
	(0.62,1.8574)
	(0.64,1.8651)
	(0.66,1.8726)
	(0.68,1.8798)
	(0.70,1.8867)
	(0.72,1.8934)
	(0.74,1.8998)
	(0.76,1.9060)
	(0.78,1.9119)
	(0.80,1.9177)
	(0.82,1.9232)
	(0.84,1.9285)
	(0.86,1.9336)
	(0.88,1.9386)
	(0.90,1.9433)
	(0.92,1.9478)
	(0.94,1.9522)
	(0.96,1.9564)
	(0.98,1.9604)
	(1.00,1.9643)
};
\addlegendentry{Peripheral \(P\)}
\end{axis}
\end{tikzpicture}
\caption{Prices}
\end{subfigure}
\hfill
\begin{subfigure}[b]{0.48\textwidth}
\centering
\begin{tikzpicture}
\begin{axis}[
	width=\textwidth,
	height=0.68\textwidth,
	xlabel={\(\theta\)},
	ylabel={Quantity},
	xmin=0, xmax=1,
	ymin=1.4, ymax=5.8,
	xtick={0,0.2,0.4,0.6,0.8,1.0},
	grid=both
]
\addplot+[mark=none, thick] coordinates {
	(0.00,3.7195)
	(0.02,3.6772)
	(0.04,3.6384)
	(0.06,3.6029)
	(0.08,3.5706)
	(0.10,3.5414)
	(0.12,3.5151)
	(0.14,3.4916)
	(0.16,3.4708)
	(0.18,3.4527)
	(0.20,3.4371)
	(0.22,3.4241)
	(0.24,3.4135)
	(0.26,3.4053)
	(0.28,3.3994)
	(0.30,3.3960)
	(0.32,3.3949)
	(0.34,3.3961)
	(0.36,3.3997)
	(0.38,3.4056)
	(0.40,3.4139)
	(0.42,3.4246)
	(0.44,3.4377)
	(0.46,3.4533)
	(0.48,3.4715)
	(0.50,3.4923)
	(0.52,3.5158)
	(0.54,3.5420)
	(0.56,3.5712)
	(0.58,3.6034)
	(0.60,3.6387)
	(0.62,3.6773)
	(0.64,3.7193)
	(0.66,3.7650)
	(0.68,3.8146)
	(0.70,3.8683)
	(0.72,3.9263)
	(0.74,3.9890)
	(0.76,4.0567)
	(0.78,4.1297)
	(0.80,4.2086)
	(0.82,4.2937)
	(0.84,4.3856)
	(0.86,4.4850)
	(0.88,4.5925)
	(0.90,4.7089)
	(0.92,4.8351)
	(0.94,4.9722)
	(0.96,5.1215)
	(0.98,5.2842)
	(1.00,5.4622)
};
\addplot+[mark=none, dashed, thick] coordinates {
	(0.00,1.6921)
	(0.02,1.6765)
	(0.04,1.6622)
	(0.06,1.6491)
	(0.08,1.6372)
	(0.10,1.6263)
	(0.12,1.6165)
	(0.14,1.6077)
	(0.16,1.5999)
	(0.18,1.5931)
	(0.20,1.5871)
	(0.22,1.5821)
	(0.24,1.5780)
	(0.26,1.5747)
	(0.28,1.5723)
	(0.30,1.5707)
	(0.32,1.5700)
	(0.34,1.5701)
	(0.36,1.5711)
	(0.38,1.5729)
	(0.40,1.5756)
	(0.42,1.5792)
	(0.44,1.5836)
	(0.46,1.5889)
	(0.48,1.5952)
	(0.50,1.6024)
	(0.52,1.6105)
	(0.54,1.6197)
	(0.56,1.6299)
	(0.58,1.6412)
	(0.60,1.6536)
	(0.62,1.6672)
	(0.64,1.6820)
	(0.66,1.6982)
	(0.68,1.7157)
	(0.70,1.7347)
	(0.72,1.7552)
	(0.74,1.7774)
	(0.76,1.8014)
	(0.78,1.8273)
	(0.80,1.8553)
	(0.82,1.8855)
	(0.84,1.9181)
	(0.86,1.9534)
	(0.88,1.9916)
	(0.90,2.0329)
	(0.92,2.0777)
	(0.94,2.1264)
	(0.96,2.1794)
	(0.98,2.2372)
	(1.00,2.3004)
};
\end{axis}
\end{tikzpicture}
\caption{Quantities}
\end{subfigure}
\caption{Prices and Quantities in the Star Network, $\beta=0.4$}

\label{fig:star_price_quantity_beta04}
\end{figure}

\end{example}

\begin{example}[Exact double-star example]
\label{ex:exact_double_star}
We next consider a double-star network. There are two connected centers, denoted by \(L\) and \(R\). Center \(L\) is connected to \(M\) peripheral consumers, center \(R\) is connected to \(N\) peripheral consumers, and the two centers are connected to each other. This network allows the two centers to have different degrees:
\[
	d_L=M+1,
	\qquad
	d_R=N+1.
\]

The four consumer types are \(L\), \(R\), a peripheral consumer attached to \(L\), and a peripheral consumer attached to \(R\). For type-symmetric vectors, the action of the adjacency matrix is represented exactly by
\[
	\mathcal G_{M,N}
	=
	\begin{pmatrix}
		0 & 1 & M & 0 \\
		1 & 0 & 0 & N \\
		1 & 0 & 0 & 0 \\
		0 & 1 & 0 & 0
	\end{pmatrix}.
\]
Let
\[
	\bm 1_4=(1,1,1,1)^\top.
\]
The exact Katz-Bonacich centrality vector for the four types is
\[
	\bm b^{DS}
	=
	(\bm I_4-\alpha\mathcal G_{M,N})^{-1}\bm 1_4.
\]
Hence, for each type \(t\in\{L,R,P_L,P_R\}\),
\[
	p_t^*
	=
	c+\frac{1-\theta}{2-\theta}(a-c)
	+
	\frac{\theta-\beta}{(2-\beta)(2-\theta)}(a-c)b_t^{DS}.
\]

Similarly, the exact quantity vector is
\[
	\bm x^{DS}
	=
	(a-c)
	(\bm I_4-\delta\mathcal G_{M,N})
	\left[(1+\beta)\bm I_4-\delta(1+\theta)\mathcal G_{M,N}\right]^{-1}
	\left[(2-\beta)\bm I_4-\delta(2-\theta)\mathcal G_{M,N}\right]^{-1}
	\bm 1_4.
\]
Thus, for each type \(t\),
\[
	CS_t^*=(1+\beta)(x_t^{DS})^2.
\]

Table \ref{tab:exact_double_star} reports exact equilibrium values for a double-star network with \(M=8\), \(N=4\), \(a=3\), \(c=1\), \(\beta=0.5\), and \(\delta=0.2\). The case \(\theta=0.5\) corresponds to \(\theta=\beta\), where all consumer types face the same price. When \(\theta=0.2<\beta\), the two centers receive discounts, and the more connected center \(L\) receives the larger discount. When \(\theta=0.8>\beta\), the ranking is reversed: the more connected center \(L\) pays the highest price. Comparing \(\theta=0.5\) and \(\theta=0.8\), both centers pay higher prices but also consume more and obtain higher surplus. Thus, even when interoperability turns centrality into a price premium, central consumers may still benefit from stronger cross-platform network access.

\begin{table}[H]
\centering
\caption{Exact Equilibrium Values in a Double-Star Network}
\label{tab:exact_double_star}
\resizebox{\textwidth}{!}{
\begin{tabular}{lccccccccc}
\toprule
& \multicolumn{3}{c}{\(\theta=0.2\)}
& \multicolumn{3}{c}{\(\theta=0.5\)}
& \multicolumn{3}{c}{\(\theta=0.8\)} \\
\cmidrule(lr){2-4}\cmidrule(lr){5-7}\cmidrule(lr){8-10}
Consumer type
& Price & Quantity & Surplus
& Price & Quantity & Surplus
& Price & Quantity & Surplus \\
\midrule
Center \(L\)
& 0.199 & 4.871 & 35.590
& 1.667 & 4.257 & 27.183
& 2.452 & 4.871 & 35.590 \\
Center \(R\)
& 0.796 & 3.264 & 15.978
& 1.667 & 2.918 & 12.775
& 2.142 & 3.264 & 15.978 \\
Peripheral \(P_L\)
& 1.261 & 1.939 & 5.637
& 1.667 & 1.740 & 4.543
& 1.846 & 1.939 & 5.637 \\
Peripheral \(P_R\)
& 1.404 & 1.586 & 3.773
& 1.667 & 1.473 & 3.253
& 1.796 & 1.586 & 3.773 \\
\bottomrule
\end{tabular}
}
\end{table}

\begin{figure}[H]
\centering
\begin{tikzpicture}[scale=0.95, every node/.style={circle, draw, inner sep=1.8pt, font=\small}]
	\node[fill=black!8] (L) at (-1.2,0) {\(L\)};
	\node[fill=black!8] (R) at (1.2,0) {\(R\)};
	\draw (L)--(R);

	\node (LL1) at (-3.2,1.6) {};
	\node (LL2) at (-3.4,1.0) {};
	\node (LL3) at (-3.5,0.35) {};
	\node (LL4) at (-3.5,-0.35) {};
	\node (LL5) at (-3.4,-1.0) {};
	\node (LL6) at (-3.2,-1.6) {};
	\node (LL7) at (-2.5,2.0) {};
	\node (LL8) at (-2.5,-2.0) {};
	\draw (L)--(LL1);
	\draw (L)--(LL2);
	\draw (L)--(LL3);
	\draw (L)--(LL4);
	\draw (L)--(LL5);
	\draw (L)--(LL6);
	\draw (L)--(LL7);
	\draw (L)--(LL8);

	\node (RR1) at (3.2,1.4) {};
	\node (RR2) at (3.4,0.45) {};
	\node (RR3) at (3.4,-0.45) {};
	\node (RR4) at (3.2,-1.4) {};
	\draw (R)--(RR1);
	\draw (R)--(RR2);
	\draw (R)--(RR3);
	\draw (R)--(RR4);

	\node[draw=none] at (-3.8,-2.35) {\(M=8\) leaves};
	\node[draw=none] at (3.8,-2.35) {\(N=4\) leaves};
\end{tikzpicture}
\caption{A Double-Star Network with \(M=8\) and \(N=4\)}
\label{fig:double_star_network}
\end{figure}
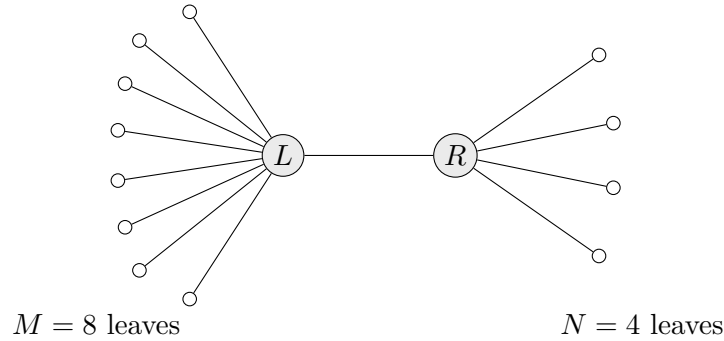

Figure \ref{fig:double_star_price_quantity} plots the exact equilibrium prices and quantities as \(\theta\) varies from \(0\) to \(1\), holding \(M=8\), \(N=4\), \(a=3\), \(c=1\), \(\beta=0.5\), and \(\delta=0.2\) fixed. The more connected center \(L\) has the steepest price response to interoperability. For \(\theta<\beta\), it receives the largest discount; for \(\theta>\beta\), it pays the largest premium. Quantities are lowest around the neutral point \(\theta=\beta\) and rise as interoperability becomes sufficiently strong.

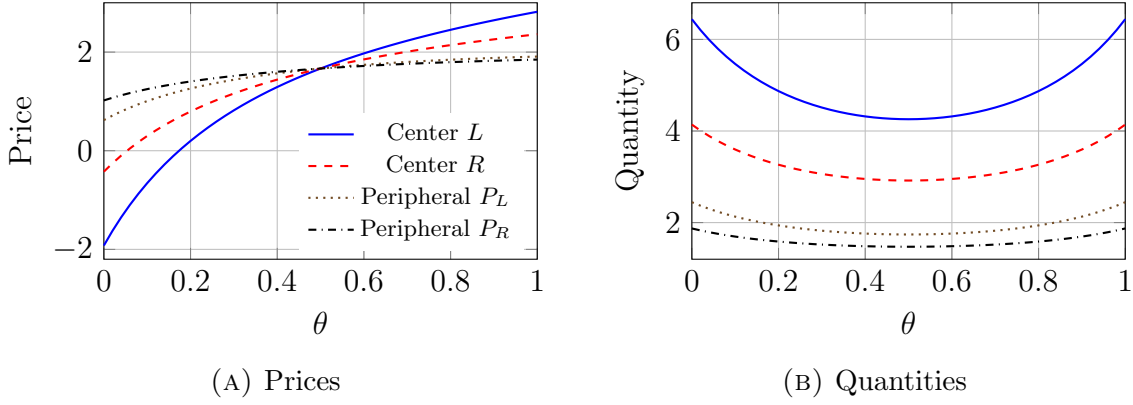
\begin{figure}[H]
\centering
\begin{subfigure}[b]{0.48\textwidth}
\centering
\begin{tikzpicture}
\begin{axis}[
	width=\textwidth,
	height=0.68\textwidth,
	xlabel={\(\theta\)},
	ylabel={Price},
	xmin=0, xmax=1,
	ymin=-2.2, ymax=3.0,
	xtick={0,0.2,0.4,0.6,0.8,1.0},
	grid=both,
	legend style={draw=none, at={(0.97,0.03)}, anchor=south east, font=\scriptsize}
]
\addplot+[mark=none, thick] coordinates {
	(0.00,-1.9224)
	(0.02,-1.6225)
	(0.04,-1.3492)
	(0.06,-1.0989)
	(0.08,-0.8689)
	(0.10,-0.6566)
	(0.12,-0.4602)
	(0.14,-0.2778)
	(0.16,-0.1080)
	(0.18,0.0506)
	(0.20,0.1990)
	(0.22,0.3382)
	(0.24,0.4691)
	(0.26,0.5925)
	(0.28,0.7089)
	(0.30,0.8191)
	(0.32,0.9235)
	(0.34,1.0225)
	(0.36,1.1167)
	(0.38,1.2064)
	(0.40,1.2919)
	(0.42,1.3735)
	(0.44,1.4516)
	(0.46,1.5263)
	(0.48,1.5979)
	(0.50,1.6667)
	(0.52,1.7327)
	(0.54,1.7962)
	(0.56,1.8574)
	(0.58,1.9164)
	(0.60,1.9732)
	(0.62,2.0282)
	(0.64,2.0813)
	(0.66,2.1327)
	(0.68,2.1824)
	(0.70,2.2307)
	(0.72,2.2775)
	(0.74,2.3229)
	(0.76,2.3670)
	(0.78,2.4100)
	(0.80,2.4517)
	(0.82,2.4924)
	(0.84,2.5320)
	(0.86,2.5707)
	(0.88,2.6084)
	(0.90,2.6453)
	(0.92,2.6813)
	(0.94,2.7165)
	(0.96,2.7510)
	(0.98,2.7848)
	(1.00,2.8178)
};
\addlegendentry{Center \(L\)}
\addplot+[mark=none, dashed, thick] coordinates {
	(0.00,-0.4245)
	(0.02,-0.2539)
	(0.04,-0.0979)
	(0.06,0.0453)
	(0.08,0.1774)
	(0.10,0.2995)
	(0.12,0.4129)
	(0.14,0.5185)
	(0.16,0.6171)
	(0.18,0.7093)
	(0.20,0.7959)
	(0.22,0.8774)
	(0.24,0.9542)
	(0.26,1.0268)
	(0.28,1.0954)
	(0.30,1.1605)
	(0.32,1.2223)
	(0.34,1.2811)
	(0.36,1.3372)
	(0.38,1.3906)
	(0.40,1.4417)
	(0.42,1.4905)
	(0.44,1.5373)
	(0.46,1.5822)
	(0.48,1.6253)
	(0.50,1.6667)
	(0.52,1.7065)
	(0.54,1.7448)
	(0.56,1.7818)
	(0.58,1.8174)
	(0.60,1.8519)
	(0.62,1.8851)
	(0.64,1.9173)
	(0.66,1.9485)
	(0.68,1.9786)
	(0.70,2.0079)
	(0.72,2.0362)
	(0.74,2.0637)
	(0.76,2.0905)
	(0.78,2.1165)
	(0.80,2.1417)
	(0.82,2.1663)
	(0.84,2.1902)
	(0.86,2.2135)
	(0.88,2.2362)
	(0.90,2.2584)
	(0.92,2.2800)
	(0.94,2.3011)
	(0.96,2.3217)
	(0.98,2.3418)
	(1.00,2.3614)
};
\addlegendentry{Center \(R\)}
\addplot+[mark=none, dotted, thick] coordinates {
	(0.00,0.6207)
	(0.02,0.7130)
	(0.04,0.7967)
	(0.06,0.8731)
	(0.08,0.9429)
	(0.10,1.0070)
	(0.12,1.0660)
	(0.14,1.1204)
	(0.16,1.1708)
	(0.18,1.2176)
	(0.20,1.2611)
	(0.22,1.3016)
	(0.24,1.3394)
	(0.26,1.3748)
	(0.28,1.4079)
	(0.30,1.4390)
	(0.32,1.4682)
	(0.34,1.4957)
	(0.36,1.5215)
	(0.38,1.5459)
	(0.40,1.5689)
	(0.42,1.5907)
	(0.44,1.6113)
	(0.46,1.6307)
	(0.48,1.6492)
	(0.50,1.6667)
	(0.52,1.6833)
	(0.54,1.6990)
	(0.56,1.7140)
	(0.58,1.7282)
	(0.60,1.7417)
	(0.62,1.7545)
	(0.64,1.7667)
	(0.66,1.7784)
	(0.68,1.7894)
	(0.70,1.8000)
	(0.72,1.8100)
	(0.74,1.8196)
	(0.76,1.8287)
	(0.78,1.8374)
	(0.80,1.8456)
	(0.82,1.8535)
	(0.84,1.8610)
	(0.86,1.8681)
	(0.88,1.8749)
	(0.90,1.8813)
	(0.92,1.8874)
	(0.94,1.8933)
	(0.96,1.8988)
	(0.98,1.9041)
	(1.00,1.9090)
};
\addlegendentry{Peripheral \(P_L\)}
\addplot+[mark=none, dashdotted, thick] coordinates {
	(0.00,1.0201)
	(0.02,1.0743)
	(0.04,1.1238)
	(0.06,1.1691)
	(0.08,1.2107)
	(0.10,1.2492)
	(0.12,1.2848)
	(0.14,1.3179)
	(0.16,1.3487)
	(0.18,1.3775)
	(0.20,1.4044)
	(0.22,1.4296)
	(0.24,1.4533)
	(0.26,1.4755)
	(0.28,1.4965)
	(0.30,1.5164)
	(0.32,1.5351)
	(0.34,1.5529)
	(0.36,1.5697)
	(0.38,1.5857)
	(0.40,1.6009)
	(0.42,1.6153)
	(0.44,1.6291)
	(0.46,1.6422)
	(0.48,1.6547)
	(0.50,1.6667)
	(0.52,1.6781)
	(0.54,1.6890)
	(0.56,1.6994)
	(0.58,1.7094)
	(0.60,1.7190)
	(0.62,1.7282)
	(0.64,1.7370)
	(0.66,1.7455)
	(0.68,1.7536)
	(0.70,1.7614)
	(0.72,1.7688)
	(0.74,1.7760)
	(0.76,1.7830)
	(0.78,1.7896)
	(0.80,1.7960)
	(0.82,1.8022)
	(0.84,1.8081)
	(0.86,1.8138)
	(0.88,1.8193)
	(0.90,1.8246)
	(0.92,1.8297)
	(0.94,1.8346)
	(0.96,1.8393)
	(0.98,1.8438)
	(1.00,1.8482)
};
\addlegendentry{Peripheral \(P_R\)}
\end{axis}
\end{tikzpicture}
\caption{Prices}
\end{subfigure}
\hfill
\begin{subfigure}[b]{0.48\textwidth}
\centering
\begin{tikzpicture}
\begin{axis}[
	width=\textwidth,
	height=0.68\textwidth,
	xlabel={\(\theta\)},
	ylabel={Quantity},
	xmin=0, xmax=1,
	ymin=1.2, ymax=6.8,
	xtick={0,0.2,0.4,0.6,0.8,1.0},
	grid=both
]
\addplot+[mark=none, thick] coordinates {
	(0.00,6.4417)
	(0.02,6.2035)
	(0.04,5.9895)
	(0.06,5.7965)
	(0.08,5.6222)
	(0.10,5.4645)
	(0.12,5.3216)
	(0.14,5.1920)
	(0.16,5.0744)
	(0.18,4.9677)
	(0.20,4.8710)
	(0.22,4.7835)
	(0.24,4.7044)
	(0.26,4.6332)
	(0.28,4.5694)
	(0.30,4.5124)
	(0.32,4.4618)
	(0.34,4.4175)
	(0.36,4.3789)
	(0.38,4.3460)
	(0.40,4.3185)
	(0.42,4.2962)
	(0.44,4.2790)
	(0.46,4.2668)
	(0.48,4.2595)
	(0.50,4.2570)
	(0.52,4.2595)
	(0.54,4.2668)
	(0.56,4.2790)
	(0.58,4.2962)
	(0.60,4.3185)
	(0.62,4.3460)
	(0.64,4.3789)
	(0.66,4.4175)
	(0.68,4.4618)
	(0.70,4.5124)
	(0.72,4.5694)
	(0.74,4.6332)
	(0.76,4.7044)
	(0.78,4.7835)
	(0.80,4.8710)
	(0.82,4.9677)
	(0.84,5.0744)
	(0.86,5.1920)
	(0.88,5.3216)
	(0.90,5.4645)
	(0.92,5.6222)
	(0.94,5.7965)
	(0.96,5.9895)
	(0.98,6.2035)
	(1.00,6.4417)
};
\addplot+[mark=none, dashed, thick] coordinates {
	(0.00,4.1403)
	(0.02,4.0078)
	(0.04,3.8887)
	(0.06,3.7812)
	(0.08,3.6840)
	(0.10,3.5959)
	(0.12,3.5161)
	(0.14,3.4436)
	(0.16,3.3777)
	(0.18,3.3180)
	(0.20,3.2637)
	(0.22,3.2146)
	(0.24,3.1702)
	(0.26,3.1302)
	(0.28,3.0943)
	(0.30,3.0622)
	(0.32,3.0338)
	(0.34,3.0088)
	(0.36,2.9871)
	(0.38,2.9685)
	(0.40,2.9530)
	(0.42,2.9404)
	(0.44,2.9307)
	(0.46,2.9238)
	(0.48,2.9197)
	(0.50,2.9183)
	(0.52,2.9197)
	(0.54,2.9238)
	(0.56,2.9307)
	(0.58,2.9404)
	(0.60,2.9530)
	(0.62,2.9685)
	(0.64,2.9871)
	(0.66,3.0088)
	(0.68,3.0338)
	(0.70,3.0622)
	(0.72,3.0943)
	(0.74,3.1302)
	(0.76,3.1702)
	(0.78,3.2146)
	(0.80,3.2637)
	(0.82,3.3180)
	(0.84,3.3777)
	(0.86,3.4436)
	(0.88,3.5161)
	(0.90,3.5959)
	(0.92,3.6840)
	(0.94,3.7812)
	(0.96,3.8887)
	(0.98,4.0078)
	(1.00,4.1403)
};
\addplot+[mark=none, dotted, thick] coordinates {
	(0.00,2.4451)
	(0.02,2.3684)
	(0.04,2.2994)
	(0.06,2.2372)
	(0.08,2.1810)
	(0.10,2.1301)
	(0.12,2.0840)
	(0.14,2.0422)
	(0.16,2.0043)
	(0.18,1.9698)
	(0.20,1.9386)
	(0.22,1.9104)
	(0.24,1.8848)
	(0.26,1.8619)
	(0.28,1.8412)
	(0.30,1.8228)
	(0.32,1.8065)
	(0.34,1.7922)
	(0.36,1.7797)
	(0.38,1.7691)
	(0.40,1.7602)
	(0.42,1.7530)
	(0.44,1.7474)
	(0.46,1.7434)
	(0.48,1.7411)
	(0.50,1.7403)
	(0.52,1.7411)
	(0.54,1.7434)
	(0.56,1.7474)
	(0.58,1.7530)
	(0.60,1.7602)
	(0.62,1.7691)
	(0.64,1.7797)
	(0.66,1.7922)
	(0.68,1.8065)
	(0.70,1.8228)
	(0.72,1.8412)
	(0.74,1.8619)
	(0.76,1.8848)
	(0.78,1.9104)
	(0.80,1.9386)
	(0.82,1.9698)
	(0.84,2.0043)
	(0.86,2.0422)
	(0.88,2.0840)
	(0.90,2.1301)
	(0.92,2.1810)
	(0.94,2.2372)
	(0.96,2.2994)
	(0.98,2.3684)
	(1.00,2.4451)
};
\addplot+[mark=none, dashdotted, thick] coordinates {
	(0.00,1.8719)
	(0.02,1.8289)
	(0.04,1.7901)
	(0.06,1.7550)
	(0.08,1.7233)
	(0.10,1.6946)
	(0.12,1.6685)
	(0.14,1.6448)
	(0.16,1.6233)
	(0.18,1.6037)
	(0.20,1.5860)
	(0.22,1.5699)
	(0.24,1.5553)
	(0.26,1.5422)
	(0.28,1.5304)
	(0.30,1.5199)
	(0.32,1.5105)
	(0.34,1.5023)
	(0.36,1.4952)
	(0.38,1.4891)
	(0.40,1.4840)
	(0.42,1.4798)
	(0.44,1.4766)
	(0.46,1.4744)
	(0.48,1.4730)
	(0.50,1.4726)
	(0.52,1.4730)
	(0.54,1.4744)
	(0.56,1.4766)
	(0.58,1.4798)
	(0.60,1.4840)
	(0.62,1.4891)
	(0.64,1.4952)
	(0.66,1.5023)
	(0.68,1.5105)
	(0.70,1.5199)
	(0.72,1.5304)
	(0.74,1.5422)
	(0.76,1.5553)
	(0.78,1.5699)
	(0.80,1.5860)
	(0.82,1.6037)
	(0.84,1.6233)
	(0.86,1.6448)
	(0.88,1.6685)
	(0.90,1.6946)
	(0.92,1.7233)
	(0.94,1.7550)
	(0.96,1.7901)
	(0.98,1.8289)
	(1.00,1.8719)
};
\end{axis}
\end{tikzpicture}
\caption{Quantities}
\end{subfigure}
\caption{Prices and Quantities in the Double-Star Network, $\beta=0.5$}
\label{fig:double_star_price_quantity}
\end{figure}

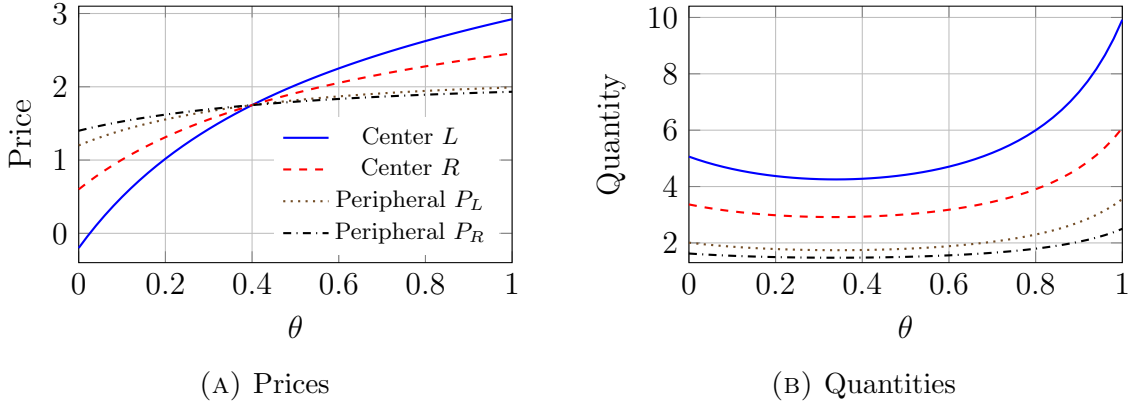
\begin{figure}[H]
\centering
\begin{subfigure}[b]{0.48\textwidth}
\centering
\begin{tikzpicture}
\begin{axis}[
	width=\textwidth,
	height=0.68\textwidth,
	xlabel={\(\theta\)},
	ylabel={Price},
	xmin=0, xmax=1,
	ymin=-0.4, ymax=3.1,
	xtick={0,0.2,0.4,0.6,0.8,1.0},
	grid=both,
	legend style={draw=none, at={(0.97,0.03)}, anchor=south east, font=\scriptsize}
]
\addplot+[mark=none, thick] coordinates {
	(0.00,-0.2000)
	(0.02,-0.0422)
	(0.04,0.1054)
	(0.06,0.2439)
	(0.08,0.3740)
	(0.10,0.4967)
	(0.12,0.6124)
	(0.14,0.7218)
	(0.16,0.8254)
	(0.18,0.9237)
	(0.20,1.0172)
	(0.22,1.1061)
	(0.24,1.1908)
	(0.26,1.2717)
	(0.28,1.3490)
	(0.30,1.4229)
	(0.32,1.4938)
	(0.34,1.5617)
	(0.36,1.6270)
	(0.38,1.6897)
	(0.40,1.7500)
	(0.42,1.8081)
	(0.44,1.8642)
	(0.46,1.9182)
	(0.48,1.9705)
	(0.50,2.0210)
	(0.52,2.0698)
	(0.54,2.1171)
	(0.56,2.1630)
	(0.58,2.2075)
	(0.60,2.2507)
	(0.62,2.2926)
	(0.64,2.3334)
	(0.66,2.3731)
	(0.68,2.4117)
	(0.70,2.4493)
	(0.72,2.4860)
	(0.74,2.5218)
	(0.76,2.5567)
	(0.78,2.5908)
	(0.80,2.6242)
	(0.82,2.6568)
	(0.84,2.6887)
	(0.86,2.7199)
	(0.88,2.7505)
	(0.90,2.7805)
	(0.92,2.8099)
	(0.94,2.8388)
	(0.96,2.8671)
	(0.98,2.8950)
	(1.00,2.9223)
};
\addlegendentry{Center \(L\)}
\addplot+[mark=none, dashed, thick] coordinates {
	(0.00,0.6000)
	(0.02,0.6914)
	(0.04,0.7771)
	(0.06,0.8577)
	(0.08,0.9337)
	(0.10,1.0054)
	(0.12,1.0733)
	(0.14,1.1376)
	(0.16,1.1987)
	(0.18,1.2568)
	(0.20,1.3121)
	(0.22,1.3648)
	(0.24,1.4152)
	(0.26,1.4633)
	(0.28,1.5094)
	(0.30,1.5536)
	(0.32,1.5960)
	(0.34,1.6368)
	(0.36,1.6760)
	(0.38,1.7137)
	(0.40,1.7500)
	(0.42,1.7850)
	(0.44,1.8189)
	(0.46,1.8515)
	(0.48,1.8831)
	(0.50,1.9137)
	(0.52,1.9432)
	(0.54,1.9719)
	(0.56,1.9997)
	(0.58,2.0267)
	(0.60,2.0528)
	(0.62,2.0783)
	(0.64,2.1030)
	(0.66,2.1270)
	(0.68,2.1504)
	(0.70,2.1732)
	(0.72,2.1954)
	(0.74,2.2170)
	(0.76,2.2381)
	(0.78,2.2586)
	(0.80,2.2787)
	(0.82,2.2983)
	(0.84,2.3175)
	(0.86,2.3362)
	(0.88,2.3545)
	(0.90,2.3724)
	(0.92,2.3899)
	(0.94,2.4070)
	(0.96,2.4238)
	(0.98,2.4402)
	(1.00,2.4563)
};
\addlegendentry{Center \(R\)}
\addplot+[mark=none, dotted, thick] coordinates {
	(0.00,1.2000)
	(0.02,1.2470)
	(0.04,1.2908)
	(0.06,1.3316)
	(0.08,1.3698)
	(0.10,1.4055)
	(0.12,1.4389)
	(0.14,1.4703)
	(0.16,1.4998)
	(0.18,1.5276)
	(0.20,1.5539)
	(0.22,1.5786)
	(0.24,1.6020)
	(0.26,1.6241)
	(0.28,1.6450)
	(0.30,1.6649)
	(0.32,1.6837)
	(0.34,1.7016)
	(0.36,1.7185)
	(0.38,1.7347)
	(0.40,1.7500)
	(0.42,1.7646)
	(0.44,1.7785)
	(0.46,1.7918)
	(0.48,1.8044)
	(0.50,1.8164)
	(0.52,1.8279)
	(0.54,1.8389)
	(0.56,1.8493)
	(0.58,1.8593)
	(0.60,1.8689)
	(0.62,1.8780)
	(0.64,1.8867)
	(0.66,1.8950)
	(0.68,1.9029)
	(0.70,1.9105)
	(0.72,1.9178)
	(0.74,1.9247)
	(0.76,1.9313)
	(0.78,1.9376)
	(0.80,1.9436)
	(0.82,1.9494)
	(0.84,1.9549)
	(0.86,1.9601)
	(0.88,1.9651)
	(0.90,1.9698)
	(0.92,1.9743)
	(0.94,1.9786)
	(0.96,1.9827)
	(0.98,1.9866)
	(1.00,1.9903)
};
\addlegendentry{Peripheral \(P_L\)}
\addplot+[mark=none, dashdotted, thick] coordinates {
	(0.00,1.4000)
	(0.02,1.4286)
	(0.04,1.4554)
	(0.06,1.4805)
	(0.08,1.5041)
	(0.10,1.5263)
	(0.12,1.5472)
	(0.14,1.5670)
	(0.16,1.5857)
	(0.18,1.6034)
	(0.20,1.6202)
	(0.22,1.6362)
	(0.24,1.6513)
	(0.26,1.6658)
	(0.28,1.6795)
	(0.30,1.6926)
	(0.32,1.7052)
	(0.34,1.7171)
	(0.36,1.7286)
	(0.38,1.7395)
	(0.40,1.7500)
	(0.42,1.7600)
	(0.44,1.7697)
	(0.46,1.7789)
	(0.48,1.7878)
	(0.50,1.7963)
	(0.52,1.8045)
	(0.54,1.8124)
	(0.56,1.8199)
	(0.58,1.8272)
	(0.60,1.8342)
	(0.62,1.8410)
	(0.64,1.8475)
	(0.66,1.8538)
	(0.68,1.8598)
	(0.70,1.8656)
	(0.72,1.8713)
	(0.74,1.8767)
	(0.76,1.8819)
	(0.78,1.8869)
	(0.80,1.8918)
	(0.82,1.8965)
	(0.84,1.9010)
	(0.86,1.9054)
	(0.88,1.9096)
	(0.90,1.9137)
	(0.92,1.9176)
	(0.94,1.9214)
	(0.96,1.9251)
	(0.98,1.9286)
	(1.00,1.9320)
};
\addlegendentry{Peripheral \(P_R\)}
\end{axis}
\end{tikzpicture}
\caption{Prices}
\end{subfigure}
\hfill
\begin{subfigure}[b]{0.48\textwidth}
\centering
\begin{tikzpicture}
\begin{axis}[
	width=\textwidth,
	height=0.68\textwidth,
	xlabel={\(\theta\)},
	ylabel={Quantity},
	xmin=0, xmax=1,
	ymin=1.3, ymax=10.4,
	xtick={0,0.2,0.4,0.6,0.8,1.0},
	grid=both
]
\addplot+[mark=none, thick] coordinates {
	(0.00,5.0624)
	(0.02,4.9563)
	(0.04,4.8601)
	(0.06,4.7731)
	(0.08,4.6946)
	(0.10,4.6239)
	(0.12,4.5605)
	(0.14,4.5040)
	(0.16,4.4539)
	(0.18,4.4100)
	(0.20,4.3719)
	(0.22,4.3394)
	(0.24,4.3123)
	(0.26,4.2904)
	(0.28,4.2736)
	(0.30,4.2618)
	(0.32,4.2549)
	(0.34,4.2529)
	(0.36,4.2557)
	(0.38,4.2634)
	(0.40,4.2760)
	(0.42,4.2936)
	(0.44,4.3164)
	(0.46,4.3443)
	(0.48,4.3776)
	(0.50,4.4166)
	(0.52,4.4614)
	(0.54,4.5123)
	(0.56,4.5698)
	(0.58,4.6341)
	(0.60,4.7057)
	(0.62,4.7853)
	(0.64,4.8733)
	(0.66,4.9705)
	(0.68,5.0777)
	(0.70,5.1958)
	(0.72,5.3260)
	(0.74,5.4696)
	(0.76,5.6280)
	(0.78,5.8030)
	(0.80,5.9967)
	(0.82,6.2116)
	(0.84,6.4507)
	(0.86,6.7176)
	(0.88,7.0168)
	(0.90,7.3536)
	(0.92,7.7349)
	(0.94,8.1692)
	(0.96,8.6674)
	(0.98,9.2439)
	(1.00,9.9173)
};
\addplot+[mark=none, dashed, thick] coordinates {
	(0.00,3.3653)
	(0.02,3.3062)
	(0.04,3.2527)
	(0.06,3.2043)
	(0.08,3.1606)
	(0.10,3.1213)
	(0.12,3.0860)
	(0.14,3.0546)
	(0.16,3.0268)
	(0.18,3.0025)
	(0.20,2.9814)
	(0.22,2.9635)
	(0.24,2.9486)
	(0.26,2.9367)
	(0.28,2.9276)
	(0.30,2.9213)
	(0.32,2.9178)
	(0.34,2.9171)
	(0.36,2.9191)
	(0.38,2.9238)
	(0.40,2.9314)
	(0.42,2.9417)
	(0.44,2.9549)
	(0.46,2.9711)
	(0.48,2.9903)
	(0.50,3.0127)
	(0.52,3.0383)
	(0.54,3.0674)
	(0.56,3.1002)
	(0.58,3.1368)
	(0.60,3.1775)
	(0.62,3.2226)
	(0.64,3.2724)
	(0.66,3.3274)
	(0.68,3.3879)
	(0.70,3.4545)
	(0.72,3.5278)
	(0.74,3.6085)
	(0.76,3.6974)
	(0.78,3.7955)
	(0.80,3.9040)
	(0.82,4.0241)
	(0.84,4.1575)
	(0.86,4.3063)
	(0.88,4.4728)
	(0.90,4.6600)
	(0.92,4.8716)
	(0.94,5.1123)
	(0.96,5.3881)
	(0.98,5.7067)
	(1.00,6.0785)
};
\addplot+[mark=none, dotted, thick] coordinates {
	(0.00,2.0089)
	(0.02,1.9743)
	(0.04,1.9429)
	(0.06,1.9145)
	(0.08,1.8888)
	(0.10,1.8656)
	(0.12,1.8447)
	(0.14,1.8261)
	(0.16,1.8096)
	(0.18,1.7951)
	(0.20,1.7824)
	(0.22,1.7716)
	(0.24,1.7625)
	(0.26,1.7551)
	(0.28,1.7493)
	(0.30,1.7451)
	(0.32,1.7426)
	(0.34,1.7416)
	(0.36,1.7422)
	(0.38,1.7443)
	(0.40,1.7481)
	(0.42,1.7534)
	(0.44,1.7604)
	(0.46,1.7691)
	(0.48,1.7796)
	(0.50,1.7918)
	(0.52,1.8060)
	(0.54,1.8221)
	(0.56,1.8403)
	(0.58,1.8607)
	(0.60,1.8835)
	(0.62,1.9089)
	(0.64,1.9370)
	(0.66,1.9680)
	(0.68,2.0023)
	(0.70,2.0401)
	(0.72,2.0817)
	(0.74,2.1277)
	(0.76,2.1784)
	(0.78,2.2345)
	(0.80,2.2966)
	(0.82,2.3655)
	(0.84,2.4421)
	(0.86,2.5278)
	(0.88,2.6237)
	(0.90,2.7318)
	(0.92,2.8542)
	(0.94,2.9936)
	(0.96,3.1535)
	(0.98,3.3385)
	(1.00,3.5547)
};
\addplot+[mark=none, dashdotted, thick] coordinates {
	(0.00,1.6236)
	(0.02,1.6042)
	(0.04,1.5866)
	(0.06,1.5706)
	(0.08,1.5561)
	(0.10,1.5431)
	(0.12,1.5315)
	(0.14,1.5210)
	(0.16,1.5118)
	(0.18,1.5037)
	(0.20,1.4967)
	(0.22,1.4907)
	(0.24,1.4857)
	(0.26,1.4816)
	(0.28,1.4785)
	(0.30,1.4764)
	(0.32,1.4751)
	(0.34,1.4747)
	(0.36,1.4753)
	(0.38,1.4768)
	(0.40,1.4791)
	(0.42,1.4824)
	(0.44,1.4867)
	(0.46,1.4919)
	(0.48,1.4981)
	(0.50,1.5053)
	(0.52,1.5137)
	(0.54,1.5231)
	(0.56,1.5338)
	(0.58,1.5457)
	(0.60,1.5590)
	(0.62,1.5737)
	(0.64,1.5899)
	(0.66,1.6078)
	(0.68,1.6275)
	(0.70,1.6492)
	(0.72,1.6731)
	(0.74,1.6994)
	(0.76,1.7283)
	(0.78,1.7602)
	(0.80,1.7954)
	(0.82,1.8345)
	(0.84,1.8778)
	(0.86,1.9261)
	(0.88,1.9801)
	(0.90,2.0408)
	(0.92,2.1093)
	(0.94,2.1873)
	(0.96,2.2765)
	(0.98,2.3795)
	(1.00,2.4995)
};
\end{axis}
\end{tikzpicture}
\caption{Quantities}
\end{subfigure}
\caption{Prices and Quantities in the Double-Star Network, $\beta=0.4$}
\label{fig:double_star_price_quantity_beta04}
\end{figure}
\end{example}

\section{Designing Interoperability}
\label{sec:designing_interoperability}

We now interpret interoperability as a design parameter. A regulator or platform designer may choose the strength of cross-platform interoperability, \(\theta\), taking the consumer network and the degree of product substitutability as given. This section summarizes the preferences of consumers, platforms, and a welfare-maximizing designer over \(\theta\).

The analysis also clarifies an important implication of the pricing formula: stronger interoperability tends to raise equilibrium prices. This effect is exact in regular networks and holds locally in general networks.

\begin{proposition}[Interoperability and equilibrium prices]
\label{prop:price_theta_design}
Suppose Assumptions \ref{ass:stability} and \ref{ass:full_symmetry} hold. Let
\[
	\bm V_\theta=(2-\beta)\bm I-\delta(2-\theta)\bm G .
\]
Then the exact derivative of the discriminatory price vector with respect to interoperability is
\begin{align}
\frac{\partial \bm p^*}{\partial\theta}
=
(a-c)\delta
\bm G(\bm I-\delta\bm G)\bm V_\theta^{-2}\bm 1 .
\label{eq:price_theta_derivative_general}
\end{align}
In a \(d\)-regular network, this derivative reduces to
\begin{align}
\frac{\partial p^*}{\partial\theta}
=
\frac{
	\delta d(1-\delta d)
}{
	\left[(2-\beta)-\delta d(2-\theta)\right]^2
}
(a-c)
>0.
\label{eq:price_theta_derivative_regular}
\end{align}
Moreover, in a general network, for sufficiently small \(\delta\),
\begin{align}
\frac{\partial p_i^*}{\partial\theta}
=
\frac{\delta(a-c)}{(2-\beta)^2}d_i
+
O(\delta^2).
\label{eq:price_theta_derivative_local}
\end{align}
Thus, to the first order, stronger interoperability raises prices more for consumers with higher degree.
\end{proposition}

Proposition \ref{prop:price_theta_design} shows that the centrality-premium mechanism is not merely a comparison between two fixed values of \(\theta\). As interoperability increases, the price charged to more connected consumers tends to increase more strongly. In regular networks this price effect is globally positive. In nonregular networks, the local approximation shows that the incidence of the price increase is ordered by degree.

In nonregular networks, interoperability may also reallocate pricing pressure across consumers. Rather than raising all individual prices uniformly, stronger interoperability can increase the prices charged to central consumers while reducing those charged to peripheral consumers. Hence, interoperability is not only a competition-softening force; it also changes the incidence of prices across network positions.

We next consider the design of \(\theta\) in a regular network. Let
\[
	q=\delta d.
\]
To keep the feasible set fixed as \(\theta\) varies over \([0,1]\), suppose that the stability condition holds for every \(\theta\in[0,1]\), which is guaranteed by
\[
	q<\min\left\{1-\beta,\frac{1+\beta}{2}\right\}.
\]
Define the consumer-optimal, platform-optimal, and welfare-optimal interoperability levels by
\[
	\theta^{CS}
	\in
	\operatorname*{arg\,max}_{\theta\in[0,1]} CS^*(\theta),
\]
\[
	\theta^{\Pi}
	\in
	\operatorname*{arg\,max}_{\theta\in[0,1]} \Pi^*(\theta),
\]
and
\[
	\theta^{TW}
	\in
	\operatorname*{arg\,max}_{\theta\in[0,1]} TW^*(\theta).
\]

\begin{proposition}[Designing interoperability in regular networks]
\label{prop:optimal_interoperability_regular}
Suppose Assumptions \ref{ass:stability} and \ref{ass:full_symmetry} hold for all \(\theta\in[0,1]\), and suppose that \(\bm G\) is \(d\)-regular with \(q=\delta d>0\). Then:

\begin{enumerate}[label=(\roman*)]
	\item Platform profit is strictly increasing in \(\theta\). Hence,
	\[
		\theta^{\Pi}=1.
	\]

	\item Consumer surplus is maximized at complete interoperability if products are sufficiently differentiated:
	\[
		\beta<\frac{1}{2}
		\quad\Longrightarrow\quad
		\theta^{CS}=1.
	\]
	It is maximized at no interoperability if products are sufficiently close substitutes:
	\[
		\beta>\frac{1}{2}
		\quad\Longrightarrow\quad
		\theta^{CS}=0.
	\]
	If \(\beta=1/2\), then consumers are indifferent between the two endpoints:
	\[
		CS^*(0)=CS^*(1).
	\]

	\item A welfare-optimal interoperability level satisfies
	\[
		\theta^{TW}
		\in
		\operatorname*{arg\,max}_{\theta\in[0,1]} TW^*(\theta),
	\]
	where an interior optimum must solve
	\begin{align}
	\underbrace{s(1+\beta)(v-m^+)}_{\text{consumer-surplus effect}}
	+
	\underbrace{2m^+\left[s^2-m^-s+(m^-)^2\right]}_{\text{profit effect}}
	=
	0,
	\label{eq:theta_tw_foc}
	\end{align}
	with
	\[
		s=1-q,
		\qquad
		m^+=(1+\beta)-q(1+\theta),
	\]
	\[
		m^-=(1-\beta)-q(1-\theta),
		\qquad
		v=(2-\beta)-q(2-\theta).
	\]
	In particular, if
	\[
		\beta\leq\frac{1-q}{2},
	\]
	then consumer surplus and platform profits both increase with \(\theta\), so
	\[
		\theta^{TW}=1.
	\]
\end{enumerate}
\end{proposition}

Proposition \ref{prop:optimal_interoperability_regular} highlights the design trade-off. Platforms always prefer stronger interoperability in the regular-network benchmark because interoperability softens competition and raises markups. Consumers prefer complete interoperability when products are sufficiently differentiated, since cross-platform network benefits dominate the price increase. When products are close substitutes, however, consumers may prefer no interoperability, because the competition-softening effect dominates.

The welfare-optimal level of interoperability balances these two forces. If interoperability raises consumer surplus, then both consumers and platforms favor higher \(\theta\), and complete interoperability is optimal. If interoperability lowers consumer surplus, the welfare comparison becomes nontrivial: stronger interoperability raises platform profits but may reduce consumer surplus. In that case, the welfare-optimal \(\theta\) is characterized by \eqref{eq:theta_tw_foc} together with endpoint comparisons.

This design perspective connects the pricing results to interoperability policy. A policy that increases interoperability does not simply create more network value. It also changes the intensity of price competition and the distribution of surplus between consumers and platforms. The socially desirable level of interoperability therefore depends on product substitutability, network strength, and the weight placed on consumer surplus relative to platform profits.

\section{Network Formation and Platform Incentives}
\label{sec:network_formation}

This section studies how interoperability affects consumers' and platforms' preferences over the density of the consumer network. The analysis follows the logic of \citet{ChenZenouZhou2018RAND}: stronger social connections increase consumers' direct network benefits, but they may either raise or reduce platform profits depending on how they affect the intensity of competition.

Let \(\mathcal G_n\) denote a class of feasible undirected networks on \(N\). For two networks \(\bm G,\bm G'\in\mathcal G_n\), write
\[
	\bm G'\succeq \bm G
\]
if \(g'_{ij}\geq g_{ij}\) for all \(i,j\in N\). Let \(\bm E_n\) denote the empty network and \(\bm P_n\) denote the complete network in \(\mathcal G_n\). Thus, \(\bm E_n\preceq \bm G\preceq \bm P_n\) for every \(\bm G\in\mathcal G_n\). We restrict attention to networks for which Assumption \ref{ass:stability} holds.

For consumers, denser networks are always beneficial in the present linear-quadratic environment. The reason is that social links directly increase the value of consumption and, through equilibrium pricing, do not overturn this positive consumption effect.

\begin{proposition}[Consumer preferences over network density]
\label{prop:consumer_network}
Suppose Assumptions \ref{ass:stability} and \ref{ass:platform_symmetry} hold. If \(\bm G'\succeq \bm G\), then
\[
	\bm x^*(\bm G';\beta,\delta,\theta)
	\geq
	\bm x^*(\bm G;\beta,\delta,\theta),
\]
and hence
\[
	CS^*(\bm G';\beta,\delta,\theta)
	\geq
	CS^*(\bm G;\beta,\delta,\theta).
\]
Consequently, consumers prefer the complete network:
\[
	\bm P_n
	\in
	\bm G^{CS}(\beta,\delta,\theta)
	:=
	\operatorname*{arg\,max}_{\bm G\in\mathcal G_n}
	CS^*(\bm G;\beta,\delta,\theta).
\]
\end{proposition}

Proposition \ref{prop:consumer_network} shows that interoperability does not overturn consumers' preference for denser social connections. Although interoperability changes equilibrium prices, a denser network still increases equilibrium consumption and consumer surplus. This result is consistent with the consumer-side conclusion in \citet{ChenZenouZhou2018RAND}.

The platform side is more subtle. Denser networks increase consumers' willingness to consume, but they may also intensify competition for influential consumers. Interoperability changes this trade-off because it reduces the exclusivity of network benefits. To isolate this mechanism, we consider the small-\(\delta\) expansion of platform profits. Let
\[
	\bm r=\bm a-\bm c.
\]
Under Assumptions \ref{ass:stability} and \ref{ass:platform_symmetry}, each platform's equilibrium profit admits the expansion
\begin{align}
\Pi^*(\bm G;\beta,\delta,\theta)
=
\frac{1-\beta}{(1+\beta)(2-\beta)^2}\bm r^\top\bm r
+
\frac{\delta}{(1+\beta)^2(2-\beta)^3}
\mathcal Q(\beta,\theta)\bm r^\top\bm G\bm r
+
\mathcal O(\delta^2),
\label{eq:profit_expansion}
\end{align}
where
\begin{align}
\mathcal Q(\beta,\theta)
=
\underbrace{2-3\beta-\beta^3}_{\text{substitution effect}}
+
\underbrace{2\theta(1-\beta+\beta^2)}_{\text{interoperability effect}} .
\label{eq:Q_function}
\end{align}

The first-order effect of network density on platform profits is governed by \(\mathcal Q(\beta,\theta)\). The term \(2-3\beta-\beta^3\) is decreasing in product substitutability and captures the competition-intensifying effect of denser networks. When products are close substitutes, denser networks make central consumers more valuable competitive targets, which can reduce platform profits. The term \(2\theta(1-\beta+\beta^2)\) captures the effect of interoperability. Greater interoperability makes network benefits less exclusive to either platform, mitigating the competitive cost of network expansion.

\begin{proposition}[Platform incentives over network density]
\label{prop:platform_network}
Suppose Assumptions \ref{ass:stability} and \ref{ass:platform_symmetry} hold. For any two networks \(\bm G,\bm G'\in\mathcal G_n\) such that \(\bm G'\succeq \bm G\) and \(\bm r^\top(\bm G'-\bm G)\bm r>0\), we have, for sufficiently small \(\delta>0\),
\[
	\operatorname{sign}
	\left\{
		\Pi^*(\bm G';\beta,\delta,\theta)
		-
		\Pi^*(\bm G;\beta,\delta,\theta)
	\right\}
	=
	\operatorname{sign}
	\left\{
		\mathcal Q(\beta,\theta)
	\right\}.
\]
Consequently, if
\[
	\mathcal Q(\beta,\theta)>0,
\]
then platforms prefer denser networks in the order \(\succeq\), and in particular prefer the complete network \(\bm P_n\). If
\[
	\mathcal Q(\beta,\theta)<0,
\]
then platforms prefer sparser networks in the order \(\succeq\), and in particular prefer the empty network \(\bm E_n\).
\end{proposition}

Equivalently, define the interoperability threshold
\begin{align}
	\bar\theta(\beta)
	=
	\frac{\beta^3+3\beta-2}{2(1-\beta+\beta^2)}.
\label{eq:theta_bar_network}
\end{align}
Then
\[
	\mathcal Q(\beta,\theta)>0
	\quad\Longleftrightarrow\quad
	\theta>\bar\theta(\beta).
\]
Let \(\hat\beta\) denote the unique root of
\[
	2-3\beta-\beta^3=0,
\]
so that \(\hat\beta\simeq0.596\). When \(\beta<\hat\beta\), we have \(\bar\theta(\beta)<0\), and platforms have incentives to support dense networks even without interoperability. When \(\beta>\hat\beta\), platforms may prefer sparse networks if interoperability is low, but a sufficiently high level of interoperability reverses this incentive. In particular, when \(\theta=1\),
\[
	\mathcal Q(\beta,1)>0
	\quad
	\text{for all } \beta\in[0,1),
\]
so fully interoperable platforms prefer denser consumer networks in the small-\(\delta\) region.

\begin{figure}[H]
	\centering
	\includegraphics[width=0.78\textwidth]{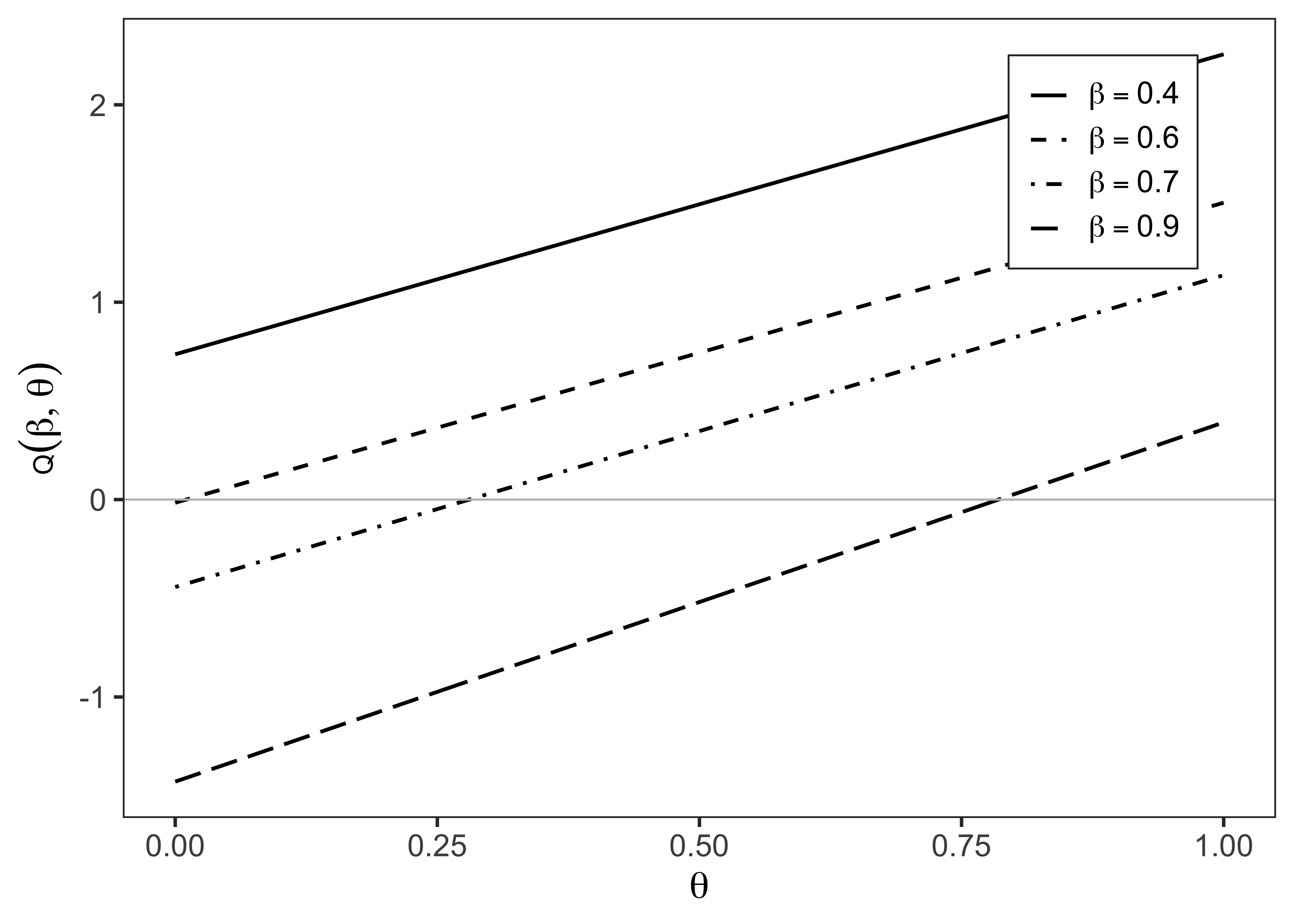}
	\caption{Interoperability and Network Formation Incentives}
	\label{fig:network_formation}
\end{figure}

Figure \ref{fig:network_formation} plots \(\mathcal Q(\beta,\theta)\) against interoperability \(\theta\) for different values of product substitutability \(\beta\). For each \(\beta\), \(\mathcal Q(\beta,\theta)\) increases linearly in \(\theta\), showing that interoperability systematically raises the profitability of network expansion. For relatively differentiated products, such as \(\beta=0.4\), \(\mathcal Q(\beta,\theta)\) is positive for all \(\theta\in[0,1]\), so platforms prefer denser consumer networks even without interoperability. For more substitutable products, such as \(\beta=0.9\), \(\mathcal Q(\beta,\theta)\) is negative when interoperability is low but becomes positive once interoperability is sufficiently high.

The economic message is that interoperability can change firms' attitudes toward the consumer network itself. Without interoperability, a denser network may intensify competition for central consumers and reduce profits when products are close substitutes. With strong interoperability, the same network becomes a source of shared cross-platform value. Thus, interoperability can transform platforms from suppressors of consumer connections into promoters of network formation.

\section{Uniform versus Discriminatory Pricing}
\label{sec:uniform_discriminatory}

The analysis so far allows platforms to charge consumer-specific prices. We now study the consequences of restricting platforms to uniform pricing. A uniform-pricing constraint means that each platform must charge the same price to all consumers, but it does not require the two platforms to charge the same price. Thus, platform \(A\) chooses a scalar \(p_u^A\) and charges \(p_u^A\bm 1\), while platform \(B\) chooses a scalar \(p_u^B\) and charges \(p_u^B\bm 1\).

This section compares the uniform-pricing regime with the discriminatory-pricing regime characterized in Section \ref{sec:equilibrium_pricing}. The comparison is useful for two reasons. First, it shows how the ability to price discriminate transforms network centrality into consumer-specific discounts or premia. Second, it clarifies the distributional consequences of banning third-degree price discrimination in interoperable platform markets. We proceed in four steps. We first derive the equilibrium under uniform pricing. We then identify benchmark cases in which uniform and discriminatory pricing are equivalent. Next, we compare prices and individual consumer surplus across network positions. Finally, we compare aggregate consumer surplus, platform profits, and total welfare using local approximations.

\subsection{Uniform-Pricing Equilibrium}
\label{subsec:uniform_equilibrium}

Throughout this subsection, suppose that platforms are restricted to charge uniform prices across consumers. For any pair of uniform prices \(p_u^A\) and \(p_u^B\), Proposition \ref{prop:consumption} gives
\begin{align*}
\bm x_u^A
&=
\frac{\bm M^++\bm M^-}{2}(\bm a-p_u^A\bm 1)
+
\frac{\bm M^+-\bm M^-}{2}(\bm a-p_u^B\bm 1),
\\
\bm x_u^B
&=
\frac{\bm M^++\bm M^-}{2}(\bm a-p_u^B\bm 1)
+
\frac{\bm M^+-\bm M^-}{2}(\bm a-p_u^A\bm 1),
\end{align*}
where we have imposed Assumption \ref{ass:platform_symmetry}. Although the uniform-pricing constraint does not require \(p_u^A=p_u^B\), symmetry across platforms implies that the equilibrium uniform prices are equal.

\begin{theorem}[Uniform-pricing equilibrium]
\label{thm:uniform_pricing}
Suppose Assumptions \ref{ass:stability} and \ref{ass:platform_symmetry} hold, and suppose both platforms are restricted to charge uniform prices to all consumers. Then the uniform-pricing game admits a unique symmetric equilibrium satisfying $p_u^A=p_u^B=p^u$ and the equilibrium uniform price is
\begin{align}
p^u
=
\frac{
	\langle \bm 1,\bm M^+(2\bm a+\bm c)\rangle
	+
	\langle \bm 1,\bm M^-\bm c\rangle
}{
	\langle \bm 1,(3\bm M^++\bm M^-)\bm 1\rangle
}.
\label{eq:uniform_price_general}
\end{align}
The corresponding equilibrium quantities satisfy
\[
	\bm x_u^A=\bm x_u^B=\bm x^u=\bm M^+(\bm a-p^u\bm 1).
\]
If Assumption \ref{ass:full_symmetry} also holds, then \eqref{eq:uniform_price_general} reduces to
\begin{align}
p^u
=
c
+
2(a-c)
\frac{
	\langle \bm 1,\bm M^+\bm 1\rangle
}{
	\langle \bm 1,(3\bm M^++\bm M^-)\bm 1\rangle
}.
\label{eq:uniform_price_full_symmetry}
\end{align}
\end{theorem}

Theorem \ref{thm:uniform_pricing} shows that uniform pricing does not remove the network from the pricing problem. Even though platforms cannot condition prices on individual network positions, the equilibrium uniform price depends on the aggregate network terms
\[
	\langle \bm 1,\bm M^+\bm 1\rangle
	\quad\text{and}\quad
	\langle \bm 1,\bm M^-\bm 1\rangle.
\]
The first term captures how aggregate consumption is amplified by within-platform and cross-platform network externalities. The second captures the effect of substitution between the two platforms. Thus, under uniform pricing, the network affects prices through aggregate interaction intensity rather than through consumer-specific centrality.

To obtain a transparent expression, let
\[
	\bm d=\bm G\bm 1,
	\qquad
	d_i=\sum_{j=1}^n g_{ij},
	\qquad
	\bar d=\frac{1}{n}\bm 1^\top\bm G\bm 1
	=
	\frac{1}{n}\sum_{i=1}^n d_i.
\]
Thus, \(d_i\) is consumer \(i\)'s degree and \(\bar d\) is the average degree of the network.

\begin{corollary}[Local approximation of the uniform price]
\label{cor:uniform_price_approx}
Suppose Assumptions \ref{ass:stability} and \ref{ass:full_symmetry} hold, and suppose both platforms are restricted to charge uniform prices to all consumers. For sufficiently small \(\delta\),
\begin{align}
p^u
=
c
+
(a-c)
\left[
	\frac{1-\beta}{2-\beta}
	+
	\frac{\delta(\theta-\beta)}{(2-\beta)^2}\bar d
\right]
+
O(\delta^2).
\label{eq:uniform_price_approx}
\end{align}
\end{corollary}

Corollary \ref{cor:uniform_price_approx} parallels the discriminatory-pricing formula in Section \ref{sec:equilibrium_pricing}. The benchmark term,
\[
	c+(a-c)\frac{1-\beta}{2-\beta},
\]
is the equilibrium price when network effects are locally shut down. The first-order network adjustment is governed by \(\theta-\beta\). When \(\theta>\beta\), interoperability dominates product substitution and the aggregate network raises the uniform price. When \(\theta<\beta\), the business-stealing force dominates and the aggregate network lowers the uniform price.

The role of uniform pricing is therefore not to eliminate network effects, but to aggregate them. Under discriminatory pricing, the first-order network adjustment depends on individual network position. Under uniform pricing, the corresponding adjustment depends on the average degree \(\bar d\). This observation will be useful below when comparing the distributional consequences of the two pricing regimes.

\subsection{Exact Equivalence Benchmarks}
\label{subsec:exact_equivalence}

We next identify cases in which the uniform-pricing and discriminatory-pricing regimes generate exactly the same equilibrium outcomes. These benchmarks clarify where the value of price discrimination comes from. Price discrimination matters only when network positions create heterogeneous pricing incentives. If either the pricing effect of network position disappears or all consumers occupy equivalent network positions, the two regimes coincide.

\begin{proposition}[Exact equivalence of pricing regimes]
\label{prop:exact_equivalence}
Suppose Assumptions \ref{ass:stability} and \ref{ass:full_symmetry} hold. The uniform-pricing and discriminatory-pricing regimes generate the same equilibrium prices, quantities, consumer surplus, platform profits, and total welfare in either of the following cases:

\begin{enumerate}[label=(\roman*)]
	\item \(\theta=\beta\);
	\item the network is \(d\)-regular, so that \(\bm G\bm 1=d\bm 1\).
\end{enumerate}
\end{proposition}

The first case follows from the centrality-reversal result. When
\[
	\theta=\beta,
\]
the network-position component of discriminatory prices vanishes. From Corollary \ref{cor:centrality_reversal},
\[
	p_i^*
	=
	\frac{a+c}{2}
	-
	\frac{\beta}{2(2-\beta)}(a-c)
	=
	c+\frac{1-\beta}{2-\beta}(a-c)
	\quad
	\text{for all } i.
\]
Thus, discriminatory pricing already produces a uniform price. The uniform-pricing constraint is therefore nonbinding.

The same conclusion can be verified from the uniform-pricing formula. When \(\theta=\beta\),
\[
	\bm M^+
	=
	\frac{1}{1+\beta}(\bm I-\delta\bm G)^{-1},
	\qquad
	\bm M^-
	=
	\frac{1}{1-\beta}(\bm I-\delta\bm G)^{-1}.
\]
Substituting these expressions into \eqref{eq:uniform_price_full_symmetry} gives
\[
	p^u
	=
	c+\frac{1-\beta}{2-\beta}(a-c),
\]
which is exactly the discriminatory price \(p_i^*\) for every consumer.

The second case shuts down heterogeneity in network position. If the network is \(d\)-regular, then
\[
	\bm G\bm 1=d\bm 1,
\]
so all consumers have the same Katz-Bonacich centrality under Assumption \ref{ass:full_symmetry}. Discriminatory pricing therefore yields a common price across consumers. In particular, from Corollary \ref{cor:regular_price},
\[
	p_i^*
	=
	c+
	\frac{(1-\beta)-\delta d(1-\theta)}
	{(2-\beta)-\delta d(2-\theta)}
	(a-c)
	\quad
	\text{for all } i.
\]
The uniform-pricing formula gives the same expression. Indeed, under a \(d\)-regular network,
\[
	\bm M^+\bm 1
	=
	\frac{1}{(1+\beta)-\delta d(1+\theta)}\bm 1,
	\qquad
	\bm M^-\bm 1
	=
	\frac{1}{(1-\beta)-\delta d(1-\theta)}\bm 1.
\]
Substituting these expressions into \eqref{eq:uniform_price_full_symmetry} yields
\[
	p^u
	=
	c+
	\frac{(1-\beta)-\delta d(1-\theta)}
	{(2-\beta)-\delta d(2-\theta)}
	(a-c)
	=
	p_i^*
	\quad
	\text{for all } i.
\]

Proposition \ref{prop:exact_equivalence} shows that the comparison between uniform and discriminatory pricing is meaningful only when two conditions are jointly present. First, network positions must be heterogeneous. Second, the network-position component of prices must be nonzero, which requires
\[
	\theta\neq\beta.
\]
If either condition fails, discriminatory pricing does not create centrality-based price dispersion, and imposing uniform pricing has no effect on equilibrium allocations or welfare.

\subsection{Price and Individual-Surplus Comparison}
\label{subsec:individual_comparison}

We now compare the two pricing regimes at the consumer level. The exact-equivalence benchmarks show that price discrimination matters only when network positions are heterogeneous and \(\theta\neq\beta\). In this subsection, we make this comparison explicit by using local approximations around \(\delta=0\). Throughout the subsection, suppose Assumptions \ref{ass:stability} and \ref{ass:full_symmetry} hold.

Recall from Corollary \ref{cor:uniform_price_approx} that
\[
p^u
=
c
+
(a-c)
\left[
	\frac{1-\beta}{2-\beta}
	+
	\frac{\delta(\theta-\beta)}{(2-\beta)^2}\bar d
\right]
+
O(\delta^2).
\]
The corresponding discriminatory price satisfies
\[
p_i^*
=
c
+
(a-c)
\left[
	\frac{1-\beta}{2-\beta}
	+
	\frac{\delta(\theta-\beta)}{(2-\beta)^2}d_i
\right]
+
O(\delta^2).
\]
Hence,
\begin{align}
p_i^*-p^u
=
(a-c)
\frac{\delta(\theta-\beta)}{(2-\beta)^2}
(d_i-\bar d)
+
O(\delta^2).
\label{eq:price_gap_uniform_discriminatory}
\end{align}

Equation \eqref{eq:price_gap_uniform_discriminatory} gives a simple interpretation of uniform pricing. To the first order in \(\delta\), the uniform price coincides with the simple average of discriminatory prices:
\[
p^u
=
\frac{1}{n}\sum_{i=1}^n p_i^*
+
O(\delta^2).
\]
Thus, uniform pricing does not primarily change the average price level at the first order. Instead, it removes the price dispersion generated by heterogeneous network positions. Under discriminatory pricing, the individual degree \(d_i\) enters the network adjustment; under uniform pricing, it is replaced by the average degree \(\bar d\).

The direction of this redistribution depends on \(\theta-\beta\). If \(\theta>\beta\), high-degree consumers face higher discriminatory prices than the uniform price, while low-degree consumers face lower discriminatory prices. If \(\theta<\beta\), the ranking is reversed: high-degree consumers receive discounts under discriminatory pricing, while low-degree consumers pay more than the uniform price.

The same logic applies to individual consumer surplus. Let
\[
	\tau=\frac{a-c}{(2-\beta)(1+\beta)},
	\qquad
	\epsilon=\frac{1+\theta}{1+\beta},
	\qquad
	\sigma=\frac{\theta-\beta}{2-\beta}.
\]
Using the first-order approximations of equilibrium consumption, we have
\begin{align}
x_i^*
&=
\tau
\left[
	1+\delta(\epsilon-\sigma)d_i
\right]
+
O(\delta^2),
\label{eq:x_discriminatory_approx}
\\
x_i^u
&=
\tau
\left[
	1+\delta\epsilon d_i-\delta\sigma\bar d
\right]
+
O(\delta^2).
\label{eq:x_uniform_approx}
\end{align}
Therefore,
\begin{align}
x_i^u-x_i^*
=
\tau\delta\sigma(d_i-\bar d)
+
O(\delta^2).
\label{eq:quantity_gap_uniform_discriminatory}
\end{align}

Since individual consumer surplus is
\[
	CS_i=(1+\beta)x_i^2,
\]
we obtain
\begin{align}
CS_i^u-CS_i^*
=
2(1+\beta)\tau^2\delta\sigma(d_i-\bar d)
+
O(\delta^2).
\label{eq:individual_cs_gap}
\end{align}
Thus, to the first order, the sign of the individual surplus effect is governed by
\[
	(\theta-\beta)(d_i-\bar d).
\]

\begin{proposition}[Individual effects of uniform pricing]
\label{prop:individual_uniform_discriminatory}
Suppose Assumptions \ref{ass:stability} and \ref{ass:full_symmetry} hold. For sufficiently small \(\delta\), the first-order comparison between uniform and discriminatory pricing satisfies
\[
	\operatorname{sign}(p_i^*-p^u)
	=
	\operatorname{sign}\{(\theta-\beta)(d_i-\bar d)\},
\]
and
\[
	\operatorname{sign}(CS_i^u-CS_i^*)
	=
	\operatorname{sign}\{(\theta-\beta)(d_i-\bar d)\},
\]
whenever \((\theta-\beta)(d_i-\bar d)\neq 0\).
\end{proposition}

Proposition \ref{prop:individual_uniform_discriminatory} highlights the distributional effect of restricting price discrimination. When interoperability is strong relative to product substitution, \(\theta>\beta\), central consumers are charged premia under discriminatory pricing. A uniform-pricing constraint then protects high-degree consumers by replacing their individual network position with the network average. Low-degree consumers, by contrast, lose the low prices they would receive under discriminatory pricing.

When product substitution dominates interoperability, \(\theta<\beta\), the opposite occurs. Central consumers receive influence-based discounts under discriminatory pricing because platforms compete aggressively for them. Uniform pricing removes these discounts and therefore hurts high-degree consumers, while benefiting low-degree consumers. Hence, a ban on discriminatory pricing does not have a uniform distributional effect across consumers. Its incidence depends jointly on interoperability and network position.

\subsection{Aggregate Welfare Comparison}
\label{subsec:aggregate_comparison}

We finally compare aggregate consumer surplus, platform profits, and total welfare under the two pricing regimes. The comparison in this subsection is based on the first-order approximations of equilibrium prices and quantities derived above. Accordingly, we use tildes to denote the induced approximate welfare measures. These objects are not exact second-order expansions of welfare, because the second-order terms of equilibrium prices and quantities are not included. They should therefore be interpreted as approximate welfare measures generated by the first-order local equilibrium.

Let
\[
	\widetilde{CS}^*,
	\quad
	\widetilde{\Pi}^*,
	\quad
	\widetilde{TW}^*
\]
denote the approximate consumer surplus, single-platform profit, and total welfare under discriminatory pricing. Similarly, let
\[
	\widetilde{CS}^u,
	\quad
	\widetilde{\Pi}^u,
	\quad
	\widetilde{TW}^u
\]
denote the corresponding approximate welfare measures under uniform pricing. Since there are two symmetric platforms,
\[
	\widetilde{TW}^*
	=
	\widetilde{CS}^*
	+
	2\widetilde{\Pi}^*,
	\qquad
	\widetilde{TW}^u
	=
	\widetilde{CS}^u
	+
	2\widetilde{\Pi}^u.
\]

Define the degree dispersion term
\[
	V_d
	=
	\sum_{i=1}^n(d_i-\bar d)^2
	=
	\sum_{i=1}^n d_i^2-n\bar d^2.
\]
This term is zero if and only if the network is regular. It captures the extent to which consumers differ in their local network positions.

\begin{proposition}[Approximate aggregate comparison of pricing regimes]
\label{prop:aggregate_uniform_discriminatory}
Suppose Assumptions \ref{ass:stability} and \ref{ass:full_symmetry} hold. Based on the first-order approximations of equilibrium prices and quantities, the approximate aggregate welfare differences satisfy
\begin{align}
\widetilde{CS}^u-\widetilde{CS}^*
&=
\frac{\delta^2(a-c)^2}{(2-\beta)^2(1+\beta)}
\sigma(2\epsilon-\sigma)V_d,
\label{eq:approx_cs_gap}
\\
\widetilde{\Pi}^u-\widetilde{\Pi}^*
&=
\frac{\delta^2(a-c)^2}{(2-\beta)^2(1+\beta)}
\sigma(\sigma-\epsilon)V_d,
\label{eq:approx_profit_gap}
\\
\widetilde{TW}^u-\widetilde{TW}^*
&=
\frac{\delta^2(a-c)^2}{(2-\beta)^2(1+\beta)}
\sigma^2 V_d,
\label{eq:approx_tw_gap}
\end{align}
where
\[
	\epsilon=\frac{1+\theta}{1+\beta},
	\qquad
	\sigma=\frac{\theta-\beta}{2-\beta}.
\]
\end{proposition}

Because
\[
	2\epsilon-\sigma
	=
	\frac{4-\beta+\beta^2+3\theta(1-\beta)}
	{(1+\beta)(2-\beta)}
	>
	0,
\]
the sign of \(\widetilde{CS}^u-\widetilde{CS}^*\) is governed by \(\sigma\), or equivalently by \(\theta-\beta\). Thus,
\[
	\theta>\beta
	\quad\Longrightarrow\quad
	\widetilde{CS}^u>\widetilde{CS}^*
\]
while
\[
	\theta<\beta
	\quad\Longrightarrow\quad
	\widetilde{CS}^u<\widetilde{CS}^*.
\]
whenever the network is non-regular.

Consumers as a group prefer uniform pricing when interoperability dominates product substitution, because discriminatory pricing then extracts more surplus from central consumers. They prefer discriminatory pricing when product substitution dominates interoperability, because price discrimination then takes the form of competitive discounts to central consumers.

For platform profits, note that
\[
	\sigma-\epsilon
	=
	-\frac{2(1-\theta\beta)+\theta+\beta}
	{(1+\beta)(2-\beta)}
	<
	0.
\]
Hence, the sign of \(\widetilde{\Pi}^u-\widetilde{\Pi}^*\) is opposite to the sign of \(\theta-\beta\). When \(\theta>\beta\), platforms prefer discriminatory pricing, because it allows them to charge premia to well-connected consumers. When \(\theta<\beta\), platforms prefer uniform pricing, because discriminatory pricing intensifies competition for central consumers and transfers surplus to them through discounts.

The total-welfare comparison is more direct. Equation \eqref{eq:approx_tw_gap} implies
\[
	\widetilde{TW}^u-\widetilde{TW}^*
	\geq 0,
\]
with strict inequality whenever
\[
	\theta\neq\beta
	\quad\text{and}\quad
	V_d>0.
\]
Thus, within this local approximation, uniform pricing weakly increases total welfare. The reason is that discriminatory pricing reallocates consumption toward consumers with more favorable pricing positions. This redistribution benefits either consumers or platforms depending on the sign of \(\theta-\beta\), but the induced price dispersion creates a second-order welfare loss relative to uniform pricing.

Proposition \ref{prop:aggregate_uniform_discriminatory} therefore delivers a simple policy message. The distributional effect of banning discriminatory pricing depends on whether interoperability or product substitution is stronger. If \(\theta>\beta\), a ban benefits consumers but hurts platforms. If \(\theta<\beta\), a ban hurts consumers but benefits platforms. However, in the approximate total-welfare comparison, uniform pricing weakly dominates discriminatory pricing whenever network positions are heterogeneous and the centrality-pricing component is active.


\appendix

\section{Proofs for the Model and Consumption Stage}
\label{app:model_consumption}

\subsection{Proof of Lemma \ref{lem:stability_premium_region}}

We need to show that Assumption \ref{ass:stability} implies
\[
	\delta(2-\theta)\rho(\bm G)<2-\beta.
\]
There are two cases.

First, suppose \(\theta\geq\beta\). Then
\[
	\frac{1+\beta}{1+\theta}
	\leq
	\frac{2-\beta}{2-\theta},
\]
because this inequality is equivalent to \(\beta\leq\theta\). Hence,
\[
	\delta(1+\theta)\rho(\bm G)<1+\beta
\]
implies
\[
	\delta\rho(\bm G)
	<
	\frac{1+\beta}{1+\theta}
	\leq
	\frac{2-\beta}{2-\theta}.
\]
Therefore,
\[
	\delta(2-\theta)\rho(\bm G)<2-\beta.
\]

Second, suppose \(\theta<\beta\). Then
\[
	\frac{1-\beta}{1-\theta}
	<
	\frac{2-\beta}{2-\theta},
\]
because this inequality is equivalent to \(\theta<\beta\). Hence,
\[
	\delta(1-\theta)\rho(\bm G)<1-\beta
\]
implies
\[
	\delta\rho(\bm G)
	<
	\frac{1-\beta}{1-\theta}
	<
	\frac{2-\beta}{2-\theta}.
\]
Therefore,
\[
	\delta(2-\theta)\rho(\bm G)<2-\beta.
\]

Thus, in all cases,
\[
	\delta(2-\theta)\rho(\bm G)<2-\beta.
\]
Since \(\bm G\) is symmetric, this condition guarantees that
\[
	(2-\beta)\bm I-\delta(2-\theta)\bm G
\]
is nonsingular. The proof is complete.

\subsection{Proof of Proposition \ref{prop:consumption}}

\begin{proof}
According to FOC, maximizing the utility function of consumer i with respect to $x_i^k$, and we can yeild the equilibrium consumption as:
\begin{align}
	\left\{
	\begin{aligned}
		\frac{\partial u_i}{\partial x_i^A}=-x_i^A-\beta x_i^B +a_i^A -p_i^A + \delta \mathop{\Sigma}_{j \in N}g_{ij}x_j^A + \theta \delta \mathop{\Sigma}_{j \in N} g_{ij}x^B_j =0
		\\
		\frac{\partial u_i}{\partial x_i^B}=-x_i^B-\beta x_i^A +a_i^B -p_i^B + \delta \mathop{\Sigma}_{j \in N}g_{ij}x_j^B + \theta \delta \mathop{\Sigma}_{j \in N} g_{ij}x^A_j =0
	\end{aligned}
	\right.  ,
	\label{equilibrium_consumption_FOC_1}
\end{align}
and equation \eqref{equilibrium_consumption_FOC_1} can be expressed as the form of vector as 
	\begin{align}
	\left\{
	\begin{aligned}
		\bm{x}^A + \beta\bm{x}^B = \bm{a}^A - \bm{p}^A + \delta\bm{G}\bm{x}^A + \theta \delta\bm{G}\bm{x}^B
		\\
		\bm{x}^B + \beta\bm{x}^A = \bm{a}^B - \bm{p}^B + \delta\bm{G}\bm{x}^B + \theta \delta\bm{G}\bm{x}^A
	\end{aligned}
	\right.  .
	\label{equilibrium_consumption_FOC_2}
\end{align}
Perform the sum and difference operations on the two equations in \eqref{equilibrium_consumption_FOC_2}, we obtain
\begin{align*}
	\left\{
	\begin{aligned}
		(1 + \beta)(\bm{x}^A + \bm{x}^B) = (\bm{a}^A + \bm{a}^B) - (\bm{p}^A + \bm{p}^B)+\delta(1+\theta)\bm{G}(\bm{x}^A + \bm{x}^B)
		\\
		(1 - \beta)(\bm{x}^A - \bm{x}^B) = (\bm{a}^A - \bm{a}^B) - (\bm{p}^A - \bm{p}^B)+\delta(1-\theta)\bm{G}(\bm{x}^A - \bm{x}^B)
	\end{aligned}
	\right.  .
\end{align*}
Therefore, there is
\begin{align*}
	\left\{
	\begin{aligned}
		(\bm{x}^A + \bm{x}^B) = \left[(1 + \beta)- \delta(1+\theta)\bm{G}\right]^{-1}\left[(\bm{a}^A + \bm{a}^B) - (\bm{p}^A + \bm{p}^B)\right]
		\\
		(\bm{x}^A - \bm{x}^B) = \left[(1 - \beta)- \delta(1-\theta)\bm{G}\right]^{-1}\left[(\bm{a}^A - \bm{a}^B) - (\bm{p}^A - \bm{p}^B)\right]
	\end{aligned}
	\right.  ,
\end{align*}
substituting $\bm{M}^+$ and $\bm{M}^-$ into the above equation gives
\begin{align*}
	\left\{
	\begin{aligned}
		(\bm{x}^A + \bm{x}^B) = \bm{M}^+\left[(\bm{a}^A + \bm{a}^B) - (\bm{p}^A + \bm{p}^B)\right]
		\\
		(\bm{x}^A - \bm{x}^B) = \bm{M}^-\left[(\bm{a}^A - \bm{a}^B) - (\bm{p}^A - \bm{p}^B)\right]
	\end{aligned}
	\right.  .
\end{align*}

Finally, we can obtain the unique solution of equilibrium consumption as 
\begin{align}
	\left\{
	\begin{aligned}
		\bm{x}^A = \frac{1}{2}\left\{ \bm{M}^+\left[ (\bm{a}^A + \bm{a}^B) - (\bm{p}^A + \bm{p}^B)\right]  + \bm{M}^-\left[ (\bm{a}^A - \bm{a}^B) - (\bm{p}^A - \bm{p}^B) \right]   \right\}
		\\
		\bm{x}^B = \frac{1}{2}\left\{ \bm{M}^+\left[ (\bm{a}^A + \bm{a}^B) - (\bm{p}^A + \bm{p}^B)\right]  - \bm{M}^-\left[ (\bm{a}^A - \bm{a}^B) - (\bm{p}^A - \bm{p}^B) \right]   \right\}
	\end{aligned}
	\right.  ,
\end{align}
by simplifying the above equation, we can obtain the \eqref{eq:consumption_equilibrium} mentioned in the main text. The proof is complete.
\end{proof}

\section{Proofs for Equilibrium Pricing}
\label{app:pricing}

\subsection{Proof of Proposition \ref{prop:pricing}}

\begin{proof}
Before proving this proposition, we need to introduce a lemma.
\begin{lemma}
	If there is a $n \times n$ matrix $\bm{G}$ and a constant $\alpha$, then we have
	\begin{align*}
		\bm{G}(\bm{I}-\alpha\bm{G})^{-1}=\frac{1}{\alpha}\left[ (\bm{I}-\alpha\bm{G})^{-1} - \bm{I} \right]
	\end{align*}
	\label{lem:proof_of_equilbrium_price}
\end{lemma}
Proof of Lemma \ref{lem:proof_of_equilbrium_price} is very intuitive. As we already know that 
\begin{align*}
	(\bm{I} - \alpha\bm{G})^{-1}=\mathop{\Sigma}_{k=0}^{+\infty}\alpha^k\bm{G}^k=\bm{I} + \alpha\bm{G} + \alpha^2\bm{G}^2+.....+\alpha^n\bm{G}^n+.....,
\end{align*}
so
\begin{align*}
	\bm{G}	(\bm{I} - \alpha\bm{G})^{-1} = \bm{G} + \alpha\bm{G}^2+.....+\alpha^n\bm{G}^{n+1}+.....
\end{align*}
It is obvious that $\bm{G}(\bm{I}-\alpha\bm{G})^{-1}=\frac{1}{\alpha}\left[ (\bm{I}-\alpha\bm{G})^{-1} - \bm{I} \right] $, and the Lemma has been proved. Now we can focus on the proof of Proposition \ref{prop:pricing}.

According to \eqref{eq:profit}, the total profits of platform $A$ should be
\begin{align*}
	\begin{aligned}
		\Pi^A = \left< \bm{p}^A-\bm{c}^A, \bm{x}^A \right>
		=\left< \bm{p}^A-\bm{c}^A, \frac{\bm{M}^+ + \bm{M}^-}{2}(\bm{a}^A-\bm{p}^A) + \frac{\bm{M}^+ - \bm{M}^-}{2}(\bm{a}^B - \bm{p}^B) \right>
	\end{aligned},
\end{align*}
therefore, 
\begin{align*}
	\begin{aligned}
		\frac{\partial \Pi^A}{\partial \bm{p}^A}=\frac{\bm{M}^+ + \bm{M}^-}{2}(\bm{a}^A-\bm{p}^A) + \frac{\bm{M}^+ - \bm{M}^-}{2}(\bm{a}^B - \bm{p}^B)- \frac{\bm{M}^+ + \bm{M}^-}{2}(\bm{p}^A-\bm{c}^A)=0
	\end{aligned},
\end{align*}
then we obtain
\begin{align}
	\left\{
	\begin{aligned}
		\bm{p}^A=\frac{1}{2}(\bm{a}^A+\bm{c}^A) + \frac{1}{2}\left( \bm{M}^+ + \bm{M}^- \right)^{-1}\left( \bm{M}^+ - \bm{M}^- \right)(\bm{a}^B - \bm{p}^B)
		\\
		\bm{p}^B=\frac{1}{2}(\bm{a}^B+\bm{c}^B) + \frac{1}{2}\left( \bm{M}^+ + \bm{M}^- \right)^{-1}\left( \bm{M}^+ - \bm{M}^- \right)(\bm{a}^A - \bm{p}^A)
	\end{aligned}.
	\right.
	\label{proof_of_EP_1}
\end{align}
Assumptions \ref{ass:stability} and \ref{ass:platform_symmetry} hold, there exists a unique symmetric equilibrium satisfying $\bm{p}^A = \bm{p}^B = \bm{p}$. According to \eqref{proof_of_EP_1}, the equation should be
\begin{align*}
	\begin{aligned}
		(\bm{M}^+ + \bm{M}^-)(\bm{p} - \bm{c}) = (\bm{M}^+ + \bm{M}^-)(\bm{a}-\bm{p}) + (\bm{M}^+ - \bm{M}^-)(\bm{a}-\bm{p})
	\end{aligned}.
\end{align*}
By combining like terms and moving terms in the above equation, we can obtain
\begin{align}
	\begin{aligned}
		(3\bm{M}^+ + \bm{M}^-)\bm{p}=2\bm{M}^+\bm{a} + (\bm{M}^+ + \bm{M}^-)\bm{c}.
	\end{aligned}
	\label{proof_of_EP_2}
\end{align}
Let $n \times n$ matrix $\bm{H} = 3\bm{M}^+ + \bm{M}^-$, then $\bm{M}^+ +\bm{M}^-=\bm{H} -2 \bm{M}^+$, substituting $H$ into \eqref{proof_of_EP_2} gives
\begin{align*}
	\begin{aligned}
		\bm{H}\bm{p} = 2\bm{M}^+\bm{a} + (\bm{H}- 2\bm{M}^+)\bm{c}
	\end{aligned},
\end{align*}
then the equilibrium price $\bm{p}^*$ should be
\begin{align}
	\begin{aligned}
		\bm{p}^* &= \bm{c} + 2\bm{H}^{-1}\bm{M}^+(\bm{a} - \bm{c})
		\\
		&=\bm{c} + 2(3\bm{M}^+ + \bm{M}^-)^{-1}\bm{M}^+(\bm{a} - \bm{c})
	\end{aligned}.
	\label{proof_of_EP_3}
\end{align}

Denote two $n \times n $ matrices as
\begin{align*}
	\begin{aligned}
		\bm{D} = [(1+\beta)\bm{I} - \delta(1 + \theta)\bm{G}] = (\bm{M}^+)^{-1}\quad and \quad\bm{K} = [(1-\beta)\bm{I} - \delta(1 - \theta)\bm{G}] = (\bm{M}^-)^{-1}
	\end{aligned},
\end{align*}
then substitute D and K into \eqref{proof_of_EP_3}, we can obtain
\begin{align*}
	\begin{aligned}
		\bm{p}^* = \bm{c} + 2(3\bm{D}^{-1} + \bm{K}^{-1})^{-1}\bm{D}^{-1}(\bm{a} - \bm{c})
	\end{aligned}.
\end{align*}
Then by taking out an $\bm{D}^{-1}$ from the left side of the inverse matrix, we get
\begin{align*}
	\begin{aligned}
		\bm{p}^* &= \bm{c} + 2\left[\bm{D}^{-1}(3\bm{I} + \bm{D}\bm{K}^{-1})\right]^{-1}\bm{D}^{-1}(\bm{a} - \bm{c})
		\\
		& = \bm{c} + 2\left(3\bm{I} + \bm{D}\bm{K}^{-1}\right)^{-1}\bm{D}\bm{D}^{-1}(\bm{a} - \bm{c})
		\\
		& = \bm{c} + 2\left(3\bm{I} + \bm{D}\bm{K}^{-1}\right)^{-1}(\bm{a} - \bm{c})
	\end{aligned}.
\end{align*}
Perform a similar operation, extracting a $\bm{K}^{-1}$ from the right side of the inverse matrix, we obtain
\begin{align*}
	\begin{aligned}
		\bm{p}^* & = \bm{c} + 2\left[(3\bm{K} + \bm{D})\bm{K}^{-1}\right]^{-1}(\bm{a} - \bm{c})
		\\
		& = \bm{c} + 2\bm{K}\left(3\bm{K} + \bm{D}\right)^{-1}(\bm{a} - \bm{c})
	\end{aligned}.
\end{align*}
Finally, we have the equilibrium price as
\begin{align}
	\begin{aligned}
		\bm{p}^* 
		=\bm{c} + \left[(1- \beta)\bm{I} - \delta(1- \theta)\bm{G}\right] \left[ (2-\beta)\bm{I} - \delta(2-\theta)\bm{G} \right]^{-1}(\bm{a}-\bm{c})
	\end{aligned}
	\label{proof_of_EP_4}
\end{align}

Now denote a matrix $\bm{V} =  (2-\beta)\bm{I} - \delta(2-\theta)\bm{G}$, the \eqref{proof_of_EP_4} becomes
\begin{align*}
	\begin{aligned}
		\bm{p}^* = \bm{c} + \bm{K}\bm{V}^{-1}(\bm{a}-\bm{c})
	\end{aligned}.
\end{align*}

Suppose $\bm{K} = \gamma\bm{V} + \epsilon\bm{G}$, then we have
\begin{align*}
	\begin{aligned}
		\bm{K} &= \gamma\left[ (2-\beta)\bm{I} - \delta(2-\theta)\bm{G} \right] + \epsilon\bm{G}
		\\
		&=\gamma(2-\beta)\bm{I} -\left[ \gamma\delta(2-\theta) - \epsilon \right]\bm{G}
	\end{aligned},
\end{align*}
by comparing this equation with the original $K=[(1-\beta)\bm{I} - \delta(1 - \theta)\bm{G}]$, we can obtain
\begin{align}
	\left\{
	\begin{aligned}
		&\gamma(2-\beta)=1-\beta
		\\
		&\gamma\delta(2-\theta) - \epsilon = \delta(1-\theta)
	\end{aligned}.
	\right.
	\label{proof_of_EP_5}
\end{align}
According to \eqref{proof_of_EP_5}, it is obvious that $\gamma = \frac{1-\beta}{2-\beta}$. Substituting $\gamma$ into the second equation allows us to solve $\epsilon$ as
\begin{align*}
	\epsilon &= \gamma\delta(2-\theta)-\delta(1-\theta),
	\\
	\epsilon &= \frac{\delta(1-\beta)(2-\theta)}{2-\beta} - \delta(1-\theta),
	\\
	\epsilon &=\frac{\delta(1-\beta)(2-\theta) - \delta(1-\theta)(2-\beta)}{2-\beta},
	\\
	\epsilon &= \frac{\delta(\theta - \beta)}{2-\beta}.
\end{align*}
Therefore, we get
\begin{align*}
	\bm{K}=\frac{1-\beta}{2-\beta}\bm{V} + \frac{\delta(\theta - \beta)}{2-\beta}\bm{G},
\end{align*}
then 
\begin{align}
	\begin{aligned}
		\bm{p}^* &= \bm{c} + \bm{K}\bm{V}^{-1}(\bm{a} - \bm{c})
		\\
		&=\bm{c} + \left[ \frac{1-\beta}{2-\beta}\bm{V} + \frac{\delta(\theta - \beta)}{2-\beta}\bm{G} \right]\bm{V}^{-1}(\bm{a} - \bm{c})
		\\
		&=\bm{c} + \left[\frac{1-\beta}{2-\beta}(\bm{a} - \bm{c}) + \frac{\delta(\theta - \beta)}{2-\beta}\bm{G} \bm{V}^{-1}(\bm{a} - \bm{c})\right]
		\\
		&=\bm{c} + \left[\frac{1-\beta}{2-\beta}(\bm{a} - \bm{c}) + \frac{\delta(\theta - \beta)}{2-\beta}\bm{G} \left[(2-\beta)\bm{I} - \delta(2-\theta)\bm{G}\right]^{-1}(\bm{a} - \bm{c})\right]
		\\
		&=\bm{c} + \left[\frac{1-\beta}{2-\beta}(\bm{a} - \bm{c}) + \frac{\delta(\theta - \beta)}{(2-\beta)^2}\bm{G} \left[\bm{I} - \frac{\delta(2-\theta)}{2-\beta}\bm{G}\right]^{-1}(\bm{a} - \bm{c})\right]
	\end{aligned}
	\label{proof_of_EP_6}
\end{align}

Now, according to Lemma \ref{lem:proof_of_equilbrium_price}, we can translate \eqref{proof_of_EP_6} to
\begin{align*}
	\begin{aligned}
		\bm{p}^*&= \bm{c} + \frac{1-\beta}{2-\beta}(\bm{a} - \bm{c}) + \frac{\delta(\theta - \beta)}{(2-\beta)^2} \cdot \frac{2-\beta}{\delta(2-\theta)}\left[ \left(  \bm{I} - \frac{\delta(2-\theta)}{2-\beta}\bm{G}\right)^{-1} - \bm{I} \right](\bm{a}-\bm{c})
		\\
		&= \bm{c} + \frac{1-\beta}{2-\beta}(\bm{a} - \bm{c}) +\frac{(\theta - \beta)}{(2-\beta)(2-\theta)}\left[ \left(  \bm{I} - \frac{\delta(2-\theta)}{2-\beta}\bm{G}\right)^{-1}(\bm{a}-\bm{c}) - (\bm{a}-\bm{c}) \right]
	\end{aligned}
\end{align*}

According to \citet{BallesterCalvoArmengolZenou2006Econometrica}, there is a K-B centrality in the formula of $\bm{p}^*$, and we can express the equilibrium price as 
\begin{align*}
	\begin{aligned}
		\bm{p}^*= \bm{c} + \left[\frac{1-\beta}{2-\beta}-\frac{\theta - \beta}{(2-\beta)(2-\theta)} \right](\bm{a} - \bm{c}) +\frac{\theta - \beta}{(2-\beta)(2-\theta)}\bm{b}\left(\bm{G}, \frac{\delta(2-\theta)}{2-\beta}, \bm{a}-\bm{c}\right)
	\end{aligned},
\end{align*}
then we can obtain the final form of the equilibrium price
\begin{align*}
	\begin{aligned}
		\bm{p}^* = \bm{c} + \frac{1-\theta}{2-\theta}(\bm{a} -\bm{c}) + \frac{\theta - \beta}{(2-\beta)(2-\theta)}\bm{b}\left(\bm{G}, \frac{\delta(2-\theta)}{2-\beta}, \bm{a}-\bm{c}\right)
	\end{aligned}.
\end{align*}
Proposition \ref{prop:pricing} proved.

\end{proof}

\subsection{Proof of Corollary \ref{cor:centrality_reversal}}

\begin{proof}
To be completed.
\end{proof}

\subsection{Proof of Corollary \ref{cor:regular_price}}

\begin{proof}
Under Assumptions \ref{ass:stability} and \ref{ass:full_symmetry}, for a regular graph where each consumer has $d$ neighbors, the equilibrium price can be computed from \eqref{proof_of_EP_4}:\footnote{If $\bm{A}\bm{1}=\alpha \bm{1}$, then $\bm{A}^{-1}\bm{1}=\frac{1}{\alpha}\bm{1}$.}
\begin{align*}
	\begin{aligned}
		\bm{p}^* 
		&=c\bm{1} + (a-c)\left[(1- \beta)\bm{I} - \delta(1- \theta)\bm{G}\right] \left[ (2-\beta)\bm{I} - \delta(2-\theta)\bm{G} \right]^{-1}\bm{1}
		\\
		&=c\bm{1} +(a-c)[(2-\beta)-\delta d(2-\theta)]^{-1}\left[(1- \beta)\bm{I} - \delta(1- \theta)\bm{G}\right]\bm{1}
		\\
		&=c\bm{1} +(a-c)[(2-\beta)-\delta d(2-\theta)]^{-1}\left[(1- \beta)- \delta d(1- \theta)\right]\bm{1}
	\end{aligned}
\end{align*} 
Corollary \ref{cor:regular_price} proved.
\end{proof}

\section{Proofs for Welfare Analysis}
\label{app:welfare}

\subsection{Proof of Proposition \ref{prop:welfare_objects}}

\begin{proof}
Suppose Assumptions \ref{ass:stability} and \ref{ass:platform_symmetry} hold. we define 4 operatos as
\begin{align*}
	\begin{aligned}
		\bm{\Phi}^{EC}&=\frac{(\bm{I}-\delta\bm{G})}{[(1+\beta)\bm{I} - \delta(1+\theta)\bm{G}][(2-\beta)\bm{I}-\delta(2-\theta)\bm{G}]}~~;
	\end{aligned}
\end{align*}
\begin{align*}
	\bm{\Phi}^{CS}&=(1+\beta)\left[\frac{(\bm{I}-\delta\bm{G})}{[(1+\beta)\bm{I} - \delta(1+\theta)\bm{G}][(2-\beta)\bm{I}-\delta(2-\theta)\bm{G}]}\right]^2~~;
\end{align*}
\begin{align*}
	\bm{\Phi}^{PT}=\frac{(\bm{I} - \delta\bm{G})[(1-\beta)\bm{I}-\delta(1-\theta)\bm{G}]}{[(1+\beta)\bm{I} - \delta(1+\theta)\bm{G}][(2-\beta)\bm{I}-\delta(2-\theta)\bm{G}]^2}~~,
\end{align*}
\begin{align*}
	\bm{\Phi}^{TW}=\bm{\Phi}^{CS} + 2\bm{\Phi}^{PT}.
\end{align*}

According to \eqref{proof_of_EP_4}
	\begin{align*}
	\bm{p}^* 
	=\bm{c} + \left[(1- \beta)\bm{I} - \delta(1- \theta)\bm{G}\right] \left[ (2-\beta)\bm{I} - \delta(2-\theta)\bm{G} \right]^{-1}(\bm{a}-\bm{c})=\bm{c}+\bm{K}\bm{V}^{-1}(\bm{a}-\bm{c}),
\end{align*}
and we alread know that under Assumption   \ref{ass:stability} and \ref{ass:platform_symmetry}, the equilibrium consumptions is
\begin{align*}
	\bm{x}^*=\bm{M}^+(\bm{a}-\bm{p}).
\end{align*}
There is 
\begin{align*}
	\begin{aligned}
		\bm{x}^*&=\bm{M}^+\left( \bm{a}-\bm{c}-\bm{K}\bm{V}^{-1}(\bm{a}-\bm{c}) \right)
		\\
		&=\bm{M}^+(\bm{V}-\bm{K})\bm{V}^{-1}(\bm{a}-\bm{c})
		\\
		&=\bm{M}^+(\bm{I}-\delta\bm{G})\bm{V}^{-1}(\bm{a}-\bm{c})
	\end{aligned}.
\end{align*}
Because the three matrices, $\bm{M}^+$, $(\bm{I}-\delta\bm{G})$ and $\bm{V}^{-1}$, are all commutative and symmetric, so we can express $\bm{x}^*$ as
\begin{align}
	\begin{aligned}
		\bm{x}^*=(\bm{I}-\delta\bm{G})\bm{M}^+\bm{V}^{-1}(\bm{a}-\bm{c}) = \bm{\phi}^{EC}(\bm{a}-\bm{c})
	\end{aligned}.
	\label{16}
\end{align}

For consumer suplus, it can be expressed as the form of vector
\begin{align}
	\begin{aligned}
		CS(\bm{G};\beta,\delta,\theta)=\mathop{\Sigma}_{i \in N}u_i = 2(\bm{a}-\bm{p})^T\bm{x}-(1+\beta)\bm{x}^T\bm{x}+2\delta(1+\theta)\bm{x}^T\bm{G}\bm{x}
	\end{aligned},
	\label{14}
\end{align}
according to the FOC condition, we have
\begin{align*}
	\bm{a}-\bm{p}=\left[(1+\beta)\bm{I}-\delta(1+\theta)\bm{G}\right]\bm{x}^*,
\end{align*}
substitute this condition into \eqref{14}, we obtain
\begin{align}
	CS^*(\bm{G};\beta,\delta,\theta)=\mathop{\Sigma}_{i \in N} u_i=(1+\beta)( \bm{x}^*)^T\bm{x}^*.
\end{align}
Substituting the result of \eqref{16} step into the above equation gives
	\begin{align*}
	\begin{aligned}
		CS^*(\bm{G};\beta,\delta,\theta)&=(1+\beta)[\bm{\Phi}^{EC}(\bm{a}-\bm{c})]^T(\bm{\Phi}^{EC}(\bm{a}-\bm{c}))
		\\
		&=(\bm{a}-\bm{c})^T(1+\beta)(\bm{\Phi}^{EC})^2(\bm{a}-\bm{c})
		\\
		&=\left< \bm{a}-\bm{c}, \bm{\Phi}^{CS}(\bm{a}-\bm{c}) \right>
	\end{aligned}.
\end{align*}

As for the platforms' perspective, we know from \eqref{eq:profit} that the form of platform's profit is
\begin{align*}
	\begin{aligned}
		\Pi^*(\bm{G};\beta,\delta,\theta)&=(\bm{P^*}-\bm{c})^T\bm{x}^*
		\\
		&=[\bm{K}\bm{V}^{-1}(\bm{a}-\bm{c})]^T[(\bm{I}-\delta\bm{G})\bm{M}^+\bm{V}^{-1}(\bm{a}-\bm{c})]
		\\
		&=(\bm{a}-\bm{c})^T[\bm{K}\bm{V}^{-1}\bm{M}^+(\bm{I}-\delta\bm{G})\bm{V}^{-1}](\bm{a}-\bm{c})
		\\
		&=(\bm{a}-\bm{c})^T\frac{(\bm{I}-\delta\bm{G})[(1-\beta)\bm{I}-\delta(1-\theta)\bm{G}]}{[(1+\beta)\bm{I} - \delta(1+\theta)\bm{G}][(2-\beta)\bm{I}-\delta(2-\theta)\bm{G}]^2}(\bm{a}-\bm{c})
		\\
		&=(\bm{a}-\bm{c})^T\bm{\Phi}^{PT}(\bm{a}-\bm{c})
	\end{aligned}.
\end{align*}

Based on the expressions of consumer surplus and platform profit derived in the previous text, we can easily have the total welfare of this society is
\begin{align*}
	\begin{aligned}
		TW^*(\bm{G};\beta,\delta,\theta)&= CS^*(\bm{G};\beta,\delta,\theta)+2\Pi^*(\bm{G};\beta,\delta,\theta)
		\\
		&=(\bm{a}-\bm{c})^T(\bm{\Phi}^{CS} + 2\bm{\Phi}^{PT})(\bm{a}-\bm{c})
		\\&=\left<\bm{a}-\bm{c}, ~\bm{\Phi}^{CS} + 2\bm{\Phi}^{PT}(\bm{a}-\bm{c})\right>
	\end{aligned}.
\end{align*}

Therefore, we have
\begin{align*}
	\begin{aligned}
		\bm{x}^*&=\bm{\Phi}^{EC}(\bm{a}-\bm{c})~~;
		\\
		CS^*&=\left<\bm{a}-\bm{c} , \bm{\Phi}^{CS}(\bm{a}-\bm{c}) \right>~~;
		\\
		\Pi^*&=\left<\bm{a}-\bm{c} , \bm{\Phi}^{PT}(\bm{a}-\bm{c})\right>~~;
		\\
		TW^*&=\left<\bm{a}-\bm{c}, (\bm{\Phi}^{CS}+2\bm{\Phi}^{PT})(\bm{a}-\bm{c})  \right>~~.
	\end{aligned}
\end{align*}

The proof of Theorem \ref{prop:welfare_objects} is complete.
\end{proof}

\subsection{Proof of Proposition \ref{prop:regular_welfare}}

\begin{proof}
When Assumption \ref{ass:stability} and \ref{ass:full_symmetry} hold, for a regular graph where each consumer has $d$ neighbors, the equilibrium consumption can be computed as bellow
\begin{align*}
	\begin{aligned}
		\bm{x}^*&=(a-c)\bm{\Phi}^{EC}\bm{1}
		\\
		&=(a-c)(\bm{I}-\delta\bm{G})[(1+\beta)\bm{I}-\delta(1+\theta)\bm{G}]^{-1}[(2-\beta)\bm{I}-\delta(2-\theta)\bm{G}]^{-1}\bm{1}
		\\
		&=\frac{(a-c)(1-\delta d)}{[(1+\beta)-\delta d (1+\theta)][(2-\beta)-\delta d(2-\theta)]}\bm{1}
	\end{aligned}~~.
\end{align*}

As for the consumer suplus, we have $CS^*=(1+\beta)(\bm{x}^*)^T(\bm{x}^*)$, therefore
\begin{align*}
	\begin{aligned}
		CS^*&=(1+\beta)(a-c)^2(\bm{\Phi}^{EC}\bm{1})^T(\bm{\Phi}^{EC}\bm{1})
		\\
		&= n(1+\beta)\left(\frac{(a-c)(1-\delta d)}{[(1+\beta)-\delta d (1+\theta)][(2-\beta)-\delta d(2-\theta)]}\right)^2
	\end{aligned}.
\end{align*}

Similarly, we can calculate $\Pi^*$ as
\begin{align*}
	\begin{aligned}
		\Pi^*&= (a-c)^2\bm{1}^T\bm{\Phi}^{PT}\bm{1}
		\\
		&= \frac{n(a-c)^2(1-\delta d)[(1-\beta) - \delta d(1-\theta)]}{[(1+\beta)-\delta d (1+\theta)][(2-\beta)-\delta d(2-\theta)]^2}
	\end{aligned}~~.
\end{align*}

The proof of Proposition \ref{prop:regular_welfare} is complete.
\end{proof}

\subsection{Proof of Proposition \ref{prop:cs_profit_comparative}}

\begin{proof}
(i) When the network is d-regular and Assumption \ref{ass:stability} and \ref{ass:full_symmetry} hold, taking the derivative of $\bm{x}^*$ with respect to $\theta$ gives
\begin{align*}
	\frac{\partial \bm{x}^*(\bm{G};\beta,\delta,\theta)}{\partial \theta} 
	= \frac{\delta d (a - c)(1 - \delta d)\left[ (1-2\beta) - \delta d(1 - 2\theta) \right]}{\left[ (1 + \beta) - \delta d(1 + \theta) \right]^2 \left[ (2 - \beta) - \delta d(2 - \theta) \right]^2} \bm{1}~,
\end{align*}
Taking the derivative of $CS^*$ with respect to $\theta$ gives
\begin{align*}
	\begin{aligned}
		\frac{\partial CS^*(\bm{G};\beta,\delta,\theta)}{\partial \theta} = 2(1 + \beta) \left( \bm{a} - \bm{c} \right)^\top \bm{\Phi}^{EC} \frac{\partial \bm{\Phi}^{EC}}{\partial \theta} \left( \bm{a} - \bm{c} \right)
	\end{aligned}.
\end{align*}

Therefore, it is obvious that
\begin{align*}
	\begin{aligned}
		sign\left\{ \frac{\partial x^*}{\partial \theta} \right\} = sign\left\{\frac{\partial CS^*}{\partial \theta}\right\} = sign\left\{1-2\beta-\delta d +2\delta d \theta\right\}.
	\end{aligned}
\end{align*}

(ii) Let 
\begin{align*}
	\begin{aligned}
		s=1-\delta d;~~~~~~~~~
		\\
		m^+=[(1+\beta)-\delta d (1+\theta)];
		\\
		m^-=[(1-\beta)-\delta d (1-\theta)];
		\\
		v=[(2-\beta)-\delta d (2-\theta)]~,
	\end{aligned}
\end{align*}
then from Proposition \ref{prop:regular_welfare}, we can write the profit as 
\begin{align*}
	\begin{aligned}
		\Pi^*(\bm{G};\beta,\delta,\theta) = n(a-c)^2(1-\delta d )\frac{m^-}{m^+v^2}
	\end{aligned}.
\end{align*}

Because $n(a-c)^2(1-\delta d )>0$ are independent of the parameter $\theta$, let $f(\theta)=\frac{m^-}{m^+v^2}$, therefore,
\begin{align*}
	sign\{ \frac{\partial \Pi^*(\bm{G};\beta,\delta,\theta)}{\partial \theta}\} = sign\left\{ \frac{\partial f(\theta)}{\partial \theta} \right\}.
\end{align*}
We have
\begin{align*}
	\begin{aligned}
		lnf(\theta)=lnm^- - lnm^+ - 2lnv
	\end{aligned},
\end{align*}
so 
\begin{align}
	\begin{aligned}
		\frac{\partial lnf(\theta)}{\partial \theta} = \delta d \left(  \frac{1}{m^-} + \frac{1}{m^+} - \frac{2}{v} \right)
	\end{aligned}.
	\label{18}
\end{align}
It is obvious that $m^- + m^+=2s$ and $v=s+m^-$, because $\frac{\partial lnf(\theta)}{\partial \theta} \cdot f(\theta) = \frac{\partial f(\theta)}{\partial \theta}$ and we know $f(\theta)>0$. Hence, we conclude that
\begin{align*}
	\begin{aligned}
		sign\left\{ \frac{\partial f(\theta)}{\partial \theta} \right\}&=sign \left\{ \frac{1}{m^-} + \frac{1}{m^+} - \frac{2}{v} \right\}
		\\
		&=sign \left\{ \frac{2s}{m^-m^+}-\frac{2}{m^- + s} \right\}
		\\
		&=sign \left\{2 \frac{s^2-m^-s +(m^-)^2}{m^-m^+(m^- + s)} \right\}.
	\end{aligned}
\end{align*}
And we have
\begin{align*}
	\begin{aligned}
		s^2-m^- s + (m^-)^2 = \left(s-\frac{1}{2}m^-\right)^2 + \frac{3}{4}(m^-)^2>0
	\end{aligned},
\end{align*}
therefore, 
\begin{align*}
	\frac{\partial f(\theta)}{\partial \theta} > 0 \iff \frac{\partial \Pi^*}{\partial \theta} > 0
\end{align*}

The proof of Proposition \ref{prop:cs_profit_comparative} is complete.
\end{proof}

\subsection{Proof of Proposition \ref{prop:total_welfare}}

\begin{proof}
Under Assumption \ref{ass:stability} and \ref{ass:full_symmetry}, if the network is regular, then we have
\begin{align*}
	\begin{aligned}
		\frac{\partial CS^*(\bm{G};\beta,\delta,\theta)}{\partial \theta} &= n(1+\beta)(a-c)^2s^2\left[\frac{2\delta d m^+v^2-2\delta d (m^+)^2v}{(m^+)^4v^4}\right]
		\\
		&=2\delta d n(1+\beta)(a-c)^2s^2\frac{v-m^+}{(m^+)^3 v^3}
	\end{aligned},
\end{align*}
and according to \eqref{18}
\begin{align*}
	\begin{aligned}
		\frac{\partial \Pi^*(\bm{G};\beta,\delta,\theta)}{\partial \theta} &=n (a-c)^2s\frac{m^-}{m^+v^2} \cdot \left[2\delta d \frac{s^2-m^-s+(m^-)^2}{m^-m^+v}\right]
		\\
		&=2\delta d n(a-c)^2s\frac{s^2-m^-s+(m^-)^2}{(m^+)^3v^3}
	\end{aligned}.
\end{align*}

And we have $TW^*= CS^* + 2\Pi^*$, so there is 
\begin{align*}
	\begin{aligned}
		\frac{\partial TW^*(\bm{G};\beta,\delta,\theta)}{\partial \theta} &= \frac{\partial CS^*(\bm{G};\beta,\delta,\theta)}{\partial \theta} + 2\frac{\partial \Pi^*(\bm{G};\beta,\delta,\theta)}{\partial \theta}
		\\
		&=\frac{2n\delta d (a-c)^2s}{(m^+)^3v^3}\left[s(v-m^+)(1+\beta) + 2m^+[s^2-m^-s+(m^-)^2]\right]
	\end{aligned}.
\end{align*}

Hence, 
\begin{align*}
	\begin{aligned}
		sign\left\{ \frac{\partial TW^*(\bm{G};\beta,\delta,\theta)}{\partial \theta}\right\} = sign \left\{s (v-m^+)(1+\beta) + 2m^+[s^2-m^-s+(m^-)^2] \right\}
	\end{aligned}
\end{align*}.
\end{proof}

\section{Proofs for Designing Interoperability}
\subsection{Proof of Proposition \ref{prop:price_theta_design}}

Under Assumption \ref{ass:full_symmetry}, let \(r=a-c\). From the equilibrium price formula,
\[
	\bm p^*-\bm c
	=
	\left[(1-\beta)\bm I-\delta(1-\theta)\bm G\right]
	\left[(2-\beta)\bm I-\delta(2-\theta)\bm G\right]^{-1}
	r\bm 1 .
\]
Let
\[
	\bm K_\theta=(1-\beta)\bm I-\delta(1-\theta)\bm G,
	\qquad
	\bm V_\theta=(2-\beta)\bm I-\delta(2-\theta)\bm G.
\]
Then
\[
	\frac{\partial \bm K_\theta}{\partial\theta}
	=
	\delta\bm G,
	\qquad
	\frac{\partial \bm V_\theta}{\partial\theta}
	=
	\delta\bm G.
\]
Since \(\bm K_\theta\), \(\bm V_\theta\), and \(\bm G\) are all affine functions of \(\bm G\), they commute. Therefore,
\[
\frac{\partial \bm p^*}{\partial\theta}
=
r\left[
	\delta\bm G\bm V_\theta^{-1}
	-
	\bm K_\theta\bm V_\theta^{-1}
	\delta\bm G\bm V_\theta^{-1}
\right]\bm 1 .
\]
Using commutativity,
\[
\frac{\partial \bm p^*}{\partial\theta}
=
r\delta\bm G
\bm V_\theta^{-1}
\left[
	\bm I-\bm K_\theta\bm V_\theta^{-1}
\right]\bm 1.
\]
Since
\[
	\bm V_\theta-\bm K_\theta
	=
	\bm I-\delta\bm G,
\]
we obtain
\[
\frac{\partial \bm p^*}{\partial\theta}
=
r\delta
\bm G(\bm I-\delta\bm G)\bm V_\theta^{-2}\bm 1,
\]
which proves \eqref{eq:price_theta_derivative_general}.

If \(\bm G\) is \(d\)-regular, then
\[
	\bm G\bm 1=d\bm 1
\]
and
\[
	\bm V_\theta\bm 1=
	\left[(2-\beta)-\delta d(2-\theta)\right]\bm 1.
\]
Substituting into the derivative formula yields
\[
\frac{\partial p^*}{\partial\theta}
=
\frac{
	\delta d(1-\delta d)
}{
	\left[(2-\beta)-\delta d(2-\theta)\right]^2
}
(a-c)>0,
\]
where positivity follows from Assumption \ref{ass:stability}.

Finally, for small \(\delta\),
\[
	\bm V_\theta^{-2}
	=
	\frac{1}{(2-\beta)^2}\bm I
	+
	O(\delta).
\]
Substituting this into \eqref{eq:price_theta_derivative_general} gives
\[
\frac{\partial p_i^*}{\partial\theta}
=
\frac{\delta(a-c)}{(2-\beta)^2}d_i
+
O(\delta^2).
\]
The proof is complete.

\section{Proofs for Network Formation}
\label{app:network_formation}

\subsection{Proof of Proposition \ref{prop:consumer_network}}

\begin{proof}
Suppose Assumptions \ref{ass:stability} and \ref{ass:platform_symmetry} hold. According to \eqref{16}
\begin{align*}
	\begin{aligned}
		\bm{x}^*=(\bm{I}-\delta \bm{G})\bm{D}^{-1}\bm{V}^{-1}(\bm{a}-\bm{c}),
	\end{aligned}
\end{align*}
we have $(\bm{I}-\delta \bm{G})=\frac{1}{3}(\bm{D}+\bm{V})$, hence
\begin{align*}
	\begin{aligned}
		\bm{x}^*&=\frac{1}{3}(\bm{V}^{-1}+\bm{D}^{-1})(\bm{a}-\bm{c})
		\\
		&=\frac{1}{3}\left[ \frac{1}{2-\beta}[\bm{I}-\frac{\delta(2-\theta)}{2-\beta}\bm{G}]^{-1} + \frac{1}{1+\beta}[\bm{I}-\frac{\delta(1+\theta)}{1+\beta}\bm{G}]^{-1} \right](\bm{a}-\bm{c})
		\\
		&=\frac{1}{3}\left[ \frac{1}{2-\beta}\mathop{\Sigma}_{k=0}^{+\infty}\left( \frac{\delta(2-\theta)}{2-\beta}\bm{G} \right)^k + \frac{1}{1+\beta}\mathop{\Sigma}_{k=0}^{+ \infty}\left( \frac{\delta (1+\theta)}{1+\beta}\bm{G} \right)^k \right](\bm{a}-\bm{c})
	\end{aligned}.
\end{align*}
It is obvious that, if $\bm{G}' \succeq \bm{G}$, then $\bm{x}^*(\bm{G}';\beta,\delta,\theta) \geq \bm{x}^*(\bm{G};\beta,\delta,\theta)$. From Theorem \ref{prop:welfare_objects}, we know $CS^*=(1+\beta)(\bm{x}^*)^T(\bm{x}^*)$, therefore $CS^*(\bm{G}';\beta,\delta,\theta) \geq CS^*(\bm{G};\beta,\delta,\theta)$.

The proof of Proposition \ref{prop:consumer_network} is complete.
\end{proof}

\subsection{Proof of Proposition \ref{prop:platform_network}}

\begin{proof}
Suppose Assumptions \ref{ass:stability} and \ref{ass:platform_symmetry} hold. According to \eqref{proof_of_EP_4} and \eqref{16}, we have
\begin{align*}
	\begin{aligned}
		\bm{x}^*=(\bm{I}-\delta \bm{G})\bm{D}^{-1}\bm{V}^{-1}(\bm{a}-\bm{c});
		\\
		\bm{p}^*-\bm{c}\bm{K}\bm{V}^{-1}(\bm{a}-\bm{c})
	\end{aligned}.
\end{align*}
Therefore, 
\begin{align*}
	\begin{aligned}
		\bm{x}^*&=\frac{1}{(2-\beta)(1+\beta)}(\bm{I}-\delta \bm{G})\left[ \bm{I} - \frac{\delta(2-\theta)}{2-\beta}\bm{G} \right]^{-1}\left[\bm{I} - \frac{\delta(1+\theta)}{1+\beta}\bm{G} \right]^{-1}(\bm{a}-\bm{c})
		\\
		&=\frac{1}{(2-\beta)(1+\beta)}(\bm{I}-\delta \bm{G})\left[ \bm{I} + \frac{\delta(2-\theta)}{2-\beta}\bm{G} + \mathcal{O}(\delta^2)\right]\left[ \bm{I} + \frac{\delta(1+\theta)}{1+\beta}\bm{G} + \mathcal{O}(\delta^2) \right](\bm{a}-\bm{c})
		\\
		&=\frac{1}{(2-\beta)(1+\beta)}\left[ \bm{I} + \frac{\delta (\beta -\theta)}{2-\beta}\bm{G} + \mathcal{O}(\delta^2) \right]\left[ \bm{I} + \frac{\delta(1+\theta)}{1+\beta}\bm{G} + \mathcal{O}(\delta^2) \right](\bm{a}-\bm{c})
		\\
		&=\frac{1}{(2-\beta)(1+\beta)}\left[ \bm{I}+\frac{\delta(\beta-\theta)}{2-\beta}\bm{G} +\frac{\delta(1+\theta)}{1+\beta}\bm{G} + \mathcal{O}(\delta^2) \right](\bm{a}-\bm{c})
		\\
		&=\frac{1}{(2-\beta)(1+\beta)}\left[\bm{I} + \frac{\delta(2+\beta^2 +\theta -2\theta\beta)}{(1+\beta)(2-\beta)}\bm{G}  \right](\bm{a}-\bm{c}) + \mathcal{O}(\delta^2)
	\end{aligned},
\end{align*}
and
\begin{align*}
	\begin{aligned}
		\bm{p}^*-\bm{c} &=\bm{K}\bm{V}^{-1}(\bm{a}-\bm{c})
		\\
		&=\frac{1}{2-\beta}[(1-\beta)\bm{I} - \delta(1-\theta)\bm{G}] \left[ \bm{I} + \frac{\delta (2-\theta)}{2-\beta}\bm{G} + \mathcal{O}(\delta^2)\right](\bm{a}-\bm{c})
		\\
		&=\frac{1}{2-\beta}\left[ (1-\beta)\bm{I} + \frac{\delta(2-\theta)(1-\beta)}{2-\beta}\bm{G} - \delta(1-\theta)\bm{G} + \mathcal{O}(\delta^2) \right](\bm{a}-\bm{c})
		\\
		&=\frac{1}{2-\beta}\left[ (1-\beta)\bm{I} + \frac{\delta(\theta -\beta)}{(2-\beta)}\bm{G} \right](\bm{a}-\bm{c}) + \mathcal{O}(\delta^2)
	\end{aligned}.
\end{align*}
$\bm{a}-\bm{c}=\bm{r}$, then
\begin{align*}
	\begin{aligned}
		\Pi^*(\bm{G};\beta,\delta,\theta)&=<\bm{p}^*-\bm{c}, \bm{x}^*>
		\\
		&=\left[ \frac{1}{(2-\beta)(1+\beta)}\bm{r}^T + \frac{\delta(2+\beta^2 + \theta -2\theta\beta)}{(1+\beta)^2(2-\beta)^2}\bm{r}^T\bm{G} \right]\left[ \frac{1-\beta}{2-\beta}\bm{r} + \frac{\delta(\theta - \beta)}{(2-\beta)^2}\bm{G}\bm{r} \right] + \mathcal{O}(\delta^2)
		\\
		&=\frac{1-\beta}{(1+\beta)(2-\beta)^2}\bm{r}^T\bm{r} + \frac{\delta}{(1+\beta)^2(2-\beta)^3}\mathcal{Q}(\beta,\theta)\bm{r}^T\bm{G}\bm{r} + \mathcal{O}(\delta^2)
	\end{aligned},
\end{align*}
where
\begin{align*}
	\begin{aligned}
		\mathcal{Q}(\beta,\theta)&=(\theta -\beta)(1+\beta) + (1-\beta)(2+\beta^2+\theta -2\theta \beta)
		\\
		&=2-3\beta-\beta^3+2\theta(1-\beta+\beta^2)
	\end{aligned}.
\end{align*}

The sign of $\mathcal{Q}(\beta, \theta)$ determines the direction of the network effect to platforms' profits. Let $\bar{\theta}(\beta) = \frac{\beta^3 + 3\beta -2}{2(1-\beta +\beta^2)}$, and when $\theta>\bar{\theta}(\beta)$, $\mathcal{Q}(\beta,\theta)>0$, conversly, $\mathcal{Q}(\beta,\theta)<0$. If $\mathcal{Q}(\beta,\theta)=o$, then the substitution effect and the interoperation effect cancel each other out.

The proof of Proposition \ref{prop:platform_network} is complete.
\end{proof}

\section{Uniform versus Discriminatory Pricing}
\label{app:uniform_pricing}

\subsection{Proof of Theorem \ref{thm:uniform_pricing}}

\begin{proof}
	(i) If Assumptions \ref{ass:stability} and \ref{ass:platform_symmetry} hold, if both platforms charge uniform prices, then the equilibrium profit of a single platform A should be $\Pi^A=<\bm{x}^A_u, p^A_u\bm{1}-\bm{c}>$. Therefore, the FOC is
\begin{align*}
	\begin{aligned}
		\frac{\partial \Pi^A}{\partial p^A_u}=<\bm{x}^A_u,\bm{1}> + <\frac{\partial \bm{x}^A_u}{\partial p^A_u}, p^A_u\bm{1}-\bm{c}>=0
	\end{aligned},
\end{align*}
and from \eqref{eq:consumption_equilibrium}, we have $\frac{\partial \bm{x}^A_u}{\partial p^A_u}=-\frac{\bm{M}^++\bm{M}^-}{2}\bm{1}$. Then
\begin{align*}
	\begin{aligned}
		<\bm{x}^A_u, \bm{1}>=<\frac{\bm{M}^+ +\bm{M}^-}{2}\bm{1}, p^A_u\bm{1}-\bm{c}>
	\end{aligned}.
\end{align*}

Because Assumption \ref{ass:platform_symmetry} holds, there is $\bm{x}^A_u=\bm{x}^B_u=\bm{x}^u$ and $p^A_u=p^B_u=p^u$, hence
\begin{align*}
	\begin{aligned}
		<\bm{M}^+(\bm{a}-p^u\bm{1}),\bm{1}> = <\frac{\bm{M}^+ +\bm{M}^-}{2}\bm{1}, p^u\bm{1}-\bm{c}>
	\end{aligned},
\end{align*}
and
\begin{align*}
	\begin{aligned}
		\bm{1}^T\bm{M}^+\bm{a}-p^u\bm{1}^T\bm{M}^+\bm{1}&=\frac{1}{2}p^u\bm{1}^T(\bm{M}^++\bm{M}^-)\bm{1}-\frac{1}{2}\bm{1}^T(\bm{M}^+ +\bm{M}^-)\bm{c};
		\\
		2\bm{1}^T\bm{M}^+\bm{a}+\bm{1}^T&\bm{M}^+\bm{c}+\bm{1}^T\bm{M}^-\bm{c}=p^u\bm{1}(3\bm{M}^+ + \bm{M}^-)\bm{1},
	\end{aligned}
\end{align*}
therefore
\begin{align*}
	\begin{aligned}
		p^u=\frac{<\bm{1},\bm{M}^+(2\bm{a}+\bm{c})> + <\bm{1},\bm{M}^-\bm{c}>}{<\bm{1},(3\bm{M}^++\bm{M}^-)\bm{1}>}
	\end{aligned}
\end{align*}

If Assumption \ref{ass:full_symmetry} also holds, we have
\begin{align*}
	\begin{aligned}
		p^u=\frac{2a\bm{1}^T\bm{M}^+\bm{1}+c\bm{1}^T\bm{M}^+\bm{1} + c\bm{1}^T\bm{M}^-\bm{1}}{3\bm{1}^T\bm{M}^+\bm{1} + \bm{1}^T\bm{M}^-\bm{1}},
	\end{aligned}
\end{align*}
then
\begin{align*}
	\begin{aligned}
		p^u=\frac{2a\bm{1}^T\bm{M}^+\bm{1}+c(3\bm{1}^T\bm{M}^+\bm{1}+\bm{1}^T\bm{M}^-\bm{1})-2c\bm{1}^T\bm{M}^+\bm{1}}{3\bm{1}^T\bm{M}^+\bm{1}+\bm{1}^T\bm{M}^-\bm{1}}.
	\end{aligned}
\end{align*}
Finally, we have
\begin{align*}
	p^u=c+2(a-c)\frac{<\bm{1},\bm{M}^+\bm{1}>}{<\bm{1},(3\bm{M}^+ +\bm{M}^-)\bm{1}>}.
\end{align*}
The proof of Theorem \ref{thm:uniform_pricing} is complete.
\end{proof}

\subsection{Proof of Corollary \ref{cor:uniform_price_approx}}

\begin{proof}
When Assumptions \ref{ass:stability} and \ref{ass:platform_symmetry} hold, and $\delta$ is small enough, we have
\begin{align*}
	\begin{aligned}
		<\bm{1},\bm{M}^+\bm{1}>&=\frac{1}{1+\beta}\bm{1}^T(\bm{I}+\frac{\delta(1+\theta)}{1+\beta}\bm{G} +\mathcal{O}(\delta^2))\bm{1}
		\\
		&=\frac{1}{1+\beta}(n+\frac{\delta(1+\theta)}{1+\beta}\mathop{\Sigma}_{i=1}^{n}d_i) + \mathcal{O}(\delta^2)
	\end{aligned}.
\end{align*}
Similarly,
\begin{align*}
	\begin{aligned}
		<\bm{1},\bm{M}^-\bm{1}>=\frac{1}{1-\beta}(n + \frac{\delta(1-\theta)}{1-\beta}\mathop{\Sigma}_{i=1}^{n}d_i) + \mathcal{O}(\delta^2)
	\end{aligned}.
\end{align*}

Let 
\begin{align*}
	\begin{aligned}
		R_1=\frac{3}{1+\beta} + \frac{1}{1-\beta}=\frac{2(2-\beta)}{(1-\beta^2)} ~\text{ and }~ R_2=\frac{3(1+\theta)}{(1+\beta)^2} + \frac{1-\theta}{(1-\beta)^2}
	\end{aligned},
\end{align*}
then we have
\begin{align*}
	\begin{aligned}
		p^u&=c+ (a-c)\left[\frac{\frac{2n}{1+\beta} + 2\delta \mathop{\Sigma}_{i=1}^{n}d_i\frac{1+\theta}{(1+\beta)^2}}{nR_1 + \delta\mathop{\Sigma}_{i=1}^{n}d_i R_2} + \mathcal{O}(\delta^2)\right]
		\\
		&=c+ (a-c)\left[\frac{\frac{2}{(1+\beta)R_1} + 2\delta \bar{d}\frac{(1+\theta)}{R_1(1+\beta)^2}}{1+\delta \bar{d}\frac{R_2}{R_1}} + \mathcal{O}(\delta^2)\right]
	\end{aligned}.
\end{align*}

Because $\delta$ is small enough, therefore $(1+\delta\bar{d}\frac{R_2}{R_1})^{-1} = 1 - \delta\bar{d}\frac{R_2}{R_1} + \mathcal{O}(\delta^2)$, then
\begin{align*}
	\begin{aligned}
		p^u&=c+ (a-c)\left[\left( \frac{2}{(1+\beta)R_1} + \frac{2\delta\bar{d}(1+\theta)}{R_1(1+\beta)^2} \right)\left( 1- \delta\bar{d}\frac{R_2}{R_1} +\mathcal{O}(\delta^2) \right) + \mathcal{O}(\delta^2)\right]
		\\
		&=c+ (a-c)\left[ \frac{2}{(1+\beta)R_1} + \frac{2\delta\bar{d}}{R_1}\left( \frac{(1+\theta)}{(1+\beta)^2}-\frac{R_2}{(1+\beta)R_1} \right) \right] + \mathcal{O}(\delta^2)
		\\
		&=c+ (a-c)\left[ \frac{1-\beta}{2-\beta} + \frac{2\delta\bar{d}}{R_1}\left( \frac{(1+\theta)}{(1+\beta)^2}-\frac{1-\beta}{2(2-\beta)}\left( \frac{3(1+\theta)}{(1+\beta)^2} + \frac{1-\theta}{(1-\beta)^2} \right) \right) \right] + \mathcal{O}(\delta^2)
		\\
		&=c+ (a-c)\left[ \frac{1-\beta}{2-\beta} + \frac{2\delta \bar{d}}{R_1}\left( \frac{(1+\theta)}{(1+\beta)^2}\left( 1-\frac{3(1-\beta)}{2(2-\beta)} \right)-\frac{1-\theta}{2(2-\beta)(1-\beta)} \right) \right] + \mathcal{O}(\delta^2)
		\\
		&=c+ (a-c)\left[ \frac{1-\beta}{2-\beta} +\frac{2\delta \bar{d}}{R_1}\left( \frac{1+\theta}{(1+\beta)^2} \cdot \frac{1+\beta}{2(2-\beta)} -\frac{1-\theta}{2(2-\beta)(1-\beta)} \right) \right] + \mathcal{O}(\delta^2)
		\\
		&=c+ (a-c)\left( \frac{1-\beta}{2-\beta} + 2\delta\bar{d}\frac{\theta-\beta}{R_1(2-\beta)(1-\beta^2)} \right) + \mathcal{O}(\delta^2)
		\\
		&=c+ (a-c)\left( \frac{1-\beta}{2-\beta} + \delta\bar{d}\frac{\theta-\beta}{(2-\beta)^2} \right) + \mathcal{O}(\delta^2)
	\end{aligned}
\end{align*}
The proof of Theorem \ref{cor:uniform_price_approx} is complete.
\end{proof}





\subsection{Proof of Proposition \ref{prop:aggregate_uniform_discriminatory}}

\begin{proof}
From Corollary \ref{cor:uniform_price_approx}, we have 
\begin{align*}
	\begin{aligned}
		p^u-c=(a-c)\left( \frac{1-\beta}{2-\beta} + \delta\bar{d}\frac{\theta-\beta}{(2-\beta)^2} \right) + \mathcal{O}(\delta^2)
	\end{aligned},
\end{align*}
therefore,
\begin{align}
	\begin{aligned}
		\bm{x}^u&=(a-p^u)\bm{M}^+\bm{1}
		\\
		&=\left[ a-c-(a-c)\left( \frac{1-\beta}{2-\beta} + \delta\frac{(\theta-\beta)}{(2-\beta)^2} \right)+ \mathcal{O}(\delta^2) \right]\bm{M}^+\bm{1} 
		\\
		&=(a-c)\left( 1- \frac{1-\beta}{2-\beta} - \delta\frac{\theta-\beta}{(2-\beta)^2} + \mathcal{O}(\delta^2)\right)\bm{M}^+\bm{1}
		\\
		&=\frac{a-c}{(2-\beta)(1+\beta)}\left( 1-\delta\frac{\theta-\beta}{2-\beta}\bar{d} +\mathcal{O}(\delta^2)\right)\left( \bm{I} + \delta \frac{1+\theta}{1+\beta}\bm{G} + \mathcal{O}(\delta^2) \right)\bm{1}
		\\
		&=\frac{a-c}{(2-\beta)(1+\beta)}\left( \bm{1} + \delta\frac{1+\theta}{1+\beta}\bm{G}\bm{1} - \delta\frac{\theta-\beta}{2-\beta}\bar{d}\bm{1} \right) + \mathcal{O}(\delta^2)
	\end{aligned}
\end{align}

The equilibrium price and consumption volume under discriminatory pricing, as obtained above\footnote{In the Proof of Proposition \ref{prop:platform_network}}, are:
\begin{align*}
	\begin{aligned}
		&\bm{p}^* - c\bm{1}= \frac{a-c}{2-\beta}\left((1-\beta)\bm{1} + \delta\frac{\theta-\beta}{2-\beta}\bm{d}\right)+ \mathcal{O}(\delta^2);
		\\
		&\bm{x}^*=\frac{a-c}{(2-\beta)(1+\beta)}\left( \bm{1} + \delta\frac{2+\beta^2+\theta -2\theta\beta}{(2-\beta)(1+\beta)}\bm{1} \right) + \mathcal{O}(\delta^2)
		\\
		&~~~=\frac{a-c}{(2-\beta)(1+\beta)}\left( \bm{1} + \delta\frac{1+\theta}{1+\beta}\bm{G}\bm{1} - \delta\frac{\theta-\beta}{2-\beta}\bm{G}\bm{1}\right) +\mathcal{O}(\delta^2).
	\end{aligned}
\end{align*}

Let $\varepsilon=\frac{1+\theta}{1+\beta }$; $\sigma=\frac{\theta-\beta}{2-\beta}$ and  $\tau=\frac{a-c}{(2-\beta)(1+\beta)}$, then

\begin{align*}
	\begin{aligned}
		&\bm{x}^* = \tau\left( \bm{1} + \delta(\varepsilon-\sigma)\bm{d}\right) + \mathcal{O}(\delta^2);
		\\
		&\bm{x}^u = \tau\left( \bm{1} + \delta\varepsilon\bm{d} - \delta \sigma\bar{d}\bm{1} \right) + \mathcal{O}(\delta^2);
		\\
		&\bm{p}^* -c\bm{1} = \tau \left[(1-\beta^2)\bm{1} + \delta \sigma(1+\beta)\bm{d}\right] + \mathcal{O}(\delta^2);
		\\
		&p^u\bm{1} -c\bm{1} = \tau \left[ (1-\beta^2)\bm{1} + \delta \sigma(1+\beta)\bar{d}\bm{1}  \right] + \mathcal{O}(\delta^2).
	\end{aligned}
\end{align*}

According to the previous text, we have $CS=(1+\beta)\bm{x}^T\bm{x}$. Using the first-order approximation of equilibrium consumption, the corresponding approximation of consumer surplus is
\begin{align*}
	\begin{aligned}
		CS^* &= (1+\beta)\tau^2 (\bm{1} + \delta(\varepsilon-\sigma)\bm{d} + \mathcal{O}(\delta^2))^T(\bm{1} + \delta(\varepsilon-\sigma)\bm{d} + \mathcal{O}(\delta^2));
		\\
		\widetilde{CS}^*& = \frac{(a-c)^2}{(2-\beta)^2(1+\beta)}\left[ n + 2\delta(\varepsilon - \sigma)\mathop{\Sigma}_{i=1}^{n}d_i + \delta^2(\varepsilon - \sigma)^2\mathop{\Sigma}_{i=1}^{n}d_i^2 \right],
	\end{aligned}
\end{align*}
and
\begin{align*}
	\begin{aligned}
		CS^u&=(1+\beta)\tau^2(\bm{1} + \delta \varepsilon\bm{d} -\delta\sigma\bar{d}\bm{1} + \mathcal{O}(\delta^2))^T(\bm{1} + \delta \varepsilon\bm{d} -\delta\sigma\bar{d}\bm{1} + \mathcal{O}(\delta^2));
		\\
		\widetilde{CS}^u& = \frac{(a-c)^2}{(2-\beta)^2(1+\beta)} \left[ n + 2\delta(\varepsilon-\sigma)\mathop{\Sigma}_{i=1}^{n}d_i + \delta^2\varepsilon^2\mathop{\Sigma}_{i=1}^{n}d_i^2 + \delta^2\sigma^2n\bar{d}^2 - 2\delta^2\varepsilon\sigma  n \bar{d}^2 \right] .
	\end{aligned}
\end{align*}

Therefore, we can obtain
\begin{align*}
	\begin{aligned}
		\widetilde{\Delta CS} &= \widetilde{CS}^u-\widetilde{CS}^*
		\\
		&= \frac{(a-c)^2}{(2-\beta)^2(1+\beta)}\left[ \delta^2\varepsilon^2\mathop{\Sigma}_{i=1}^{n}d_i^2 + \delta^2\sigma^2n\bar{d}^2 - 2\delta^2 \varepsilon\sigma n \bar{d}^2 -  \delta^2(\varepsilon - \sigma)^2\mathop{\Sigma}_{i=1}^{n}d_i^2\right]
		\\
		& = \frac{\delta^2(a-c)^2}{(2-\beta)^2(1+\beta)}\left[ \sigma(\sigma - 2\varepsilon)n\bar{d}^2 - \sigma(\sigma-2\varepsilon)\mathop{\Sigma}_{i=1}^{n}d_i^2 \right]
		\\
		& =\frac{\delta^2(a-c)^2}{(2-\beta)^2(1+\beta)}\sigma(2\varepsilon-\sigma)(\mathop{\Sigma}_{i=1}^{n}d_i^2-n\bar{d}^2)
	\end{aligned}.
\end{align*}

It is obvious that $\mathop{\Sigma}_{i=1}^{n}d_i^2-n\bar{d}^2 = \mathop{\Sigma}_{i=1}^{n}(d_i-\bar{d})^2 > 0$, and $2\varepsilon-\sigma=\frac{3\theta(1-\beta) + [4-\beta(1-\beta)]}{(1+\beta)(2-\beta)} > 0$. Hence, the sign of $\widetilde{\Delta CS}$ depends on $\sigma=\frac{\theta-\beta}{2-\beta}$.

Finally, we have
\begin{align*}
	\begin{aligned}
		\theta > \beta &\iff  \widetilde{CS}^u>\widetilde{CS}^*  ;
		\\
		\theta = \beta &\iff  \widetilde{CS}^u=\widetilde{CS}^* ;
		\\
		\theta < \beta &\iff \widetilde{CS}^u<\widetilde{CS}^* .
	\end{aligned}
\end{align*}

And according to previous text, the profit of a single platform is $\Pi=< \bm{p}-c\bm{1}, \bm{x}>$, using the first-order approximation of equilibrium price and consumption, the corresponding approximation is
\begin{align*}
	\begin{aligned}
		\Pi^* &=<\bm{p}^*-c\bm{1}, \bm{x}^*>
		\\
		&=\tau^2\left( (1-\beta^2)\bm{1} + \delta\sigma(1+\beta)\bm{d} + \mathcal{O}(\delta^2)\right)^T(\bm{1}+\delta(\varepsilon-\sigma)\bm{d} + \mathcal{O}(\delta^2));
		\\
		\widetilde{\Pi}^*&=\tau^2\left[  (1-\beta^2)n + \delta\sigma(1+\beta)n\bar{d} + \delta(\varepsilon-\sigma)(1-\beta^2)n\bar{d} + \delta^2\sigma(\varepsilon-\sigma)(1+\beta)\mathop{\Sigma}_{i=1}^n d_i^2\right],
	\end{aligned}
\end{align*}
and 
\begin{align*}
	\begin{aligned}
		\Pi^u &= (p^u-c)(\bm{x}^u)^T\bm{1}
		\\
		&=\tau^2((1-\beta^2)+\delta\sigma(1+\beta)\bar{d}+\mathcal{O}(\delta^2))(\bm{1} + \delta\varepsilon\bm{d} - \delta\sigma\bar{d}\bm{1} + \mathcal{O}(\delta^2) )^T\bm{1}
		\\
		&=\tau^2((1-\beta^2)+\delta\sigma(1+\beta)\bar{d}+\mathcal{O}(\delta^2))(n + \delta\varepsilon n \bar{d} - \delta\sigma n\bar{d} + \mathcal{O}(\delta^2));
		\\
		\widetilde{\Pi}^u&= \tau^2\left[ (1-\beta^2)n + \delta\sigma(1+\beta)n\bar{d} + \delta(\varepsilon-\sigma)(1-\beta^2)n\bar{d}+\delta^2\sigma(\varepsilon-\sigma)(1+\beta)n\bar{d}^2 \right],
	\end{aligned}
\end{align*}
hence,
\begin{align*}
	\begin{aligned}
		\widetilde{\Delta \Pi}&=\widetilde{\Pi}^u-\widetilde{\Pi}^*
		\\
		& = \tau^2\delta^2(1+\beta)\sigma(\sigma-\varepsilon)(\mathop{\Sigma}_{i=1}^n d_i^2 - n\bar{d}^2)
	\end{aligned}.
\end{align*}

It is obvious that $(\sigma-\varepsilon) =\frac{-[2(1-\theta\beta) + \theta +\beta]}{(2-\beta)(1+\beta)} < 0$, therefore, the sign of $\Delta \Pi$ depends on the sign of $\sigma$, then 
\begin{align*}
	\begin{aligned}
		&\theta > \beta \iff \widetilde{\Pi}^u<\widetilde{\Pi}^**;
		\\
		&\theta = \beta \iff \widetilde{\Pi}^u=\widetilde{\Pi}^*;
		\\
		&\theta < \beta \iff \widetilde{\Pi}^u>\widetilde{\Pi}^*.
	\end{aligned}
\end{align*}
The proof of Theorem \ref{prop:aggregate_uniform_discriminatory} is complete.
\end{proof}

\bibliographystyle{ecta}
\bibliography{reference}

\end{document}